\documentclass[a4paper,UKenglish,cleveref, autoref, thm-restate]{lipics-v2021} %version used for ESA submission!
\usepackage[utf8]{inputenc}
\usepackage{amsmath,amsfonts,amsthm,amssymb}
\usepackage{framed}
\usepackage[dvipsnames]{xcolor}
\usepackage{hyperref}

\setul{0.5ex}{0.3ex}

\usepackage[disableredefinitions]{complexity}
\usepackage{xspace}
\usepackage{booktabs}

\usepackage{algorithm}
\usepackage{algpseudocodex}

\usepackage{bold-extra}

\usepackage{thmtools}

\usepackage{newunicodechar}
\newunicodechar{≤}{\ensuremath{\leq}}
\usepackage[appendix=inline]{apxproof}
\makeatletter
 \@ifpackagewith{apxproof}{appendix=inline}{}{}
 \ifthenelse{\equal{\axp@appendix}{strip}}{}{}
 \NewEnviron{sketch}
 { \ifthenelse{\equal{\axp@appendix}{inline}}{}{
     \iftoggle{sketch}{
     \begin{proofsketch}
       \BODY
     \end{proofsketch}
       }{}
   }}
 \makeatother

\EventEditors{George B. Mertzios and Andr\'{e}a W. Richa}
\EventNoEds{2}
\EventLongTitle{5th Symposium on Algorithmic Foundations of Dynamic Networks (SAND 2026)}
\EventShortTitle{SAND 2026}
\EventAcronym{SAND}
\EventYear{2026}
\EventDate{July 1--3, 2026}
\EventLocation{Le Havre, France}
\EventLogo{}
\SeriesVolume{373}
\ArticleNo{1}
 
\usepackage{tikz}

\DeclareMathOperator{\reached}{\textit{reached}}
\DeclareMathOperator{\unreached}{\textit{unreached}}

\DeclareMathOperator{\twostep}{2-step}
\DeclareMathOperator{\visited}{\textit{visited}}
\DeclareMathOperator{\unvisited}{\textit{unvisited}}
\DeclareMathOperator{\current}{\textit{current}}
\DeclareMathOperator{\covered}{\textit{covered}}
\DeclareMathOperator{\uncovered}{\textit{uncovered}}
\DeclareMathOperator{\dominating}{\textit{dominating}}

\newcounter{ctheorem}

\newtheorem{defn}[ctheorem]{Definition}

\newtheoremrep{theorem}[ctheorem]{Theorem}
\newtheoremrep{lemma}[ctheorem]{Lemma}
\newtheoremrep{corollary}[ctheorem]{Corollary}
\newtheoremrep{observation}[ctheorem]{Observation}
\Crefname{observation}{Observation}{Observations}
\newcommand{\decisionproblem}[3]{%
\vspace{0,1cm} \noindent \fbox{%
\begin{minipage}{0.96\textwidth}
 \begin{tabular*}{\textwidth}{@{\extracolsep{\fill}}lr} \textsc{#1} & \\ \end{tabular*}%
  \vspace{1.2mm}%
\par%
{\bf{Input:}}  #2\\%
{\bf{Output:}} #3%
\end{minipage}} \vspace{0,3cm}%
}

   \hideLIPIcs
  \nolinenumbers

\title{Families of tractable problems with respect to vertex-interval-membership width and its generalisations } %TODO Please add
\author{Jessica Enright}{School of Computing Science, University of Glasgow, UK}{jessica.enright@glasgow.ac.uk}{0000-0002-0266-3292}{Supported by EPSRC grant EP/T004878/1.}
 \author{Samuel D. Hand}{School of Computing Science, University of Glasgow, UK}{s.hand.1@research.gla.ac.uk}{0000-0001-8021-249X}{Supported by an EPSRC doctoral training account.}
 \author{Laura Larios-Jones}{School of Computing Science, University of Glasgow, UK}{Laura.Larios-Jones@glasgow.ac.uk}{0000-0003-3322-0176}{}
 \author{Kitty Meeks}{School of Computing Science, University of Glasgow, UK}{kitty.meeks@glasgow.ac.uk}{0000-0001-5299-3073}{Supported by EPSRC grants  EP/T004878/1 and EP/V032305/1.}

\authorrunning{J. Enright, S.\,D. Hand, L. Larios-Jones, and K. Meeks}
\titlerunning{Tractable problems with respect to VIM width and generalisations}

\Copyright{Jessica Enright, Samuel D. Hand, Laura Larios-Jones, Kitty Meeks} %TODO mandatory, please use full first names. LIPIcs license is "CC-BY";  http://creativecommons.org/licenses/by/3.0/

 %TODO mandatory: Please choose ACM 2012 classifications from https://dl.acm.org/ccs/ccs_flat.cfm 


\relatedversiondetails[cite=enright2025families]{Full Version}{https://arxiv.org/abs/2505.15699}
\begin{document}

\maketitle

\begin{abstract}
Temporal graphs are graphs whose edges are labelled with times at which they are active. Their time-sensitivity provides a useful model of real networks, but renders many problems studied on temporal graphs more computationally complex than their static counterparts. To contend with this, there has been recent work devising parameters for which temporal problems become tractable. One such parameter is vertex-interval-membership (VIM) width. Broadly, this gives a bound on the number of vertices we need to keep track of at any given time to solve many problems. Our contributions are two-fold.  Firstly, we introduce a new parameter, tree-interval-membership (TIM) width, that generalises both VIM width and several existing generalisations.  Secondly, we provide meta-algorithms for both VIM and TIM width which can be used to prove fixed-parameter-tractability for large families of problems, bypassing the need to give involved dynamic programming arguments for every problem. In doing this, we provide a characterisation of problems in FPT with respect to both parameters. We apply these algorithms to temporal versions of Hamiltonian path, dominating set, matching, and edge deletion to limit maximum reachability. 
\end{abstract}

\section{Introduction}
Temporal graphs are graphs with a fixed vertex set and edges which appear and disappear over discrete timesteps.
They are useful for modelling networks which change over time in contexts such as communication~\cite{baker_gossips_1972}, epidemiology~\cite{enright_deleting_2021}, and transport~\cite{fuchsle_temporal_2022}. Allowing edges to appear and disappear results in many temporal problems being computationally complex, even on very restricted inputs~\cite{akrida_temporal_2021,mertzios_computing_2023}. One approach for tackling these problems is to look for fpt-algorithms\footnote{We use fpt (lowercase) as a descriptor for algorithms witnessing the inclusion of a problem in the parameterised complexity class FPT.}; that is, algorithms which run in time $f(k)|x|^{O(1)}$, where $f$ is a computable function, $x$ is the input instance, and $k$ is a parameter independent of the input size, that might capture structural information about the input, or properties of the desired solution. Problems admitting such an algorithm are said to belong to the class FPT with respect to $k$.
To design such algorithms, we need to identify temporal graph parameters for which the problem becomes tractable: this is an active area of research. In particular, there exist many temporal analogues to static structural graph parameters such as temporal feedback edge number~\cite{haag_feedback_2022}, timed feedback vertex number~\cite{casteigts_finding_2021}, temporal neighbourhood diversity, temporal modular-width, temporal cliquewidth~\cite{enright_structural_2024}, and multiple temporal versions of treewidth~\cite{fluschnik_as_2020,mans_treewidth_2014}. 

Here, we discuss a family of parameters which restrict the \emph{temporal} graph structure (in particular, the function that assigns edges to sets of times), but do not explicitly restrict the underlying graph. These parameters are specialisations of the treewidth of the undirected static expansion of a temporal graph, which has been proposed as a treewidth analogue for temporal graphs~\cite{fluschnik_as_2020}. The first, and most restrictive, parameter we consider is vertex-interval-membership width (VIM width), as defined by Bumpus and Meeks~\cite{bumpus_edge_2023}. Roughly speaking, this upper bounds the number of vertices at any given time which have been incident to an active edge in the past and will also be incident to an active edge in the future. A range of problems have been shown to be in FPT with respect to VIM width by dynamic programming arguments over a decomposition associated with the parameter~\cite{bumpus_edge_2023,enright_counting_2023,hand_making_2022}. Motivated by capturing how the connected components in a temporal graph evolve, Christodoulou et al.~\cite{christodoulou_making_2024} introduced three generalisations of VIM width: $\leq$-connected-vertex-interval-membership width, $\geq$-connected-vertex-interval-membership width, and bidirectional connected-vertex-interval-membership width. They also showed tractability of several problems by dynamic programming over a suitable decomposition.

We introduce a new parameter, tree-interval-membership (TIM) width, which generalises the parameters introduced by Christodoulou et al.~\cite{christodoulou_making_2024}, and Bumpus and Meeks~\cite{bumpus_edge_2023}. The definition is motivated by noting that, in many temporal problems, if two vertices are not in the same connected component at time $t$, they can be considered independently at that time. 
In addition to introducing this new parameter, our main contributions in this paper are two meta-algorithms
that solve large families of problems that satisfy certain simple conditions. We also show that these sufficient conditions are also necessary for the problem to be fixed-parameter tractable, and thus we give a complete characterisation of problems that are FPT when parameterised by VIM or TIM width.
These meta-algorithms provide shortcuts to proving that a problem admits an fpt-algorithm with respect to either of the parameters without the need to describe the details of a dynamic programming algorithm. We illustrate this by applying our meta-algorithms to several well-studied temporal problems. We note that finding an algorithm directly for each parameter will likely have a faster running time, but that our approach provides a framework for finding such an algorithm. In particular, we believe it is much simpler to apply our meta-algorithm than argue correctness of a specific algorithm over a TIM decomposition.

This work is organised as follows: we begin with preliminary definitions in Section~\ref{sec:prelims}. We define our new parameter and discuss its relationship with existing parameters in Section~\ref{sec:comparison}. We then give our meta-algorithms in Section~\ref{sec:meta}, with examples of their use in Sections~\ref{sec:applications-vim} and~\ref{sec:applications-tim}. We finish with some concluding remarks in Section~\ref{sec:conclusion}. 
% In the interest of space, most proofs can be found in the arXiv version of this paper~\cite{enright2025families}.

\subsection{Notation}\label{sec:prelims}
We begin by defining a temporal graph and some of its characteristics. A \emph{temporal graph} $\mathcal{G}$ consists of a static graph $G=(V,E)$ and a \emph{temporal assignment} $\lambda:E(G)\to 2^{\mathbb{N}}$ describing the times at which each edge in the graph is active. We give an example of a temporal graph in Figure~\ref{fig:eg-both}. The static graph $G$, also denoted $\mathcal{G}_{\downarrow}$, is known as the \emph{underlying} graph of $\mathcal{G}$. We refer to the pair $(e,t)$ where $t\in\lambda(e)$ as a \emph{time-edge}. The set of all time-edges in $\mathcal{G}$ is denoted by $\mathcal{E}(\mathcal{G})$. The static graph $G_t$ consisting of the vertex set $V(\mathcal{G}):= V$ and the edges which are active at time $t$ is called the \emph{snapshot} of $\mathcal{G}$ at time $t$. The \emph{lifetime} of a temporal graph, denoted $\Lambda(\mathcal{G})$ (or simply $\Lambda$ when the graph is clear from context), is the latest time at which an edge is active in the graph. That is, $\Lambda(\mathcal{G})=\max_{e\in E(\mathcal{G}_{\downarrow})}\max\lambda(e)$.
A \emph{strict temporal walk} on a temporal graph $\mathcal{G}$ is a sequence of time-edges $(e_0,t_0),\ldots,(e_l,t_l)\in \mathcal{E}(\mathcal{G})$ such that $e_0,\ldots,e_l$ form a walk in $\mathcal{G}_{\downarrow}$ and
$t_i<t_{i+1}$ for all $0\leq i< l$. A non-strict temporal walk is defined similarly, but the times on consecutive edges are non-decreasing, rather than increasing. 
% Throughout this work, all walks are assumed to be strict. 
A \emph{temporal path} is a temporal walk such that each vertex is traversed at most once.
% \begin{figure}[ht]
%     \centering
%     \includegraphics[width=0.65\linewidth,page=2]{figures.pdf}
%     \caption{An example of a temporal graph with lifetime 5.}
%     \label{fig:eg}
% \end{figure}%
% \subsection{Existing interval-membership width parameters}\label{sec:existing}

\section{A hierarchy of parameters}\label{sec:comparison}
Before introducing our new parameter, we discuss some existing interval-membership width parameters.
These rely heavily on the notion of an \emph{active interval} of a vertex $v$: that is, the time interval spanning the time of the first time-edge incident to $v$ and the last (inclusive).
We begin with \emph{vertex-interval-membership} (VIM) width, defined by Bumpus and Meeks~\cite{bumpus_edge_2023}.
\begin{defn}[Vertex-Interval-Membership Width (Bumpus and Meeks \cite{bumpus_edge_2023})]
    The vertex-interval-membership (VIM) sequence of a temporal graph $(G, \lambda)$ is the sequence $(F_t)_{t\in[\Lambda]}$ of vertex-subsets of $G$, called \emph{bags}, where $F_t = \{v \in V(G) : \exists uv,vw\in E(G)\text{, where } \min \lambda(vu) \leq t \leq \max \lambda(vw) \}$ and $\Lambda$ is the lifetime of $(G, \lambda)$. The \emph{vertex-interval-membership width} of a temporal graph $(G, \lambda)$ is the integer $\omega = \max_{t \in [\Lambda]}|F_t|$.
\end{defn}

% \begin{figure}[ht]
% \begin{subfigure}[h]{0.4\linewidth}
% \includegraphics[width=\linewidth,page=2]{figures.pdf}
% \end{subfigure}
% \hfill
% \begin{subfigure}[h]{0.55\linewidth}
% \includegraphics[width=\linewidth,page=12]{figures.pdf}
% \end{subfigure}%
% \caption{The temporal graph from Figure~\ref{fig:eg} together with the VIM sequence of this graph.}\label{fig:eg-vim}
% \end{figure}
 Although VIM width and algorithms over the VIM sequence are relatively straightforward, the class of temporal graphs for which this parameter is small is quite restricted: at each time, all but a small number of vertices must either have not yet been active or must never be active again. In a VIM sequence, all vertices in their active interval at a given time are in the same bag regardless of their graph-distance at that time. 
 By relaxing this requirement and taking connectivity into account, Christodoulou et al.~\cite{christodoulou_making_2024} introduce three notions of connected-VIM width. They use the temporal graphs $\mathcal{G}_{\leq}(t)$ and $\mathcal{G}_{\geq}(t)$ which consist of the vertices in $\mathcal{G}$ and the time-edges which appear in $\mathcal{G}$ at or before (respectively after) time $t$. The static graphs $G_{\leq}(t)$ and $G_{\geq}(t)$ are the underlying graphs of $\mathcal{G}_{\leq}(t)$ and $\mathcal{G}_{\geq}(t)$ respectively. 

\begin{defn}[Connected-Vertex-Interval-Membership Width (Christodoulou et al.~\cite{christodoulou_making_2024})]\label{def:CVIM}
    Let $d\in\{\leq, \geq\}$. For each time $t\in[\Lambda]$ and connected component $C$ in $G_{d}(t)$, let $\Psi_d(\mathcal{G},t,C)=V(C)\cap F_t$ be the $d$\emph{-connected bag} at time $t$ of $\mathcal{G}$ and $C$, where $F_t$ is the bag at time $t$ of the VIM sequence of $\mathcal{G}$. Let $\mathcal{F}_d(t)=\{\Psi_d(\mathcal{G},t,C)\, :\, C\text{ a connected component of }G_d(t)\}$. The \emph{$d$-connected-vertex-interval-membership width} $\psi$ is given by $\psi_d(\mathcal{G})=\max_{t\in[\Lambda],\Psi\in\mathcal{F}_d(t)}|\Psi|$.
\end{defn}
They show that the two directions of connected-VIM width are incomparable, but both generalise VIM width as defined by Bumpus and Meeks. Christodoulou et al.~also consider a bidirectional connected-VIM width which generalises both connected-VIM widths.%
\begin{defn}[Bidirectional Connected-Vertex-Interval-Membership Width (Christodoulou et al.~\cite{christodoulou_making_2024})]
    Let $\psi_{\leq}(\mathcal{G})$, $\psi_{\geq}(\mathcal{G})$ be the $\leq$- and $\geq$-connected-VIM width of the temporal graph $\mathcal{G}$, respectively. Then, the bidirectional connected-vertex-interval-membership width $\psi_{\sim}(t)$ of a temporal graph $\mathcal{G}$ at time $t$ is
    
    $
       \psi_{\sim}(t) =\left\{ \begin{array}{ll} \max\{\psi_{\leq}(\mathcal{G}_{\leq}(t-1)),\psi_{\geq}(\mathcal{G}_{\geq}(t+1)), |F_t(\mathcal{G})|\} & \text{if }1<t<\Lambda\\
       \psi_{\geq}(\mathcal{G}) & \text{if } t=1\\
       \psi_{\leq}(\mathcal{G}) & \text{if }t=\Lambda
       \end{array}
       \right.        
    $\\
    where $F_t(\mathcal{G})$ is the bag at time $t$ of the VIM sequence of $\mathcal{G}$.
    The \emph{bidirectional connected-VIM width} of $\mathcal{G}$ is $\min_{t\in[\Lambda]}\psi_{\sim}(t)$.
\end{defn}

% \subsection{New parameter: Tree-interval-membership width}\label{sec:TIM-Def}
We now define our generalisation of the existing interval-membership width parameters, namely tree-interval-membership (TIM) width.  
Unlike the sets in the VIM decomposition which are linearly ordered (one set is associated with each timestep), the bags of a TIM decomposition are indexed by an arbitrary directed tree (a directed graph whose underlying graph is a tree); moreover, there can be multiple bags associated with \emph{every} timestep, whereas for each of the parameters introduced by Christodoulou~et~al. there is some time associated with a single bag.  
%Intuitively, we can think of this new parameter as a combination of vertex-interval-membership width and treewidth. We can draw a comparison between the relationship between pathwidth and treewidth and the relationship between vertex-interval-membership width and tree-interval-membership width. 
%Specifically, the bags of a VIM sequence (with the addition of bags of cardinality 1 containing each vertex not in the bag of the VIM sequence at that time) form a TIM decomposition where the subgraph of the underlying graph of the decomposition induced by vertices in their active intervals is a path. 

The utility of our new parameter comes from noticing that, in order to store enough information to solve many natural problems, vertices in different connected components of a temporal graph need not be placed in the same bag despite being active at the same time. Therefore, when solving temporal problems where we can consider the connected components of each snapshot independently, we can use TIM width in place of VIM width. In the remainder of this section, we assume all temporal graphs have connected underlying graphs. For clarity, we refer to \emph{vertices} of the original graph and \emph{nodes} of the indexing tree.
\begin{defn}[Tree-Interval-Membership Width]\label{def:TIMW}
    We say a triple $(T,B,\tau)$ is a \emph{tree-interval-membership decomposition} (TIM decomposition) of a temporal graph $\mathcal{G}$ with lifetime $\Lambda$ if $T$ is a labelled directed tree, where $B = \{B(s) : s\in V(T)\}$ is a collection of subsets of $V(\mathcal{G})$, called \emph{bags}, and $\tau: V(T)\to [\Lambda]$ is a function which labels each node with a time, satisfying:
    \begin{enumerate}
        \item For all vertices $v\in V(\mathcal{G})$ and times $t\in [\Lambda]$, there exists a unique node $i\in V(T)$ such that $\tau(i)=t$ and $v\in B(i)$.
        \item For all time-edges $(uv,t)\in\mathcal{E}(\mathcal{G})$, there exists a node $i\in V(T)$ such that $\{u,v\}\subseteq B(i)$ and $\tau(i)=t$.
        \item The set of arcs of $T$ is given by $\{(i,j) : B(i)\cap B(j)\neq \emptyset \text{ and } \tau(j)=\tau(i)+ 1\}$.
    \end{enumerate}
The \emph{width} of a TIM decomposition is defined to be $\max\{|B(s)|:s\in V(T)\}$. The \emph{TIM width} of a temporal graph $\mathcal{G}$ is the minimum $\phi$ such that $G$ has a TIM decomposition of width $\phi$. 
\end{defn}

We leave open whether the TIM width of a temporal graph is computable in polynomial time, but are optimistic that choosing the grouping of vertices in bags will require an exponential factor only in the width and not in the size of the input.
\begin{figure}[ht]
    \centering
    \includegraphics[width=0.7\linewidth,page=15]{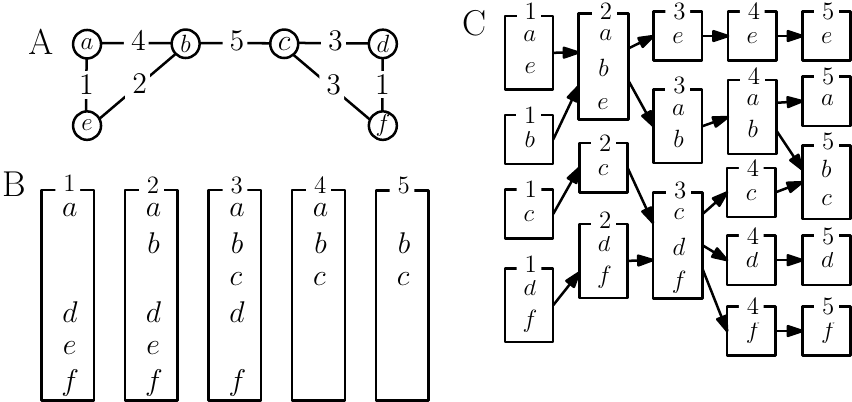}
    \caption{(A) An example temporal graph $\mathcal{G}$, where the number(s) on an edge indicates the time(s) at which it is active. Note that this graph has lifetime~5. The VIM sequence (B) of $\mathcal{G}$, and a TIM decomposition (C) of $\mathcal{G}$.}
    \label{fig:eg-both}
\end{figure}

Figure~\ref{fig:eg-both} gives an example of a temporal graph, and a comparison of its VIM sequence and a TIM decomposition of the graph. We may abuse notation by referring to $t$ as the label of a bag $B(i)$ when $\tau(i)=t$. We now compare TIM width and VIM width with some related parameters. A diagram depicting a hierarchy of parameters can be seen in Figure~\ref{fig:hierarchy}.

% It turns out we can find a minimum-width TIM decomposition in polynomial time (a proof can be found in Section~\ref{sec:proof1}).
\begin{toappendix}
This definition gives us some simple observations.
    
\end{toappendix}
% and an algorithm for computing a 2-approximation TIM width of a temporal graph which can be found in the Appendix.

% \begin{figure}[ht]
% \begin{subfigure}[h]{0.4\linewidth}
% \includegraphics[width=\linewidth,page=2]{figures.pdf}
% \end{subfigure}
% \hfill
% \begin{subfigure}[h]{0.5\linewidth}
% \includegraphics[width=\linewidth,page=7]{figures.pdf}
% \end{subfigure}%
% \caption{The temporal graph from Figure~\ref{fig:eg} together with a TIM decomposition of this graph.}\label{fig:eg-tim}
% \end{figure}

 % \begin{figure}[ht]
 %     \centering
 %     \includegraphics[width=0.5\linewidth,page=7]{figures.pdf}
 %     \caption{A TIM decomposition of the temporal graph in Figure~\ref{fig:eg}.}
 %     \label{fig:eg-tim}
 % \end{figure}

\begin{toappendix}
 Based on our definition of TIM width (Definition~\ref{def:TIMW}), we have some observations.
\begin{observation}\label{obs:tim-node-bound}
    There are at most $n\Lambda$ nodes in a TIM decomposition.
\end{observation}

\begin{observation}\label{obs:tim-all-one-bag}
    The decomposition found by creating one bag at every timestep containing all vertices in a temporal graph is a TIM decomposition.
\end{observation}

We refer to the \emph{neighbours} of a node $s$ as the nodes $s'$ such that either the arc $ss'$ or the arc $s's$ exists.

\begin{observation}\label{obs:tim-neighbours}
    For a bag $B(s)$ of a TIM decomposition $(T,B,\tau)$ with time $t=\tau(s)$, any node $s'$ neighbouring $s$ in $T$ must be assigned $t'\in\{t+1,t-1\}$ by $\tau$.
\end{observation}

\begin{observation}\label{obs:tim-no-children}
    For any node $s$ in a TIM decomposition, there are at most $2\phi$ neighbours of $s$, where $\phi$ is the width of the decomposition.
\end{observation}

\begin{observation}\label{obs:tim-directed-path}
    For all vertices $v\in V(\mathcal{G})$, $T[\{s:v\in B(s)\}]$ is a directed path. That is, for each vertex, the subgraph of the TIM decomposition obtained by deleting every node not containing $v$ in its bag from $T$ is a directed path.
\end{observation}

% \begin{observation}\label{obs:tim-timeedge-appear-once}
%     Every time-edge appears in exactly one bag of a TIM decomposition.
% \end{observation}

As a result of the above observation, we can extend this reasoning to entire connected components of a snapshot of a temporal graph.

\begin{observation}\label{obs:tim-connected-component}
    Let $u$ and $v$ be vertices in the same connected component of $G_t$ for some $t$. Then, if $u\in B(i)$ and $\tau(i)=t$, $v\in B(i)$.
\end{observation}

We conclude our observations with a note on edges which are active more than once.

\begin{observation}\label{obs:tim:nonsimple-edge}
    If $uv$ is an edge in $\mathcal{G}_{\downarrow}$ and $u$ and $v$ are in different bags of a TIM decomposition at time $t$, then either $t>\max(\lambda(uv))$ or $t<\min(\lambda(uv))$.
\end{observation}

  \end{toappendix}

\begin{figure}[ht]
    \centering
    \hspace*{2cm}
    \includegraphics[width=0.75\linewidth,page=16]{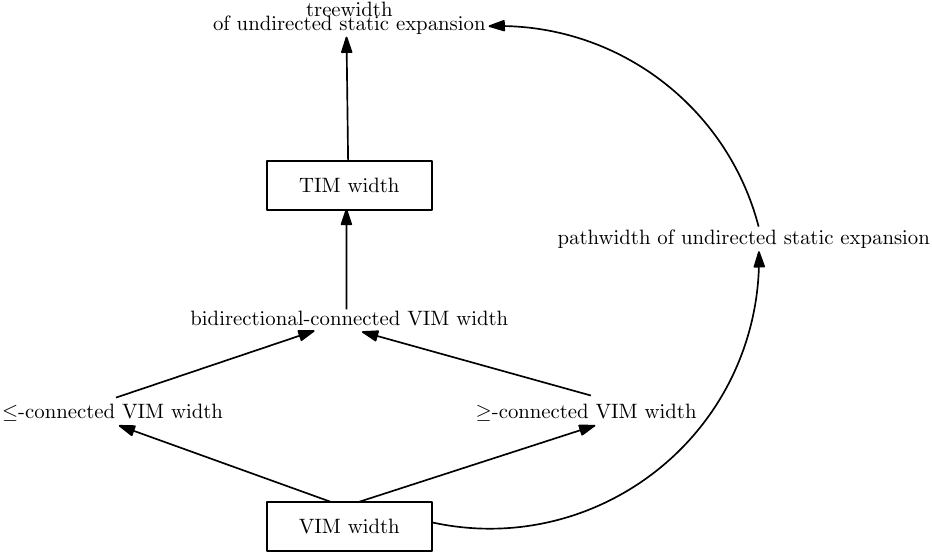}
    \caption{A hierarchy of parameters. There is an arc from parameter A to parameter B if bounding A implies that B is also bounded. The relationships are strict -- for every arc from A to B, there exists an infinite family of graphs for which B is bounded and A is unbounded. The parameters we are focussing on are highlighted with boxes.}
     \label{fig:hierarchy}
\end{figure}

% Proofs of these relationships can be found in the appendix.

% A comparison of the decompositions related to the parameters VIM width, bidirectional connected-VIM width, and TIM width can be found in Figure~\ref{fig:comparison}. 

It is straightforward to see that the VIM width of a temporal graph is always at least the TIM width, since we can turn a VIM sequence of a temporal graph $\mathcal{G}$ into a TIM decomposition of the same width by letting each set $F_t$ be a bag labelled with time $t$, and placing every vertex not active at time $t$ in a singleton bag. It transpires that TIM width also lower bounds all three parameters introduced by Christodoulou et al.~\cite{christodoulou_making_2024}.

\begin{figure}[ht]
    \centering
    \includegraphics[width=0.5\linewidth,page=6]{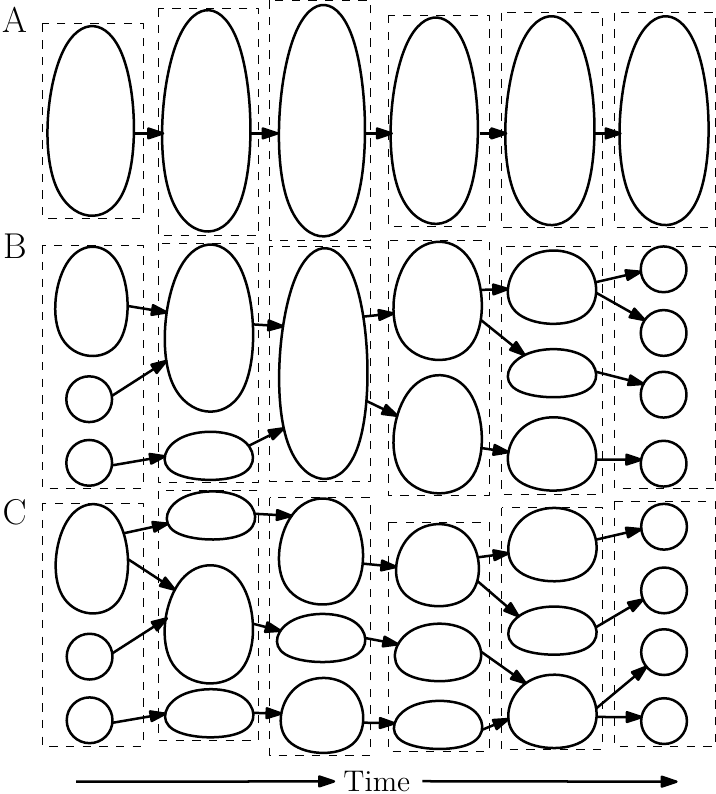}
    \caption{A comparison of the bags of a VIM sequence (A), a bidirectional connected-VIM sequence (B), and a TIM decomposition (C).  Dashed boxes group the bags of the decompositions which are labelled with the same time. The point here is that as the bags of the decompositions decrease in size, the structure of the decomposition graph becomes more unruly. In decomposition (B), there is a bag from which all bags branch out. If the image were to depict a $\leq$- or $\geq$-connected-VIM decomposition instead, the bag would be at either the start or end, respectively.}
    \label{fig:comparison}
\end{figure}

\begin{lemmarep}
\label{lem:tim-bounds-cimw}
    Let $\mathcal{G}$ be a temporal graph such that $\min\{\psi_{\leq}, \psi_{\geq}\}=k$, where $\psi_d$ is the $d$-connected-VIM width. Then, $\mathcal{G}$ has TIM width at most $k$.
\end{lemmarep}
\begin{proof}
    We have two cases to consider: $\psi_{\leq}= k$ and $\psi_{\geq}\geq k$, or $\psi_{\leq}\geq k$ and $\psi_{\geq}= k$. We prove both simultaneously by substituting $\leq$ or $\geq$ for $d$.
    
    If $\min\{\psi_{\leq}, \psi_{\geq}\}=k$, then we know that that for all $t\in[\Lambda]$ and connected components $C$ of $G_{d}(t)$, $|V(C)\cap F_t|\leq k$. Consider a decomposition such that, for all $t\in [\Lambda]$, the bags labelled with time $t$ consist of a bag containing $V(C)\cap F_t$ for each connected component $C$ of $G_{d}(t)$; all vertices in $V(G_{d}(t))\setminus F_t$ are placed in singleton bags; and arcs exist from a bag at time $t$ to a bag at time $t+1$ if their intersection is non-empty. We claim that this is a TIM decomposition. Note that, under this construction, each vertex in $\mathcal{G}$ appears in exactly one bag labelled with each time. Furthermore, all time-edges $(e,t)$ appear in a bag at time $t$. 
    
    What remains to check is that the underlying graph of this decomposition is a tree. Suppose, for a contradiction, that $u$ and $v$ are in the same bag of the decomposition at times $t_1$ and $t_3$, and a different bag at time $t_2$, where $t_1<t_2<t_3$. This implies that $u$ and $v$ are in different connected components of $G_d(t_2)$, but in the same connected component of $G_d(t_1)$ and $G_d(t_3)$. This cannot be possible, since $G_d(t)$ is found by taking the union of all edges which appear either up to and including $t$ or from $t$ onwards. Therefore, there cannot be two vertices $u$ and $v$ such that they are in the same bag of the decomposition at times $t_1$ and $t_3$, and a different bag at time $t_2$, for some times $t_1<t_2<t_3$. Thus, there are no cycles in the underlying (undirected) graph of the decomposition, and this is in fact a TIM decomposition. Therefore, the TIM width of $\mathcal{G}$ is at most the size of the largest bag; that is, $k$.
\end{proof}
\begin{lemmarep}\label{lem:tim-bounds-bi-cimw}
    Let $\mathcal{G}$ be a temporal graph such that $\psi_{\sim}(\mathcal{G})=k$. Then, $\mathcal{G}$ has TIM width at most $k$.
\end{lemmarep}
\begin{proof}
    We have 3 cases to consider. In the first two, the minimum value found when calculating $\psi_{\sim}$ is either $\psi_{\leq}(\mathcal{G})$ or $\psi_{\geq}(\mathcal{G})$. If this is the case, we leverage Lemma~\ref{lem:tim-bounds-cimw}.

    This leaves us with the case where there exists a $t$ such that $1<t<\Lambda$, and $t$ minimises $\max\{\psi_{\leq}(\mathcal{G}_{\leq}(t-1)),\psi_{\geq}(\mathcal{G}_{\geq}(t+1)), F_t\}=k$. We claim that the decomposition $(T,B,\tau)$ such that:
    \begin{itemize}
        \item for all $t'\in [t-1]$, the bags of nodes labelled with time $t'$ by $\tau$ consist of a bag containing $V(C)\cap F_t$ for each connected component $C$ of $G_{\leq}(t')$; all vertices in $V(G_{\leq}(t'))\setminus F_{t'}$ are placed in singleton bags;
        \item the bags assigned time $t$ by $\tau$ consist of a bag containing all vertices in $F_{t}$, and the remaining vertices in singleton bags;
        \item for all $t''\in [t+1,\Lambda]$, the bags labelled with time $t''$ by $\tau$ consist of a bag containing $V(C)\cap F_{t''}$ for each connected component of $G_{\geq}(t'')$; all vertices in $V(G_{\geq}(t'))\setminus F_{t''}$ are each placed in singleton bags; 
        \item and arcs exist from a bag at time $t$ to a bag at time $t+1$ if their intersection is non-empty
    \end{itemize}
    is a TIM decomposition. 
    
    It is clear from the construction that each vertex in $V(\mathcal{G})$ appears exactly one in a bag at each time, and each time-edge $(e,t')$ in $\mathcal{E}(\mathcal{G})$ appears in a bag labelled with time $t'$. What remains is to show that there do not exist vertices $u$ and $v$ and times $t_1<t_2<t_3$ such that $u$ and $v$ are in the same bag of the decomposition at times $t_1$ and $t_3$, and a different bag at time $t_2$. We know that no such times and vertices exist if $t_1>t$ or $t_3<t$ by our proof of Lemma~\ref{lem:tim-bounds-cimw}. This gives us two cases to consider: $t_2=t$, or $t_2\neq t$.

    If $t_2=t$, then $u$ and $v$ would be in different bags at time $t$. This implies that one of $u$ and $v$ is not in its active interval. Thus, it must be in a singleton bag for all $t'>t$ or all $t'<t$. This gives us a contradiction.

    If $t_2 \neq t$, then either $t_2<t$ and $u$ and $v$ are in the same connected component of $G_{\leq}(t_1)$, or $t_2>t$ and $u$ and $v$ are in the same connected component of $G_{\geq}(t_1)$. In the first case, this implies that $u$ and $v$ are in the same connected component of $G_{\leq}(t_2)$. In the second, this implies that $u$ and $v$ are in the same connected component of $G_{\geq}(t_2)$. This gives us a contradiction in both scenarios. Therefore, the decomposition described is a TIM decomposition of width $k$. Thus, $\mathcal{G}$ has TIM width at most $k$.
\end{proof}
Figure~\ref{fig:comparison} illustrates the form of the decompositions associated with VIM width, bidirectional connected-VIM width, and TIM width. 

% \noindent

In fact, the TIM width of a temporal graph can be arbitrarily smaller than its bidirectional connected-vertex-interval-membership width.  We illustrate this in Figure~\ref{fig:timw-cimw} by adapting an infinite family of graphs used by Christodoulou et al.~\cite[Figure~1]{christodoulou_making_2024} to demonstrate that $d$-connected-VIM width (with $d \in \{\leq,\geq\}$) can be arbitrarily smaller than VIM width. The temporal graph in this figure is such that each edge is active exactly once, each connected component at each time consists of at most two vertices, and the underlying graph is a tree. This gives us that the TIM width of the graph is 2. For the bidirectional connected-VIM width, we note that both $\psi_{\leq}(\mathcal{G})$ and $\psi_{\geq}(\mathcal{G})$ are both $2k$ since in the first (respectively, last) timestep the bag of the VIM sequence contains all vertices on the path and half of the leaves and they form a subgraph of a connected component of $\mathcal{G}_{\leq}(\Lambda)$ (respectively, $\mathcal{G}_{\geq}(1)$). Furthermore, for all times $t$ in $(1,\Lambda)$, the non-leaf vertices of $\mathcal{G}$ are in the bag $F_t$ of the VIM sequence of $\mathcal{G}$. Therefore, all such bags have cardinality $k$. For all times $t$ in $(1,\Lambda)$, $\psi_{\leq}(\mathcal{G}_{\leq}(t-1))=\psi_{\geq}(\mathcal{G}_{\geq}(t+1))=2$. Since we take the maximum of the bag of the VIM sequence and these two values, we get that the bidirectional connected-VIM width of this graph is $k$.

\begin{figure}[h!]
    \centering
    \includegraphics[width=0.35\linewidth,page=4]{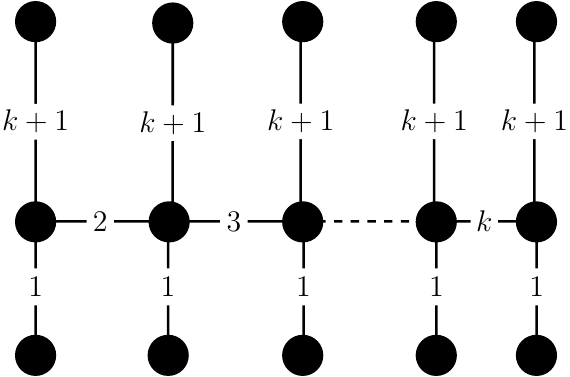}
    \caption{An example of a temporal graph $\mathcal{G}$ with TIM width 2 and $\psi_{\sim}(\mathcal{G})\geq k$
    % if $k$ is even, $\psi_{\sim}(\mathcal{G})=k+1$ if $k$ is odd. 
    The dashed edge replaces a path consisting of $k-4$ edges labelled with consecutive times.}
    \label{fig:timw-cimw}
\end{figure}

For the remainder of this section, we compare TIM and VIM width of a temporal graph $\mathcal{G}$ to structural parameters of the static expansion of $\mathcal{G}$. The static expansion of a temporal graph can be thought of as a static representation of a temporal graph, and it is also known as the time-expanded graph.
We use the definition of static expansion by Fluschnik et al.~\cite{fluschnik_as_2020}. An example can be seen in Figure~\ref{fig:eg2}.

\begin{defn}[Static expansion~\cite{fluschnik_as_2020}, Definition 2]
The \emph{static expansion} of a temporal graph $\mathcal{G}$ is a directed graph $\vec{H}:=(V',A)$, with vertices $V'=\{v_{t}\, :\, v\in V(\mathcal{G}), t\in [\Lambda]\}$ and arcs $A=A'\cup A_{\text{col}}$, where $A':=\{(u_{t},u'_{t}) \, :\, (uu',t)\in \mathcal{E}(\mathcal{G})\}$, and $A_{\text{col}}:=\{(u_{t},u_{t+1})\, :\, u\in V(\mathcal{G}),t\in [\Lambda-1]\}$.
\end{defn}

\begin{figure}
    \centering
    \includegraphics[width=0.5\linewidth, page=1]{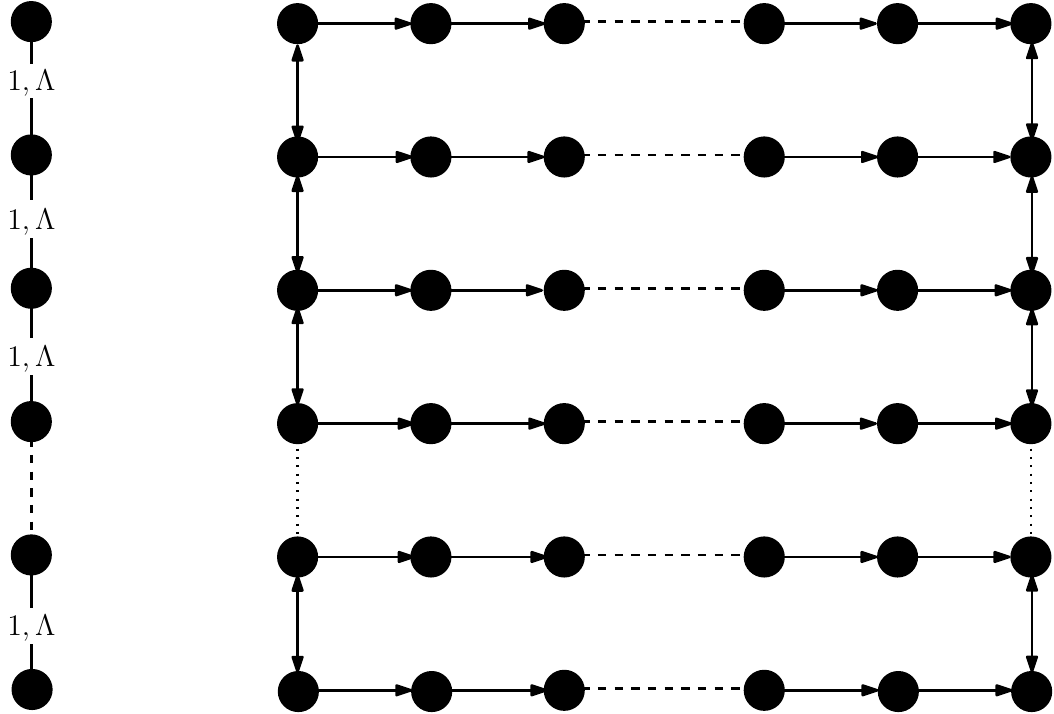}
    \caption{A path on $n$ vertices where all edges are only active at times 1 and $\Lambda$ and its static expansion. Dashed lines a portion of a (directed) path that is not pictured.}
    \label{fig:eg2}
\end{figure}

The static expansion of a temporal graph is a directed graph. To compare VIM width and TIM width to undirected static parameters, we consider the undirected static expansion. That is, for a temporal graph $\mathcal{G}$ with static expansion $\vec{H}=(V',A)$, let $H:=(V',E')$ be the undirected graph with the same vertex set as $H$ such that an edge $uv$ exists in $E'$ if there is an arc $(u,v)$ or $(v,u)$ in $A$.

% \begin{observation}\label{obs:tim-bound-tpw}
%     Let $\mathcal{G}$ be a temporal graph with undirected static expansion $H$. Then any TIM decomposition $(T, B,\tau)$ of $\mathcal{G}$ gives a tree-partition of $H$.
% \end{observation}

% Using this observation and a lower bound proved by Seese~\cite{seese1985tree}, we can relate TIM width, tree-partition width, and treewidth. Note that there exist graphs with bounded treewidth and unbounded tree-partition width~\cite{bodlaender1997domino}. Treewidth is defined as follows.
The first static parameter we consider is \emph{treewidth}.
\begin{defn}[Tree Decomposition, Treewidth]
    We say a pair $(T,B)$ is a \emph{tree decomposition} of $G$ if $T$ is a tree and $B = \{B(s) : s\in V(T)\}$ is a collection of subsets of $V(G)$, called \emph{bags}, satisfying:
    \begin{enumerate}
        \item $V(G)=\bigcup_{s\in V(T)}{B(s)}$.
        \item $\forall uv \in E(G): \exists s \in V(T): \{u,v\}\subseteq B(s)$. That is, for each edge in the graph, there is at least one bag containing both of its endpoints.
        \item $\forall v\in V(G): T[\{s:v\in B(s)\}]$ is connected; for each vertex, the subgraph obtained by deleting every node not containing $v$ in its bag from $T$ is connected.
    \end{enumerate}
The \emph{width} of a tree decomposition is defined to be $\max\{|B(s)|:s\in V(T)\}-1$. The \emph{treewidth} of a graph $G$ is the minimum $\omega$ such that $G$ has a tree decomposition of with $\omega$. 
\end{defn}

We can now upper bound the treewidth of the undirected static expansion by the TIM width of the corresponding temporal graph. We will later show that we cannot achieve a similar lower bound. 
\begin{theoremrep}\label{thm:tim-tw-tpw}
    Suppose $\mathcal{G}$ is a temporal graph with TIM width $\phi$ and undirected static expansion $H$. Denote by $\mathrm{tw}(H)$ the treewidth of $H$. Then, $2\phi\geq \mathrm{tw}(H)+1$.
\end{theoremrep}

\begin{proof}
    We can see the relationship between TIM width and treewidth of the undirected static expansion directly. We begin by labelling all vertices in a bag $B(s)$ of a TIM decomposition $(T,B,\tau)$ with the time $\tau(s)$. Then, for each pair of adjacent bags in a TIM decomposition, we subdivide the edge between them and add the union of the adjacent bags to the new bag between them. This gives us a decomposition whose bags have at most double the number of elements as the bags in the TIM decomposition. Call this new decomposition $(T',B')$.
    
    We now assert that $(T',B')$ has the desired properties of a tree decomposition. Since each vertex appears exactly once at each time in a TIM decomposition, each vertex in the undirected static expansion must be in at least one bag of $(T',B')$. Furthermore, since all time-edges in $\mathcal{G}$ appear exactly once in a TIM decomposition of $\mathcal{G}$, the edges of the undirected static expansion which are between copies of distinct vertices in $\mathcal{G}$ must be in at least one bag of $(T',B')$. What remains is the edges between consecutive copies of the same vertices in the static expansion. Since we take the union of bags of adjacent nodes in the TIM decomposition to find $(T',B')$, and (by Definition~\ref{def:TIMW} and Observation~\ref{obs:tim-neighbours}) the copy of any vertex in $B(s)$ at the time before or after $\tau(s)$ must be in a bag of a neighbour of $s$. Therefore, all edges in the undirected static expansion are in a bag of $(T',B')$. Finally, we show that the subtree of $T'$ induced by the bags containing any vertex $v_t$ of the undirected static expansion of $\mathcal{G}$ is a connected graph. To see this, recall that every vertex appears exactly once in a bag labelled with each time in $(T,B,\tau)$. Since we construct $(T',B')$ by taking the union of adjacent bags of $(T,B,\tau)$, the subgraph of $T'$ induced by the bags containing any given vertex of the undirected static expansion must be a star; the node $s$ associated to the bag $B(s)$ containing the vertex in $(T,B,\tau)$ and a node for each neighbour $s'$ of $s$ into which we add the union of $B(s)$ and $B(s')$. Thus, $(T',B')$ is a tree decomposition of the undirected static expansion, and the treewidth of the undirected static expansion of $\mathcal{G}$ is at most $2\phi$ where $\phi$ is the TIM width of $\mathcal{G}$. 
\end{proof}

As can be seen in Figure~\ref{fig:comparison}, the bags of a VIM sequence form a path rather than a tree. This allows us to draw a comparison between VIM width and pathwidth -- a width measure which requires decomposition of the graph into a path rather than a tree.

% \begin{defn}[Pathwidth]\label{def:pw}
% A \emph{path decomposition} of a static graph $H$ is a sequence of subsets $X_i$ of vertices with the properties:
% \begin{enumerate}
%     \item For each edge in $E(H)$ there exists an $i$ such that both endpoints of the edge are in $X_i$,
%     \item For every three indices $i\leq j\leq k$, $X_i\cap X_k\subseteq X_j$.
%     \end{enumerate}
% Equivalently, a path decomposition is a tree decomposition such that the decomposition graph is a path.
% The \emph{pathwidth} of a graph is defined to be $\min\max_i|X_i|-1$.
% \end{defn}
\begin{defn}\label{def:pw}
    A \emph{path decomposition} is a tree decomposition such that the decomposition graph is a path. The \emph{pathwidth} of a graph is the minimum width of a path decomposition of the graph, where the width of a decomposition is the cardinality of the largest bag minus~1.
\end{defn}
Figure~\ref{fig:eg2} gives an example of a temporal graph with unbounded VIM and TIM width, and static expansion with pathwidth (and treewidth) 5.
\begin{toappendix}
    We now discuss how to construct a path decomposition of the undirected static expansion of a temporal graph based on the VIM sequence of the graph. Recall that we denote a vertex in the static expansion by $v_t$, where $v$ is the corresponding vertex in the temporal graph, and $t$ is the time with which it is labelled. 
\end{toappendix}

\begin{theoremrep}\label{thm:vimw-bounded-pw-static}
    Suppose $\mathcal{G}$ is a temporal graph with VIM width $\omega$ and undirected static expansion $H$. Denote by $\mathrm{pw}(H)$ the pathwidth of $H$. Then, $2\omega\geq \mathrm{pw}(H)+1$.
\end{theoremrep}
\begin{proof}
    We begin our construction by creating a sequence of bags such that each bag is associated to a bag of the VIM sequence. Refer to these bags as \emph{path} bags. Let $B$ be a path bag associated to a bag $B$' at time $t$ of the VIM sequence. Then $B$ contains a copy of each vertex in $B'$ at times $t$ and $t-1$ (unless $B'$ is at time 1, when it contains a single copy of the vertices in $B'$ at time 1). 

We now note that the subgraph of the undirected static expansion induced by the vertices which are not in a path bag consists of a set of disjoint paths. Denote by $\mathcal{P}$ the set of paths found by removing all vertices in a path bag from the undirected static expansion. The paths in $\mathcal{P}$ correspond to vertices forgotten or introduced by the bags in the VIM sequence. Observe that for all $P$ in $\mathcal{P}$, the vertices in $P$ are consecutive copies of one vertex $v$ in $\mathcal{G}$ and one endpoint of $P$ is $v_{1}$ or $v_{\Lambda}$ (or both if $v$ is isolated in $\mathcal{G}$). The other endpoint of $P$ must be adjacent to a vertex in a path bag. If $v$ is an isolated vertex, we can add a sequence of bags each containing two consecutive copies of $v$ to the beginning of the path decomposition. We claim that the set $\mathcal{P}$ does not add to the width of the path decomposition of the undirected static expansion. To see this recall that in each path bag, there are two copies of every vertex in the associated bag of the VIM sequence. Therefore, for each path bag $B$ with a set of pendant paths $\mathcal{P}_B$ in $\mathcal{P}$, we add 
\begin{itemize}
    \item for all paths $P\in \mathcal{P}_B$ whose endpoints are $v_t\in B$ and $v_1$, a sequence of bags before $B$ where the bags contain the same vertices apart from the copies of each such $v$ which are decremented by $1$ until we reach a bag containing $v_1$ and $v_2$; or
    \item for all paths $P\in \mathcal{P}_B$ whose endpoints are $v_t\in B$ and $v_{\Lambda}$, a sequence of bags before $B$ where the bags contain the same vertices apart from the copies of each such $v$ which are incremented by $1$ until we reach a bag containing $v_{\Lambda-1}$ and $v_{\Lambda}$.
\end{itemize}
Observe that this decomposition remains a path and that the largest bag in this decomposition contains twice as many vertices as the largest bag of the VIM sequence.

    Given the VIM sequence of $\mathcal{G}$, we initialise the path bags of the path decomposition of the undirected static expansion in $O(k\Lambda)$ time by making $\Lambda$ bags and adding at most $2k$ vertices to each. Between each consecutive pair, we add at most $2k$ sequences corresponding to pendant paths in $\mathcal{P}$.
    % Between each consecutive pair, we then add at most one forgetting sequence followed by at most one introducing sequence. 
    % We find the set of forgotten and introduced vertices in $O(k)$ time. The bags in these sequences have size at most $O(k)$, and the sequences have length at most $O(\Lambda)$. Therefore, constructing the path decomposition requires $O(\Lambda^2k)$ time.

    % Performing the forgetting of old vertices before introduction of new vertices enforces that there are at most $2k$ vertices in any bag of the new decomposition if the VIM width of the original graph is $k$. 
    Recall that there are at most $2k$ vertices in any bag of the new decomposition if the VIM width of the original graph is $k$. 
    What remains to show is that these bags do in fact form a path decomposition of the undirected static expansion of $\mathcal{G}$. By Definition~\ref{def:pw}, we have two criteria to check: for all edges in the static expansion, is there a bag containing both endpoints? And, for all indices $i\leq j\leq k$, is it true that $X_i\cap X_k\subseteq X_j$? For the latter question, since we only introduce (and forget) each vertex once in the VIM sequence and the labels of the vertices are non-decreasing in the decomposition, all vertices of the undirected static expansion appear in a single interval of the decomposition constructed.
    Thus, if a vertex appears in bag $X_i$ and $X_k$ for $i\leq  k$, then it must also appear in each bag $X_j$ for $i\leq j\leq k$.

    To show that the first property holds, we note that for all time edges $(vu,t)$ in $\mathcal{G}$, the bag $F_t$ of the VIM sequence must contain both $v$ and $u$. Thus, there is a path bag containing the corresponding edge in the undirected static expansion. Furthermore, for every vertex $v\in V(\mathcal{G})$ and time $t\in[\Lambda]$, there is at least one bag containing the pair $v_{t-1},v_t$ in our decomposition. Since we have a bag of our decomposition graph consisting of $\{v_{t-1}, v_t \, : \, v\in F_t\}$, and the edges of the undirected static expansion of $\mathcal{G}$ are of the form $\{v_tv_{t+1} \,:\, v\in V(\mathcal{G})\}\cup \{v_{t}u_{t} \, : \, (vu,t)\in \mathcal{E}(\mathcal{G})\}$ the endpoints of all edges in the undirected static expansion of $\mathcal{G}$ must appear together in at least one bag in our decomposition. Thus, we have constructed a path decomposition of the undirected static expansion of $\mathcal{G}$ of width at most $2k-1$.
\end{proof}

\begin{toappendix}
    See Figure~\ref{fig:eg-both} for an example temporal graph and its VIM sequence. An example path decomposition of the undirected static expansion is given in Figure~\ref{fig:pw-eg}.

\begin{figure}[ht]
    \centering
    \includegraphics[width=0.9\linewidth,page=13]{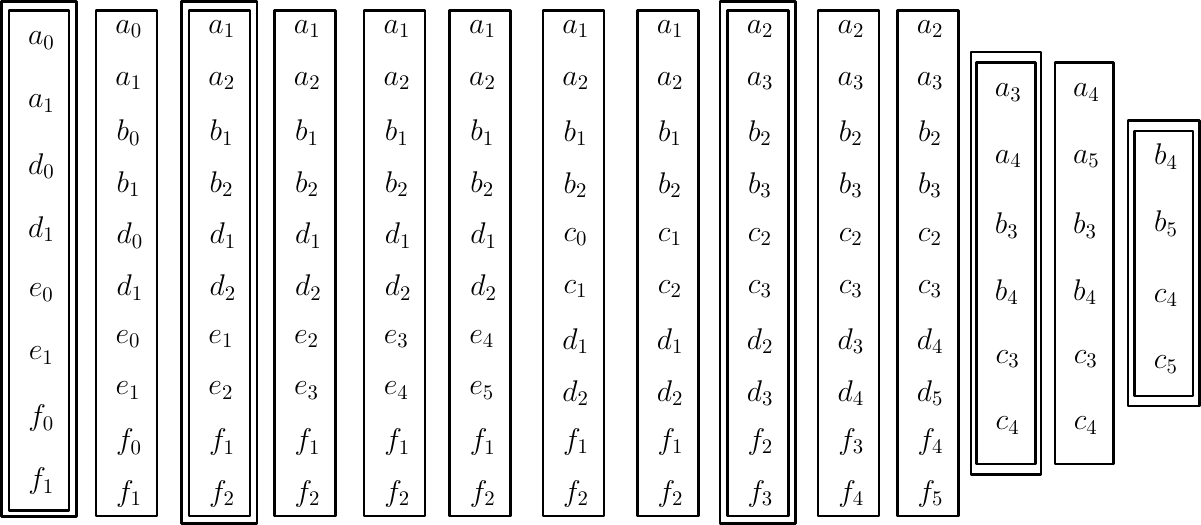}
    \caption{A path decomposition of the temporal graph in Figure~\ref{fig:eg-both} constructed using the method given in Theorem~\ref{thm:vimw-bounded-pw-static} and the VIM sequence in Figure~\ref{fig:eg-both}. Path bags are double boxed for emphasis.}
    \label{fig:pw-eg}
\end{figure}

Next we show that there exist temporal graphs whose undirected static expansion has small pathwidth and have unbounded VIM width. Note that this also implies that the treewidth of the undirected static expansion is bounded and that the TIM width is unbounded. For simplicity, consider a temporal graph whose underlying graph is a path and all edges are active only at time $1$. Clearly the VIM and TIM width of this graph are $n$. However, the undirected static expansion of this graph is the same as its underlying graph. This is a path, and so has path- and treewidth $1$.

\subsection{Algorithmic distinctions}
\end{toappendix}

Courcelle's Theorem~\cite{courcelle1997expression} implies that any temporal problem which is expressible using an MSO formula of length $l$ on the undirected static expansion is in FPT with respect to $l$ and treewidth of the undirected static expansion combined; implying inclusion in FPT with respect to $l$ and either TIM or VIM combined. We provide a simpler approach for proving tractability, likely with a faster runtime. Furthermore, we now show that we can distinguish between the algorithmic power of TIM width and the treewidth of the undirected static expansion using a temporal variant of \textsc{Equitable Connected Partition}. This variant looks for a partition of vertices such that the parts are close in size, and for any pair of vertices in the same part there exists a nonstrict temporal path from one to the other and vice versa.
When the lifetime of the input temporal graph is~$1$, we recover the static problem, which is known to be W[1]-hard with respect to the number of partition classes, pathwidth, and feedback vertex number combined~\cite{enciso2009makes}. We show that, given a TIM decomposition, the problem is in FPT parameterised by TIM width and number of parts combined. Consequently, there exist problems for which Courcelle's Theorem is insufficient to show tractability with respect to TIM width. To distinguish between TIM and VIM width algorithmically, we show that a temporal variant of \textsc{Firefighter} remains NP-hard on graphs of bounded TIM width; resolving an open problem posed by Christodoulou et al.~\cite{christodoulou_making_2024}. This problem was shown to be in FPT with respect to VIM width by Hand et al.~\cite{hand_making_2022}. We leave open whether there is a problem which is in FPT with respect to bidirectional connected-VIM width and remains hard on graphs with bounded TIM width. 

\begin{toappendix}
We now explore distinguishing the parameters in terms of their algorithmic power. To do this, we use a temporal analogue of \textsc{Equitable Connected Partition} which asks for a partition of vertices into at most $h$ classes such that the pairwise difference of the cardinalities of the classes is at most one, and each class induces a connected subgraph. To turn this into a temporal problem, we must define \emph{temporal connectivity}. This is a widely studied property of temporal (sub)graphs~\cite{akrida_complexity_2017,axiotis2016size,balev2024temporally,casteigts2020robustness,christiann2024inefficiently,klobas_complexity_2022}. Here we use nonstrict temporal paths and we say that a temporal graph $\mathcal{G}$ is temporally connected if, for all pairs of vertices $u$, $v\in V(\mathcal{G})$, there is a nonstrict temporal path from $u$ to $v$ and a nonstrict temporal path from $v$ to $u$ in $\mathcal{G}$. We say a subgraph $\mathcal{G'}$ of $\mathcal{G}$ is temporally connected if, for all pairs of vertices $u$, $v\in V(\mathcal{G'})$, there is a nonstrict temporal path from $u$ to $v$ and a nonstrict temporal path from $v$ to $u$ in $\mathcal{G}$. We refer to the pair $u$ and $v$ as \emph{mutually reachable}. Recall that a path on a temporal graph is a nonstrict temporal path if the edges in the path appear at non-decreasing times. We define the problem as follows. 

\decisionproblem{Weak Nonstrict Equitable Temporally Connected Partition (Weak NS ETCP)}{A temporal graph $\mathcal{G}$ and an integer $h$.}{Is there a partition of the vertices into $h$ classes $V_1,\ldots,V_h$ such that for all pairs $i$, $j$, $||V_i|-|V_j||\leq 1$ and each class induces a temporally connected subgraph?}

The word ``weak'' here refers to the fact that we allow the temporal path from one vertex in a part to another to traverse vertices not in that part.

Note that, when the lifetime of the input temporal graph is~$1$, we recover the static problem from \textsc{Weak NS ETCP}. Furthermore, when $\Lambda=1$, the undirected static expansion is the same as the underlying graph of the temporal graph. Therefore, since \textsc{Equitable Connected Partition} is W[1]-hard with respect to treewidth of the input graph and the number of partition classes combined~\cite{enciso2009makes}, \textsc{Weak NS ETCP} is W[1]-hard with respect to treewidth of the undirected static expansion and the number of partition classes combined. In contrast, we show that \textsc{Weak NS ETCP} is in FPT with respect to TIM width and the number of partition classes combined.
To do this, we prove some intermediate lemmas.

\begin{lemmarep}\label{lem:EP-same-bag}
    Let $v$ and $u$ be two vertices in a temporal graph $\mathcal{G}$ which are mutually reachable. Then there exists a bag of any TIM decomposition of $\mathcal{G}$ containing both $v$ and $u$.
\end{lemmarep}
\begin{proof}
    We begin by showing that if there is a nonstrict temporal path from a vertex $v$ to a vertex $u$ in $\mathcal{G}$, there is a corresponding directed path in any TIM decomposition $(T,B,\tau)$ of $\mathcal{G}$ from a node of a bag containing $v$ to a node of a bag containing $u$. Recall that if two vertices are in the same connected component at a given time $t$, then they must be in the same bag labelled with $t$. Now note that we can break any nonstrict temporal path into $\Lambda$ (potentially trivial) paths $P_1,\ldots,P_{\Lambda}$. For each $i\in[\Lambda]$, the path $P_i$ consists of the vertices traversed by $P$ at time $i$. Observe that the endpoint of $P_i$ is the starting vertex in $P_{i+1}$ and each $P_i$ is contained in a connected component of the snapshot $G_i$ for all $i\in \Lambda$. Therefore, each $P_i$ is contained in a bag of the TIM decomposition. Since there is an arc in a TIM decomposition between two bags labelled with consecutive times with non-empty intersection, there must be an arc from the bag containing $P_{i-1}$ to the bag containing $P_i$ for all $i\in(1,\Lambda]$. Thus, if there is a nonstrict temporal path in $\mathcal{G}$, there is a corresponding path in the TIM decomposition from a bag labelled with time 1 to a bag labelled with $\Lambda$.

    Now suppose that $u$ and $v$ are mutually reachable and there is no bag in a TIM decomposition $(T,B,\tau)$ containing both vertices. Then, as before, there must be a path in the TIM decomposition from a bag containing $u$ at time $1$ to a bag containing $v$ at time $\Lambda$. Similarly, there must also be a path in the TIM decomposition from a bag containing $v$ at time $1$ to a bag containing $u$ at time $\Lambda$. By definition of a TIM decomposition, the set of bags containing either vertex must form a directed path. If there is no bag on both of 
    these paths, we have a cycle in the underlying graph indexing the TIM 
    decomposition consisting of the path of all bags containing $u$, the path from the bag containing $u$ at time $1$ to the bag containing $v$ at time $\Lambda$, the path of all bags containing $v$, and the path from the bag containing $v$ at time $1$ to the bag containing $u$ at time $\Lambda$. This contradicts the fact that the underlying graph of a TIM decomposition is a tree.
\end{proof}

\begin{lemmarep}\label{lem:EP-part-bag}
    For any set $S$ of vertices which induces a temporally connected subgraph of a temporal graph $\mathcal{G}$, there is a bag of any TIM decomposition of $\mathcal{G}$ containing $S$.
\end{lemmarep}
\begin{proof}
    We prove this lemma by induction on the size of $S$. Lemma~\ref{lem:EP-same-bag} shows the statement to be true for $|S|=2$; our base case. 
    
    Now suppose that, for all sets $S\subseteq V(\mathcal{G})$ of cardinality $k$ such that all vertices in $S$ are pairwise mutually reachable, there is a bag $B$ of every TIM decomposition of $\mathcal{G}$ such that $S\subseteq B$.

    Now consider a set $S$ of vertices of cardinality $k+1$ such that all vertices in $S$ are pairwise mutually reachable. Assume without loss of generality that $|S|\geq 3$. Let $S'=S\setminus \{v\}$ for some vertex $v$. Let $(T,B,\tau)$ be an arbitrary TIM decomposition of $\mathcal{G}$. Note that, by the inductive hypothesis, there exists a bag $B'$ of $(T,B,\tau)$ containing $S'$.

    We aim to show that there exists a bag containing all of $S$. Assume, for a contradiction that no such bag exists. Let $B_1$ be the first bag (in terms of times with which the bags are labelled) containing both $v$ and a vertex in $S'$, and let $B_2$ be the last bag containing a vertex in $S'$ and $v$. We know such bags exist by Lemma~\ref{lem:EP-same-bag}. If $B_1=B_2$, then $B_1$ must contain all of $S$ and we are done.

    By our assumption, there must be a vertex in $S'$ that is not in each of $B_1$ and $B_2$. Call these vertices $x_1$ and $x_2$ respectively. 
    Note that, since $B', B_1$ and $B_2$ have non-empty pairwise intersections, there must be a path in the TIM decomposition between each of the bags. Since the underlying graph of a TIM decomposition must be a tree, this implies that there is one path $P$ containing all three bags. 
    If $B'$ is between $B_1$ and $B_2$ then, by Observation~\ref{obs:tim-directed-path}, $B'$ must contain $v$ and we are done.

    We now have two cases to consider. In the first, $P$ starts at $B'$, traverses $B_1$ and all other bags containing $v$ at times between those with which $B_1$ and $B_2$ are labelled and finishes at $B_2$. In the second, $P$ starts at $B_1$, traverses all other bags containing $v$ at times between those with which $B_1$ and $B_2$ are labelled and finishes at $B'$. We note that these cases work symmetrically, and continue by showing the first case. 
    By definition $B'$ contains both $x_1$ and $x_2$. Recall that $B_1$ is the earliest bag containing $v$ and a vertex in $S'$, and $B_2$ is the latest such bag. Therefore, there must be a bag between $B_1$ and $B_2$ containing both $v$ and $x_1$. Hence, we have a subpath on the TIM decomposition starting at $B'$ and traversing $B_1$ whose endpoints both contain $x_1$. Then, by Observation~\ref{obs:tim-directed-path}, $x_1$ must be in $B_1$; a contradiction.
    
    Hence, there exists a bag containing all vertices in $S$. Since $(T,B,\tau)$ was an arbitrary TIM decomposition of $\mathcal{G}$ we have shown, by induction, for all sets $S$ which induce a temporally connected subgraph of $\mathcal{G}$ of size $n\in \mathbb{N}$, there is a bag of any TIM decomposition of $\mathcal{G}$ containing $S$.
\end{proof}

Lemma~\ref{lem:EP-part-bag} implies that we can bound the number of vertices in any part in a solution of \textsc{Weak NS ETCP} by the TIM width of the input temporal graph. We now prove that \textsc{Weak NS ETCP} is in FPT with respect to TIM width and number of parts. Recall that this contrasts the fact that we know the problem to be W[1]-hard with respect to the treewidth of the underlying static expansion.
\begin{theoremrep}\label{thm:ECP}
    Given a TIM decomposition of the input graph $\mathcal{G}$ of width $\phi$, \textsc{Weak Nonstrict Equitable Temporally Connected Partition} is solvable in $O(h^{\phi h+4}\phi^4\Lambda^3)$ time.
\end{theoremrep}
\begin{proof}
    We begin by comparing the number of vertices in the input to the product of the TIM width of the input graph and the number of parts allowed. 
    % By Theorem~\ref{thm:tim-alg}, we can calculate TIM width using Algorithm~\ref{alg:TIMW} which runs in $O(n^4\Lambda^2\phi)$ time. 

    By Lemma~\ref{lem:EP-part-bag}, we know that any set of vertices which are pairwise mutually temporally reachable must be contained in a bag of any TIM decomposition. Thus, by the pigeonhole principle, if there exists a weak nonstrict equitable temporally connected partition of at most $h$ parts of a temporal graph $\mathcal{G}$, there must be at most $h\phi$ vertices in $\mathcal{G}$. Hence, the input temporal graph $\mathcal{G}$ has TIM width $\phi$ and more than $h\phi$ vertices, $(\mathcal{G},h)$ is a no-instance of \textsc{Weak NS ETCP}.
    
    The algorithm then finds the set $\mathcal{F}$ of all functions from the set of vertices to integers in $[h]$. This set has cardinality $h^n$. Since $n\leq h\phi$, we get that $|\mathcal{F}|\leq h^{h\phi}$. For each function $f$ in the set $\mathcal{F}$, the algorithm then performs two checks. If there exists a function $f$ in $\mathcal{F}$ which passes the checks, the algorithm returns \textbf{true}. Else, the algorithm returns \textbf{false}.  The first check compares the cardinalities of each pair of parts as prescribed by $f$. This requires $O(\phi h^2)$ time. The second check is that every pair of vertices in each part is mutually temporally reachable by a nonstrict path. Given a pair of vertices, this can be checked in time $O(n^2\Lambda+\Lambda^3)$~\cite{mertzios_temporal_2013}. Therefore, the second check requires $O(n^2(n^2\Lambda+\Lambda^3))=O(n^4\Lambda^3)\leq O(h^4\phi^4\Lambda^3)$ time. Combining these runtimes gives us that the algorithm requires $O(n^4\Lambda^2\phi+h^{\phi h+4}\phi^4\Lambda^3)$ time.

    We now show correctness of the algorithm. Suppose the algorithm returns \textbf{true}. Then there exists a function $f:V(\mathcal{G})\to [h]$ such that, for every pair of vertices $v,u$ with the same image under $f$, $v$ and $u$ are mutually reachable, and there are no two integers in $[h]$ such that the cardinality of their preimages differ by more than~$1$. Thus $f$ describes an equitable partition of the vertices in $\mathcal{G}$ such that each part induces a temporally connected subgraph. Therefore, $(\mathcal{G},h)$ is a yes-instance of \textsc{Weak NS ETCP}.

    Now suppose that the algorithm returns \textbf{false}. 
    Then there does not exist a function $f:V(\mathcal{G})\to [h]$ such that, for every pair of vertices $v,u$ with the same image under $f$, $v$ and $u$ are mutually reachable, and there are no two integers in $[h]$ such that the cardinality of their preimages differ by more than~$1$. Thus there is no equitable partition of the vertices in $\mathcal{G}$ such that each part induces a temporally connected subgraph. Therefore, $(\mathcal{G},h)$ is a no-instance of \textsc{Weak NS ETCP}.
\end{proof}
We now turn to distinguishing the algorithmic power of TIM width to that of VIM width. We show that \textsc{Temporal Firefighter} remains NP-hard on graphs with bounded TIM width and $\geq$-connected-VIM width; resolving an open question posed by Christodoulou et al.~\cite{christodoulou_making_2024} which asks whether \textsc{Temporal Firefighter} is in FPT with respect to $\geq$-connected-VIM width. We show that \textsc{Temporal Firefighter} is NP-Complete even when the TIM width is at most 3 by reduction from the \textsc{Max-2-SAT} variant of the classic \textsc{SAT} problem.

\textsc{Temporal Firefighter} asks how many vertices we can prevent from burning on a graph where, at each time $t$, we can place a defence on a vertex and the fire spreads along all edges active at time $t$ such that one endpoint is burning and the other is not burning or defended. A \emph{strategy} is a sequence of vertices $v_0,\ldots,v_{l}$ such that, at each time $t$, $v_t$ is not burning or defended. The formal definition of the problem is is given as follows.

\begin{nolinenumbers}
\decisionproblem{Temporal Firefighter}{A rooted temporal graph $(\mathcal{G},r)$ and an integer $h$.}{Does there exist a strategy that saves at least $h$ vertices on $\mathcal{G}$ when the fire starts at vertex $r$?}
\end{nolinenumbers}

For an instance $((\mathcal{G},r),h)$ of \textsc{Temporal Firefighter}, we write the instance as $x=(\mathcal{G},\beta)$ where $\beta$ is a string encoding the integer $h$ and which vertex is the root $r$ of the graph.

We use the equivalent problem of \textsc{Temporal Firefighter Reserve} (a temporal analogue of \textsc{Reserve Firefighter} defined by Fomin et al.~\cite{fomin_firefighter_2016}). In \textsc{Temporal Firefighter Reserve}, we are not required to make a defence at each time, rather our budget is incremented by one and we can simultaneously defend at most as many vertices as the value of the budget at each time. This allows us to only consider strategies that defend temporally adjacent to the fire. Furthermore we assume that we are given an instance $(\mathcal{G},\beta)$ of \textsc{Temporal Firefighter Reserve} such that $r$ has an incident edge active on timestep $1$. Note that if we are given an instance where this is not the case, we could take the earliest timestep at which $r$ has an active incident edge to be timestep $1$, and increase the starting budget according to the number of omitted timesteps, as the fire cannot leave $r$ before its first incident edge is active.

Satisfiability problems ask whether there is a truth assignment to the variables of a Boolean formula such that satisfies a certain requirement. Formulas are usually given in conjunctive normal form (abbreviated to CNF), where the formula is a conjunction of clauses: disjunctions of literals.

\textsc{Max-2-Sat} asks us to determine if a given number of clauses can be satisfied in a CNF formula, in which each clause contains two literals. This problem was shown to be \textbf{NP}-Complete by Garey et al.~\cite{garey_simplified_1976}.

\begin{nolinenumbers}
\decisionproblem{Max-2-SAT}{An integer $k$, and a pair $(B, C)$ where $B$ is a set of Boolean variables, and $C$ is a set of clauses over $B$ in CNF, each containing 2 variables.}{Is there a truth assignment to the variables such that at least $h$ clauses in $C$ are satisfied?}
\end{nolinenumbers}

\begin{theorem}[Garey et al. \cite{garey_simplified_1976}]
    \textsc{Max-2-SAT} is \textbf{NP}-Complete.
\end{theorem}

We are now ready to give our reduction. Given a CNF formula in which each clause has 2 literals, we produce a temporal graph in which the underlying graph is a tree of depth 2, and each edge is active exactly once, and at most two edges are active on every timestep. The fire begins at a root vertex $r$, and every vertex adjacent to the root corresponds to a literal from the formula. We attach leaves to these literal vertices and assign times to their incident edges such that the firefighters are forced to defend exactly one of each pair of literal vertices corresponding to a variable. Such a defence then corresponds to a truth assignment for the variables in the formula. We construct our instance such that a defence saves the desired number of vertices in \textsc{Temporal Firefighter} if and only if the corresponding truth assignment satisfies the desired number of clauses.

\begin{restatable}{theorem}{firehardthm}\label{thm:firefighter-hard-appendix}
\textsc{Temporal Firefighter} is NP-Complete even when restricted to the class of temporal trees with each edge active exactly once, and at most two edges active per timestep.
\end{restatable}
\begin{proof}
We reduce from \textsc{Max-2-SAT}. Given an instance $((B, C), h)$ of \textsc{Max-2-SAT} we construct an instance $(((G, \lambda), r), h')$ of \textsc{Temporal Firefighter} where $G$ is a tree, each edge is active exactly once, and there are at most two edges active per timestep, such that $(((G, \lambda), r), h')$ is a yes-instance if and only if $((B, C), h)$ is also a yes-instance.

Let $v=|B|$, the number of variables, and $w=|C|$, the number of clauses. Our vertex $V(G)$ set consists of $1+2v+2wv+4w$ vertices:
\begin{itemize}
    \item one root vertex $r$,
    \item $2v$ variable vertices $\{b_{i,x} : i \in [v], x \in \{1,0\}\}$,
    \item $2wv$ forcing leaf vertices $\{d_{i,x,j} : i \in [v], x \in \{1,0\}, j \in [w]\}$,
    \item $4w$ clause leaves, two for each appearance of a literal in a clause, $\{c_{j,i}, \bar{c}_{j,i} : i \in [v], j \in [w], b_i \text{ appears in clause } c_j\}$.
\end{itemize}

Our set of time edges then connects every variable vertex to the root, and every forcing and clause leaf to a variable vertex:
\begin{align*}
\{(e, t) : e \in E(\mathcal{G}), t \in \lambda(e) \}&= \{(\{b_{i,x},r\},i) : i \in [v], x \in \{1,0\}\}\\
    &\cup \{(\{d_{i,x,j}, b_{i,x}\}, v+(i-1)w+j) : i \in [v], x \in \{1,0\}, j \in [w]\}\\
    &\cup \{(\{c_{j,i},b_{i,1}\}, v+wv+j), (\{\bar{c}_{j,i},b_{i,0}\}, v+wv+w+j)\\
    &: i \in [v], j \in [w], b_i \text{ occurs positively in }c_j\}\\
    &\cup \{(\{c_{j,i},b_{i,0}\}, v+wv+j), (\{\bar{c}_{j,i},b_{i,1}\}, v+wv+w+j)\\
    &: i \in [v], j \in [w], b_i \text{ occurs negatively in }c_j\}
\end{align*}

\begin{figure}[ht]
    \centering
    \includegraphics[width=0.8\linewidth,page=5]{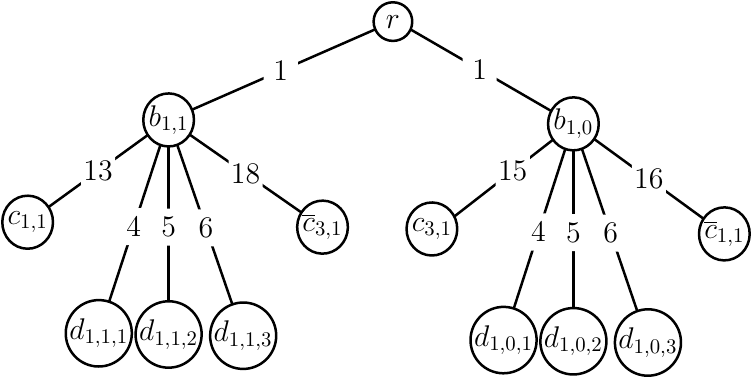}
    \caption{The section of the tree corresponding to the appearances of variable $b_1$ in the  \textsc{Max-2-SAT} instance $(b_1 \lor b_2) \land (\neg b_2 \lor b_3) \land (\neg b_1 \lor \neg b_3)$}
     \label{fig:reduction}
\end{figure}

Now, set $k'=(1+2v+2wv+4w)-(1+v+(w-k))=v+2wv+3w+k$, and this along with the above temporal graph is our instance $(((G, \lambda), r), k')$. As required, $G$ is a tree, each edge is active on exactly one timestep, and there are two edges active on every timestep. For any timestep between $1$ and $v$ inclusive, both of these edges are between the root and a variable vertex. For any timestep between $v+1$ and $v+vw$ inclusive, these edges are between a variable vertex and a forcing leaf. For any timestep between $v+vw+1$ and $v+vw+w$ inclusive these edges are between a variable vertex and a positive clause vertex. And, for any timestep between $v+wv+w+1$ and $v+wv+2w$ inclusive these edges are between a variable vertex and a negative clause vertex. An example construction of such a graph can be seen in \cref{fig:reduction}.

Now assume that $((B, C), k)$ is a yes-instance, that is that there is a truth assignment $\phi : B \to \{T, F\}$ to the variables in $B$ such that at least $k$ of the clauses in $C$ are satisfied. Given this truth assignment we then define a strategy $\sigma$, and show it to save $k'=v(2w+1)+3w+k$ vertices on $((G, \lambda), r)$, thus demonstrating that $(((G, \lambda), r), k')$ is also a yes-instance. This strategy defends as follows:

\begin{defn}[Strategy $\sigma$]
\begin{itemize}
    \item[]
    \item For each timestep $t \in [v]$, $\sigma$, if $\phi(b_t) = \textbf{true}$ then $\sigma$ defends $b_{t,1}$, and if $\phi(b_t) = \textbf{false}$ then $\sigma$ defends $b_{t,0}$,
    \item for each timestep $t \in [v+1, v+vw]$, $\sigma$ defends $d_{\lceil\frac{t-v}{w}\rceil, 0, ((t-v-1) \mod w)+1}$ if $\phi(b_{\lceil\frac{t-v}{w}\rceil})=\textbf{true}$, and $d_{\lceil\frac{t-v}{w}\rceil, 1, ((t-v-1) \mod w)+1}$ if $\phi(b_{\lceil\frac{t-v}{w}\rceil})=\textbf{false}$,
    \item for each timestep $t \in [v+wv+1, v+wv+w]$, $\sigma$ defends any clause leaf in  $\{c_{t-(v+wv),i} : b_i \text{ occurs in }c_{t-(v+wv)}\}$ that has an undefended parent. If neither of these two clause leaves have an undefended parent, then $\sigma$ defends a clause leaf in $\{\bar{c}_{t-(v+wv),i} : b_i \text{ occurs in }c_{t-(v+wv)}\}$,
    \item finally, for each timestep $t \in [v+wv+w+1,v+wv+2w]$, $\sigma$ defends any clause leaf in $\{\bar{c}_{t-(v+wv+w),i} : b_i \text{ occurs in }c_{t-(v+wv+w)}\}$ that has an undefended parent. If neither of these two leaves have an undefended parent, then $\sigma$ defends one of them arbitrarily.
\end{itemize}
\end{defn}

Now consider the number of vertices that burn under $\sigma$. To begin with, the root and half of the variable vertices burn before all of the forcing leaves are saved. Now consider some clause $c_j \in C$ containing variables $b_x$ and $b_y$. If $c_j$ is satisfied, then neither clause leaf $c_{j,x}$ or $c_{j,y}$ burns, as at least one of these leaves will have a defended parent, and if either leaf does not have a defended parent, the leaf will be defended on timestep $v+wv+j$. If $c_j$ is not satisfied, then neither of $c_{j,x}$ and $c_{j,y}$ will have a defended parent, and one of them will burn, and the other will be defended on timestep $v+wv+j$.

Finally consider a pair of negative clause leaves $\bar{c}_{j,x}$ and $\bar{c}_{j,y}$. If the parents of both of these leaves burn, the neither of the parents of the corresponding leaves $c_{j,x}$ and $c_{j,y}$ burn, and one of $\bar{c}_{j,x}$ and $\bar{c}_{j,y}$ will be defended on timestep $v+wv+j$, and the other on timestep $v+wv+w+j$. If one or fewer of the parents burn, then either of the leaves with a burning parent will be defended on timestep $v+wv+w+j$. Therefore no negative clause leaf $\bar{c}_{j,i}$ will burn. Thus in total the root, half of the variable vertices, and one clause leaf per unsatisfied clause burn, that being at most $1+v+(w-k)$ vertices. This means that at least $(1+2v+2wv+4w)-(1+v+(w-k)) = v+2wv+3w+k$ vertices are saved as required.

We now show that if $(((G, \lambda), r), k')$ is a yes-instance, that is that there is some strategy $\sigma$ that saves $v+2wv+3w+k$ vertices, then $((B, C), k)$ is also a yes-instance. We begin by showing that if there exists a strategy that saves $k'$ vertices, then there exists a strategy that on every timestep $t \in [v]$ defends one of the vertices $b_{t,0}$ and $b_{t,1}$, and also saves $k'$ vertices. Given a strategy $\sigma$ with this property, we then define a truth assignment $\phi$, such that $\phi(b_t)=\textbf{true}$ if $\sigma$ defends $b_{t,1}$ on timestep $t$, and $\phi(b_t)=\textbf{false}$ if $\sigma$ defends $b_{t,0}$ on timestep $t$.

First assume that there exists a strategy that saves $k'$ vertices, but no strategy that does so by defending only variable vertices on every timestep $t \in [v]$. Now let $\sigma$ be a strategy that saves $k'$ vertices and is maximal in the number of timesteps $t \in [v]$ on which a variable vertex is defended. Let $t$ be the earliest timestep on which $\sigma$ defends a leaf vertex $l$, and let $\{b_{i,0}, b_{i,1}\}$ be a pair of variable vertices undefended by $\sigma$, noting that such a pair must exist - if $\sigma$ defends at least one vertex from every pair of variable vertices, then it must do so by timestep $v$ at the latest, by which time every variable vertex burns. Furthermore, if it defends at least one vertex from every pair of variable vertices, then this requires at least $v$ defences, and so $\sigma$ would only defend variable vertices on every timestep $t \in [v]$, contradicting its definition. Consider now the strategy $\sigma'$ that defends $b_{i,0}$ on timestep $t$, and makes the same defences as $\sigma$ otherwise. See that $\sigma'$ saves at least all the vertices saved by $\sigma$, with the possible exception of $l$, and also saves $b_{i,0}$, which was not saved by $\sigma$. Therefore $\sigma'$ also saves at least $k'$ vertices and contradicts the maximality of $\sigma$, and so there always exists a strategy that only defends variable vertices for every timestep $t \in [v]$.

We now show by induction on the variable index $i$ that any strategy $\sigma$ that saves $k'$ vertices and defends only variable vertices on the first $v$ timesteps must defend exactly one of every pair of vertices $b_{i,0}$ and $b_{i,1}$ on timestep $i$. When $i = 1$ if neither of $b_{i,0}$ and $b_{i,1}$ are defended both of these vertices will burn. There are then $2w$ forcing leaf vertices adjacent to these variable vertices, and all of these forcing leaf vertices will burn by timestep $v+w$. Our strategy only defends variable vertices for the first $v$ timesteps, so can then only defend at most $w$ of the forcing leaf vertices by timestep $v+w$, meaning at least $w$ forcing vertices burn. Even assuming the remaining vertices in the graph are saved, the root, $v$ of the variable vertices, and $w$ forcing vertices burn, meaning that the number of saved vertices is $(1+2v+2wv+4w)-(1+v+w)=v+2wv+3w < v+2wv+3w+k$, contradicting the assumption that $\sigma$ saves $k'$ vertices. Therefore $\sigma$ must defend exactly one of $b_{1,0}$ and $b_{1,1}$ on timestep $1$. Otherwise, if $i > 1$ then by the inductive assumption $\sigma$ defends one of each pair of vertices $b_{t,0}$ and $b_{t,1}$ on every timestep $t < i$. Assume that $\sigma$ does not defend either of $b_{i,0}$ or $b_{i,1}$ on timestep $i$, so both of these vertices will burn. There are then $iw+w$ forcing leaf vertices adjacent to the burning variable vertices $b_{t,0}, b_{t,1}$ with $t \leq i$, and all of these forcing leaf vertices will burn by timestep $v+iw$. Our strategy only defends variable vertices for the first $v$ timesteps, so can then only defend at most $iw$ of the forcing leaf vertices by timestep $v+iw$, meaning that at least $w$ forcing leaves burn. Even assuming the remaining vertices in the graph are saved, the root $v$ of the variable vertices, and $w$ forcing vertices burn, meaning that the number of saved vertices is $(1+2v+2wv+4w)-(1+v+w)=v+2wv+3w < v+2wv+3w+k$, contradicting the assumption that $\sigma$ saves $k'$ vertices. Therefore $\sigma$ must defend one of $b_{i,0}$ and $b_{i,1}$ on timestep $i$.

Therefore, if there exists a strategy that saves $k'$ vertices, there must exist a strategy that only defends variable vertices during the first $v$ timesteps, and this strategy must defend exactly one of each pair of variable vertices $b_{i,0}$ and $b_{i,1}$. Thus, given such a strategy $\sigma$, we then define a truth assignment $\phi$, such that $\phi(b_i)=\textbf{true}$ if $\sigma$ defends $b_{i,1}$ on timestep $i$, and $\phi(b_i)=\textbf{false}$ if $\sigma$ defends $b_{i,0}$ on timestep $i$.

Now assume that there exists a strategy that defends one of the vertices $b_{i,0}$ or $b_{i,1}$ on each timestep $i \leq v$, and saves at least $k'$ vertices, but no strategy that saves at least $k'$ vertices, defends either $b_{i,0}$ or $b_{i,1}$ on each timestep $i \leq v$, and saves every forcing leaf. Let $\sigma$ be a strategy that saves at least $k'$ vertices, defends one of the vertices $b_{i,0}$ or $b_{i,1}$ on each timestep $i \leq v$, and is maximal in the integer $\ell$ such that every defence made by $\sigma$ on a timestep $v < t \leq v + \ell$ is made at a forcing leaf with a burning parent and an active incident edge active on timestep $t$. Consider the strategy $\sigma'$ which defends as $\sigma$ but on timestep $v+\ell+1$ defends a forcing leaf with an incident edge active on timestep $v + \ell$. Note that such a leaf must exist, as for every timestep $v < t \leq v+vw$, there are two forcing leaves with incident edges active on $t$, and one of these leaves is the child of a variable vertex $b_{i,0}$, and the other the child of $b_{i,1}$, only one of which will be defended by $\sigma$. Also see that such a leaf cannot have burnt at the start of timestep $v + \ell$, as its only incident edge is only active on this timestep. See then that any leaf that does not burn and is not defended when $\sigma$ is played also does not burn when $\sigma'$ is played, as $\sigma'$ defends the same non-leaf vertices as $\sigma$. Furthermore, the number of leaves that are defended is the same when $\sigma$ is played is the same as the number of leaves that are defended when $\sigma'$ is played, and therefore $\sigma'$ saves the same number of vertices as $\sigma$. This contradicts the maximality of $\sigma$, and therefore if there exists a strategy that saves $k'$ vertices, there exists a strategy that saves $k'$ vertices, defends one of the vertices $b_{i,0}$ or $b_{i,1}$ on each timestep $i \leq v$, and saves every forcing leaf.

When such a strategy $\sigma$ is played, $vw$ forcing leaves will have burning parents, and each of these leaves will burn by timestep $v+vw$ if undefended, meaning that on every timestep $v < t \leq v+vw$, $\sigma$ will defend a forcing leaf. Furthermore, $\sigma$ saves $v$ variable vertices, $2vw$ forcing leaves, and must therefore save at least $3w+k$ clause leaves, as $\sigma$ saves $k'=v+2wv+3w+k$ vertices. There are $2w$ negative clause leaves, and $2w$ positive clause leaves, meaning that at least $w+k$ positive clause leaves must be saved by $\sigma$. Every undefended positive clause leaf with a burning parent will burn by timestep $v+vw+w$, and $\sigma$ can only defend clause leaves from timestep $v+vw+1$ onwards. Therefore at most $w$ positive clause leaves can be saved by being defended. The remaining required $k$ positive clause leaves must therefore have parents defended by $\sigma$. If the parent $b_{i,x}$ of any positive clause leaf $c_{j,i}$ is defended, then clause $c_j$ must be satisfied by our truth assignment. Therefore at least $k$ clauses are satisfied by the truth assignment corresponding to $\sigma$ as required.
\end{proof}

Let $\mathcal{G}'$ be the temporal graph constructed from $\mathcal{G}$ by adding the time-edge $(uv,t)$ for all vertices $u,v$ and all times $t$ such that there exist times $t_1<t<t_2$ such that $(uv,t_1)$ and $(uv,t_2)$ are in $\mathcal{E}(\mathcal{G})$.
\begin{restatable}{theorem}{orderedtreetim}\label{thm:ordered-tree-low-timw}
    Let $\mathcal{G}$ be a temporal tree. Suppose we can choose a root $r$ of $\mathcal{G}$ such that, for every vertex $v\in V(\mathcal{G})$, the edges incident to $v$ are active strictly before all other edges in the subtree rooted at $v$. Then there exists a TIM decomposition of $\mathcal{G}$ with width $\max_t(\max_v\deg_{G'_t}(v))+1$. 
\end{restatable}
\begin{proof}
    We claim that the TIM decomposition of $\mathcal{G}$ is the same as the TIM decomposition of $\mathcal{G}'$. The TIM decomposition of $\mathcal{G}'$ can be found by simply putting each connected component of each snapshot $G'_t$ of $\mathcal{G}'$ into a separate bag of $(T,B,\tau)$. By definition, arcs are added from bags at time $t$ to bags at time $t+1$ with non-empty intersection for all times $t\in [1,\Lambda)$. By construction of $\mathcal{G}$ and $\mathcal{G}'$, this cannot contain any cycles. 
    
    To see this note that, since edges incident to a vertex $v$ are active at times strictly before the edges in the subtree rooted at $v$ are active, the connected component containing $v$ in any snapshot of $\mathcal{G}'$ must be a star with $v$ in the centre. Suppose there are two paths from a bag $B(i)$ to another bag $B(j)$ in the TIM decomposition at this point in the construction (i.e. a cycle). This implies that either an edge is active multiple times with a gap between these appearances, or that $\mathcal{G}$ is not a tree. Both give us a contradiction. The first cannot be true by construction of $\mathcal{G}'$, and the second contradicts the structure of $\mathcal{G}$ (and $\mathcal{G}'$). Therefore, there are no cycles at in this decomposition, and it is a minimum TIM decomposition of $\mathcal{G}'$. 
    
    Note that, when constructing a decomposition of $\mathcal{G}$ in the same way, there are only cycles caused by edges active multiple times. Removing these cycles results in both endpoints of the edge being in the same bags at all times between the first and last time the edge is active. This edge appears at all such intermediate times of $\mathcal{G}'$, therefore both temporal graphs have the same TIM decomposition.
    
    Observe that since the underlying graph of $\mathcal{G}$ is a tree, and the edges closer to the root occur earlier, each vertex is in a singleton bag before and after the edges it is incident to are active. Further, every bag $B(i)$ of the TIM decomposition of $\mathcal{G}$ consists only of a vertex $v$ and those of its children which neighbour $v$ in the snapshot of $\mathcal{G}'$ at time $t$.
    Since the bags of the TIM decomposition of $\mathcal{G}$ consist of connected components in $\mathcal{G}'$, the decomposition has width $\max_t(\max_v\deg_{G_t'}(v))+1$.
\end{proof}
In fact, we can prove a stronger claim: that the $\geq$- (and, thus bidirectional) connected-VIM width is bounded in such temporal graphs. Recall that $\mathcal{G}'$ is the temporal graph constructed from $\mathcal{G}$ by adding the time-edge $(uv,t)$ for all vertices $u,v$ and all times $t$ such that there exist times $t_1<t<t_2$ such that $(uv,t_1)$ and $(uv,t_2)$ are in $\mathcal{E}(\mathcal{G})$.

\begin{restatable}{theorem}{orderedtreecvim}\label{thm:ordered-tree-low-cvim}
    Let $\mathcal{G}$ be a temporal tree. Suppose we can choose a root $r$ of $\mathcal{G}$ such that, for every vertex $v\in V(\mathcal{G})$, the edges incident to $v$ are active strictly before all other edges in the subtree rooted at $v$. Then the $\geq$-connected-VIM width of $\mathcal{G}$ is $\max_t(\max_v\deg_{G'_t}(v))+1$. 
\end{restatable}
\begin{proof}
    We begin by recalling the definition of $\geq$-connected-VIM width (Definition~\ref{def:CVIM}). Recall that $G_{\geq}(t)$ is the underlying graph of the temporal graph $\mathcal{G}_{\geq}(t)$ consisting of the vertices in $\mathcal{G}$ and the time-edges that appear in $\mathcal{G}$ at time at least $t$. The $\geq$-connected-VIM width of a temporal graph $\mathcal{G}$ is the maximum cardinality of a bag, where the bags are found by taking, for each time $t$ and connected component $C$ of $G_{\geq}(t)$, the intersection of $F_t$ and $V(C)$, where $F_t$ is the bag at time $t$ of the VIM sequence of $\mathcal{G}$.

    Note that, for any $t$, the connected components of $G_{\geq}(t)$ are the subtrees rooted at the set of vertices $v$ such that the edge from $v$ to its parent (if it exists) is active only at times $t'<t$ and there exists a time-edge from $v$ to a child (if there is a child) of $v$ active at $t''\geq t$. Denote by $R^t$ the set of the roots of such subtrees at time $t$.
    
    By definition of the VIM sequence, for every vertex $u$, the bags of the VIM sequence containing $u$ are those at time $t$ such that there is a time-edge incident to $u$ at time $t'\leq t$ and a time-edge incident to $u$ at time $t''\geq t$. By construction of $\mathcal{G}$, the bags of the VIM sequence containing $u$ must be those with times at or before the latest time the edge from $u$ to its parent and at or after the earliest time-edge from $u$ to a child of $u$. 
    
    We claim that, for any time $t$, the intersection of $F_t$ and any connected component of $G_{\geq}(t)$ must be a subset of $N_{G_{\geq}}[r']$ for some $r'\in R^t$. Note that the closed neighbourhood of a vertex $r'$ in a tree must be a star; a graph consisting of one central vertex and leaves adjacent to it. Further observe that any vertices in a star are at distance at most~2 from each other.
    Suppose for a contradiction that there is a connected component $C$ of $G_{\geq}(t)$ such that there are two vertices $v$ and $u$ in $C\cap F_t$ that are distance greater than two from each other. Since $u$ and $v$ are both in $F_t$, they must both be incident to at least one time-edge active at or before $t$ and at least one time-edge active at or after $t$. Let $w$ be the vertex traversed on the path from $u$ to $v$ closest to the root $r'$ of the subtree. By our assumption that the distance from $u$ to $v$ is at least~3, $w$ must be distance at least~2 from one of $u$ and $v$. Assume without loss of generality that $u$ is distance at least~2 from $w$, and recall that $w$ is an ancestor of both $v$ and $u$. By construction of $\mathcal{G}$, all edges incident to $w$ must be active strictly before any edges incident to $u$. Therefore, $w$ cannot be incident to any time-edges at or after $t$ and $w$ is isolated in $G_{\geq}(t)$; a contradiction. Therefore, all vertices in $C\cap F_t$ are at distance at most~2 from one another. 
    
    By similar reasoning, we see that if two vertices are in $C\cap F_t$ and they are distance two from each other, they must be siblings. Suppose for a contradiction that $v$, $u$ are in $C\cap F_t$ and $uv\in E(G_{\geq}(t))$. That is, one of $v$ and $u$ is a grandparent of the other. Assume without loss of generality that $v$ is the grandparent of $u$. By construction of $\mathcal{G}$, the edge between $v$ and its child occurs strictly before any edges incident to $u$. Therefore, $u$ and $v$ cannot be in the same bag of the VIM sequence at any time.
    
    Thus, there must exist a vertex $x$ for which all vertices in $C\cap F_t$ are in the closed neighbourhood $N_{G_{\geq}}[x]$. Since all other vertices in $C\cap F_t$ must be siblings of each other, $x$ must be such that the edge from $x$ to its parent is active only at times $t'<t$ and there exists a time-edge from $v$ to a child of $v$ active at $t''\geq t$. Hence, $x\in R^t$.

    We note that, for a vertex $x$ such that $C\cap F_t\subseteq N_{G_{\geq}}[x]$, the vertices neighbouring $x$ that are in $C\cap F_t$ must be children $x_c$ of $x$ such that the edge $xx_c$ is active at times $t_1$ and $t_2$ such that $t_1\leq t\leq t_2$. These are precisely the neighbours of $x$ in the snapshot of $\mathcal{G}'$ at time $t$. Thus, the $\geq$-connected-VIM width of $\mathcal{G}$ is $\max_t(\max_v\deg_{G'_t}(v))+1$, and the result holds.
\end{proof}
In the graph in the construction used in the reduction for Theorem~\ref{thm:firefighter-hard-appendix}, each edge is active exactly once and the edges incident to a vertex occur strictly before any edges in the subtree rooted at that vertex. Thus, the TIM width of this graph is the maximum of the maximum degrees of the snapshots of $\mathcal{G}$ plus 1. This gives the following result; in particular, it resolves the open question posed by Christodoulou et al.~\cite{christodoulou_making_2024} which asks whether \textsc{Temporal Firefighter} is in FPT with respect to $\geq$-connected-VIM width.

\begin{restatable}{theorem}{FirefighterHard}\label{thm:firefighter-hard}
         \textsc{Temporal Firefighter} remains NP-complete even on temporal graphs whose TIM width and $\geq$-connected VIM width are at most 3.
 \end{restatable}
\end{toappendix}
\section{Meta-algorithms}\label{sec:meta}
Here we introduce two meta-algorithms for temporal problems. These algorithms rely on a few efficient checks to ensure both soundness of the states in the dynamic program, and that we can transition through consecutive states.  While the existence of an fpt-algorithm for a given problem parameterised by TIM width implies the existence of an fpt-algorithm parameterised by VIM width, there are two advantages to providing meta-algorithms for both parameters: firstly, if we are only interested in parameterising by VIM width, the specialised VIM width meta-algorithm will typically give a better running-time bound; secondly, there exist problems (e.g.~\textsc{Temporal Firefighter}) to which we can apply our VIM width meta-algorithm but which are intractable with respect to TIM width. Throughout this section, we will illustrate the definitions given by discussing their application to \textsc{Temporal Hamiltonian Path}, a temporal analogue of the classic \textsc{Hamiltonian Path} problem.

 % \begin{nolinenumbers}
\decisionproblem{Temporal Hamiltonian Path}{A temporal graph $\mathcal{G}$.}{Does there exist a strict temporal path containing every vertex in $\mathcal{G}$?}
 % \end{nolinenumbers}
\subsection{Meta-algorithm parameterised by VIM width}\label{sec:Sam-chapter}
We begin with the more intuitive of the two algorithms, the meta-algorithm parameterised by VIM width. Since the majority of algorithms parameterised by VIM width take the form of dynamic programs over the VIM sequence, we generalise this method and use it to define a large family of problems which we can solve in this manner. 
% 
% \subsection{Meta-algorithm parameterised by VIM width} \label{sec:Sam-chapter}
% 
% We give necessary and sufficient conditions for a graph problem to admit an fpt-algorithm when parameterised by VIM width. 
Informally, we require that
\begin{itemize}
    \item we can model the problem with vertex labels and counters,
    \item that we can check whether sets of labels at consecutive timesteps are compatible,
    \item that vertices can only change labels when incident at an active edge, and 
    \item that we can efficiently generate or identify starting and ending sets of labels that give a yes-instance.
\end{itemize}   
We will call any problem that satisfies these conditions \emph{locally temporally uniform}. We prove that any locally temporally uniform problem for which all required subroutines run in time fpt with respect to VIM width is in fact in FPT parameterised by VIM width; it turns out that the converse is also true, so we completely characterise the temporal graph problems belonging to FPT parameterised by VIM width.

The temporal graph problems we consider can be expressed in terms of labellings on the vertices of the input graph that change over time. We define a state of a vertex set as a labelling of the vertices with labels from a set $X$ and a $k$-length vector of integer counters.
% We define a set of conditions, and refer to any problem that obeys them as locally temporally uniform. 
% We then provide an algorithm capable of these problems. 
% Thus, in order to obtain an fpt-algorithm using VIM width for a given problem, we need only show that it is locally temporally uniform. In Section~\ref{sec:applications}, we give an example of using this method for \textsc{Temporal Hamiltonian Path}. 
% % Applications to 
% % \textsc{Temporal Dominating Set} and
% An application to
% \textsc{Temporal Firefighter} can be found in the Appendix.

\begin{defn}[$(X,k)$-State]
    A $(X,k)$-state on a vertex set $V$ is a pair $(l, c)$, where $l : V \to X$ is a labelling of the vertices in $V$ using the labels from set $X$, and $c$ is a vector containing $k$ integers each of size at most a polynomial of $|V|$.
\end{defn}

In our \textsc{Temporal Hamiltonian Path} example, we use $(X,1)$-states, where the label set $X=\{\visited,\unvisited,\current\}$, and the counter vector contains a single integer $h$, which counts the total number of visited vertices.

In order to obtain tractability with respect to VIM width, we consider problems that can be expressed in terms of sequences of states over the VIM sequence.
% In order to obtain tractability with respect to VIM width, we specifically consider problems that can be expressed in terms of sequences of states on the vertices in the VIM sequence.
% Also, it is simpler in most cases to define starting states on a snapshot containing no edges. 
We first define temporally uniform problems, for which there exists a transition routine that, when given two states, returns true if the second state can follow the first. We then provide a definition for \textit{locally} temporally uniform problems that further requires that one state can follow from another only if
the labels on vertices outside their active interval do not change. Let $A_t$ denote the set of vertices with incident edges active on timestep $t$, and note that $A_t \subseteq F_t$ (recall that $F_t$ is a bag in the VIM sequence). 
For technical convenience, we let $F_0 = F_1$, and $A_0=A_1$. Otherwise, the transition routine would never be applied to edges in the first snapshot of the input temporal graph.  Throughout this section we define an instance $x$ of a temporal problem as a pair $(\mathcal{G},\beta)$, where $\mathcal{G}$ is a temporal graph and $\beta$ is a string encoding the remaining input of a problem instance. For many problems, $\beta$ will simply encode a set of integers. For \textsc{Temporal Hamiltonian Path}, $\beta$ is an empty string.
% using a VIM sequence.

\begin{defn}[$(X,k,f_1,f_2)$-Temporally Uniform Problem]
    We say that a decision problem $P$ which takes an input instance $x=(\mathcal{G},\beta)$ is $(X,k,f_1,f_2)$-temporally uniform if and only if there exist:
    \begin{enumerate}
        \item a transition algorithm \textbf{Tr} that takes a static graph $G$, two $(X,k)$-states for the vertices of this graph, and the string $\beta$, runs in time at most $f_1(G,\beta)$, and returns \textbf{true} or \textbf{false},
        \item an accepting algorithm \textbf{Ac} that takes a $(X,k)$-state and $x$, runs in time $f_2(x)$, and returns \textbf{true} or \textbf{false}, and
        \item a starting algorithm \textbf{St} that takes an instance $x$, runs in time $f_2(x)$, and returns a set of initial $(X,k)$-states $S_{0,x}$ for the instance,
    \end{enumerate}

    \noindent such that $x$ is a yes instance of $P$ if and only if there exists a sequence $s_0, ..., s_\Lambda$ of $(X,k)$-states with $s_0 \in S_{0,x}$, $\textbf{Tr}(s_{t-1}, s_t, G_t,\beta)=\textbf{true}$ for all timesteps $1 \leq t \leq \Lambda$, and $\textbf{Ac}(s_\Lambda, x)=\textbf{true}$.
\end{defn}

Continuing with our example of \textsc{Temporal Hamiltonian Path}, our starting routine returns the set containing one state $s_0$ where all vertices are marked $\unvisited$ and $h=0$. We only start counting vertices traversed by a path once there is at least one time-edge in the path. That is, for pairs of states $s_t\neq s_{t+1}$ where $s_t=s_0$, the transition routine returns true if $s_{t+1}$ labels one vertex as $\current$, one as $\visited$ and has $h=2$. For other pairs of consecutive states $s_t,s_{t+1}$, our transition routine ensures that either the current location (and the state) stays the same between timesteps ($s_t=s_{t+1}$), or that there is an active edge between the vertices labelled $\current$ by $s_t$ and $s_{t+1}$, that the vertex labelled $\current$ in $s_t$ is labelled as $\visited$ in $s_{t+1}$, that the vertex labelled $\current$ in the $s_{t+1}$ is labelled as $\unvisited$ in $s_t$, and the counter is incremented. Pseudocode of the transition routine is given in Algorithm~\ref{alg:temp-ham}. The combination of the starting and transition routines ensure that there is exactly one current location at any given time. Finally, the acceptance routine returns true if and only if the counter is equal to $n$.

\begin{algorithm}[ht]\caption{\textsc{Temporal Hamiltonian Path Transition}}\label{alg:temp-ham}
    \begin{algorithmic}[1]
    \Require A static graph $G$ and states $(l_1, (h_1))$ and $(l_2, (h_2))$ for $V(\mathcal{G})$.
    \Ensure Returns true when $(l_2, (h_2))$ corresponds to a path that consists of a single time-edge, or traverses zero or one further vertices than the path corresponding to $(l_1, (h_1))$ and false otherwise.
    \State{Let $U_1$ and $U_2$ be the set of vertices labelled $\unvisited$ by $l_1$ and $l_2$ respectively, and equivalently for $V_1$, $V_2$, and $C_1$, $C_2$.}
    \If{$C_2 \setminus C_1$ contains a single vertex $c_2$}
        \If{$C_1\setminus C_2=\emptyset$}
            % \If{there exists $v$ such that $\{v,c_2\}\in E(G)$ and $h_2=1$}
            % \State{\Return{True}}
            % \Els
            \If{$U_1\setminus U_2=\{u,c_2\}$, $V_1\cup \{u\}=V_2$, $\{v_2,c_2\}\in E(G)$, $h_1=0$ and $h_2=2$}
            \State{\Return{True}}
            \EndIf
        \ElsIf{$C_1 \setminus C_2$ contains a single vertex $c_1$}
            \If{$\{c_1, c_2\} \in E(G)$, $c_2 \in U_1$, $h_2=h_1+1$, and $V_1 \cup \{c_1\} = V_2$}
                \State{\Return{True}}
            \EndIf
        \EndIf
    \EndIf
    \If{$(l_1, (h_1)) = (l_2, (h_2))$}
        \State{\Return{True}}
    \EndIf
    \State{\Return{False}}
    \end{algorithmic}
\end{algorithm}

We say that two states $s = (l, c)$ and $s' = (l', c')$ for vertex set $V$ \emph{agree} on a vertex set $W$ if and only if $W \subseteq V$, $l$ and $l'$ give the same label to every vertex in $W$, and $c = c'$.  We can now define locally temporally uniform problems. Locality refers to the fact that we restrict temporally uniform problems such that vertices outside of their active interval cannot change label. Recall that the set $A_0$ is the set of vertices in their active interval at time~1.

\begin{defn}[$(X,k,f_1,f_2)$-Locally Temporally Uniform Problem]\label{def:loctempunf}
    We say that a $(X,k,f_1,f_2)$-temporally uniform problem $P$ is $(X,k,f_1,f_2)$-locally temporally uniform if and only if for any temporal graph $\mathcal{G}$, and instance of the problem $x=(\mathcal{G},\beta)$:
    \begin{enumerate}
        \item There exists a label $U$ such that in every initial state $s_0 \in S_{0,x}$, all vertices not in $A_0$ are labelled $U$.
        \item For any pair of states $s$ and $s'$, if $\textbf{Tr}(s, s', G,\beta)=\textbf{true}$ then $s$ and $s'$ give the same label to every isolated vertex in $G$.
        \item For every quadruple of states $r$, $r'$, $s$, and $s'$, if $r$ and $s$ agree on the non-isolated vertices of $G$, the pairs of states $r,r'$ and $s,s'$ give the same label to every isolated vertex of $G$, and $r'$ and $s'$ agree on the non-isolated vertices of $G$, then $\textbf{Tr}(r, r', G,\beta) = \textbf{Tr}(s, s', G,\beta)$.
        \item For every pair of states $s_\Lambda$ and $s'_\Lambda$ that agree on the vertices in $A_\Lambda$, $\textbf{Ac}(s_\Lambda, x)=\textbf{true}$ if and only if $\textbf{Ac}(s'_\Lambda, x)=\textbf{true}$.
    \end{enumerate}
\end{defn}

Returning to \textsc{Temporal Hamiltonian Path}, being locally temporally uniform enforces that no vertices are labelled $\current$ or $\visited$ before they are incident to an active edge; vertices which are not in their active interval cannot change labels between timesteps; and that, in general, acceptance of any final states is not dependent on the labels of vertices which are not in their active interval. Since we only check the counter for acceptance of \textsc{Temporal Hamiltonian Path}, this final condition trivially holds.

\begin{restatable}{theorem}{vimhamiltonthm}
    \textsc{Temporal Hamiltonian Path} is $(X,1,f_1,f_2)$-locally temporally uniform, where $X=\{\visited,\unvisited,\current\}$, and $f_1(G,\beta)=f_2(x)=n$ for any snapshot $G$ of a temporal graph with $n$ vertices.
\end{restatable}

\begin{toappendix}
We now give a meta-algorithm that solves any locally temporally uniform problem using the starting routine, transition routine and the accepting routine, when given the input $x = (\mathcal{G}, \beta)$. This algorithm uses locality to avoid having to consider every possible state on each timestep. Instead, on each timestep $t$ the algorithm considers only one state for each possible state for $F_t$ (the bag of the VIM sequence at time $t$), the number of which we bound in terms of the VIM width and the labels and vectors we require to express our states. These states for $F_t$ are extended to states for the vertex set of the input graph by giving every other vertex label $U$. We first obtain a lemma that shows that states extended in this manner will agree with any state in a sequence following from an initial state, on all vertices not yet in their active interval.

\begin{restatable}{lem}{vimlabelulater}\label{lem:labulater}
    Consider any instance $(\mathcal{G}, \beta)$ of a $(X,k,f_1,f_2)$-locally temporally uniform problem, such that $S_0$ is the set of initial states generated by \textbf{St}, \textbf{Tr} is the transition algorithm and $[F_t]_{t\leq\Lambda}$ is the VIM sequence of $\mathcal{G}$.  If $s_0, ..., s_t$ is a sequence of states such that $s_0 \in S_0$, and $\textbf{Tr}(s_{i-1}, s_i, G_i,\beta)=\textbf{true}$ for $1 \leq i \leq t$, then $s_t$ gives label $U$ to every vertex $v \in \bigcup_{t' > t}{F_{t'}} \setminus F_t$.
\end{restatable}
\begin{proof}
    We proceed by induction on the timestep $t$. If $t=0$, then as the problem is locally temporally uniform, $s_0$ gives label $U$ to every vertex not in $F_0$.
    
    Then, assume by induction that $s_{t-1}$ gives label $U$ to every vertex $v \in \bigcup_{t' > t-1}{F_{t'}} \setminus F_{t-1}$. If $\textbf{Tr}(s_{t-1}, s_t, G_t,\beta)=\textbf{true}$, then $s_{t-1}$ and $s_t$ give the same label to any isolated vertex in $G_t$, and therefore give label $U$ to every vertex $v \in \bigcup_{t' > t}{F_{t'}} \setminus F_{t} \subseteq \bigcup_{t' > t-1}{F_{t'}} \setminus F_{t-1}$ as any vertex $v \in \bigcup_{t' > t}{F_{t'}} \setminus F_t$ is not in $F_t$, and therefore cannot have any active incident edges active on timestep $t$, so is isolated in $G_t$. 
\end{proof}
%S_0 needs to only be for F_0, otherwise newly produced states won't agree on vertices outside active interval

\begin{algorithm}[H]\caption{\textsc{Locally Temporally Uniform Algorithm}}\label{alg:loctempunf}
\begin{algorithmic}[1]
\Require A problem input $x = (\mathcal{G}, \beta)$ with the starting routine \textbf{St}, transition routine \textbf{Tr} and acceptance routine \textbf{Ac}.
\Ensure Whether $x$ is a yes-instance.
\State{Let $S_0$ be the set of $(X,k)$-states output by $\textbf{St}(x)$.}
\State{Fix a label $U \in X$}
\For{$t = 1, \dots, \Lambda$}
    \State{$S_t$ $\leftarrow$ $\{\}$}
    \ForAll{Possible pairs $(l_{F_t},\textbf{v}_{F_t})$ of labellings $l_{F_t}:V(F_t)\to X$ and vectors $\textbf{v}_{F_t}$ with $k$ entries of maximum magnitude $b$}
        \State{$s_t$ $\leftarrow$ the state agreeing with $(l_{F_t},\textbf{v}_{F_t})$ such that all vertices not in $F_t$ are given label $U$}
        \ForAll{$s_{t-1} \in S_{t-1}$}
            \State{$r_{t-1}$ $\leftarrow$ the state agreeing with $s_{t-1}$ on $F_t$ where all other vertices get label $U$}
            \If{\Call{Tr}{$r_{t-1}$, $s_t$, $G_t$, $\beta$}}
            \State{$S_t$ $\leftarrow$ $S_t \cup \{s_t\}$}
            \EndIf
        \EndFor
    \EndFor
    \EndFor
    \ForAll{$s_{\Lambda(\mathcal{G})} \in S_{\Lambda(\mathcal{G})}$}
        \If{\Call{Ac}{$s_{\Lambda(\mathcal{G})}$, $x$}}
            \State{\Return{True}}
        \EndIf
    \EndFor
    \State{\Return{False}}
    \end{algorithmic}
\end{algorithm}

% \resetlinenumber[294]

We claim that \cref{alg:loctempunf} solves any temporally locally uniform problem.  On a timestep $t$, the algorithm considers every possible state for $F_t$, and then extends these states to the entire graph by giving every other vertex some fixed label. We first argue that in doing so, the algorithm does not omit any required states, and for any sequence $S$ of states of VIM sequence starting with an initial state such that the transition routine returns \textbf{true} for all consecutive pairs the algorithm will produce a sequence of states of each snapshot which agrees with $S$.

\begin{restatable}{lem}{vimcorrectran}\label{lem:correctran} 
Let $x=(\mathcal{G},\beta)$ be an instance of a $(X,k,f_1,f_2)$-locally temporally uniform problem $P$ with transition routine $\textbf{Tr}$, acceptance routine $\textbf{Ac}$, and associated set $S_0$ of initial $(X,k)$-states output by starting routine $\textbf{St}$. Fix any $t$ with $0\leq t\leq \Lambda$. Then, there exists a sequence of states $s_0, \ldots, s_t$ with $s_0 \in S_0$ and $\textbf{Tr}(s_{i-1}, s_i, G_i,\beta)=\textbf{true}$ for all timesteps $1 \leq i \leq t$, if and only if there exists a state $s'_t$ in the set $S_t$ produced by \Cref{alg:loctempunf} that agrees with $s_t$ on $F_t$, where $F_t$ is the bag for timestep $t$ in the VIM sequence of $\mathcal{G}$.
\end{restatable}

\begin{proof}
    We in fact prove a stronger result, that not only does there exist a state $s'_t$ in $S_t$ that agrees with $s_t$ on $F_t$, but that this state gives label $U$ to every vertex not in $F_t$. We proceed by induction on the length of the sequence $t$. The base case when $t=0$ is trivial, as $S_0$ output by the starting routine, and by \cref{def:loctempunf} every state $s_0 \in S_0$ gives label $U$ to every vertex not in $F_0$.

    Assume by induction that there exists a sequence of states $s_0, ... s_{t-1}$ with $s_0 \in S_0$ and $\textbf{Tr}(s_{i-1}, s_i, G_i,\beta)=\textbf{true}$ for all timesteps $1 \leq i \leq t-1$, if and only if there exists a state $s'_{t-1}$ in the set $S_{t-1}$ produced by \Cref{alg:loctempunf} that agrees with $s_{t-1}$ on $F_{t-1}$, and that gives label $U$ to every vertex not in $F_{t-1}$.

    Now, consider any state $s_t$ such that there exists a sequence of states $s_0, \ldots, s_t$ with $s_0 \in S_0$ and $\textbf{Tr}(s_{i-1}, s_i, G_i,\beta)=\textbf{true}$ for all timesteps $1 \leq i \leq t$.  By induction, there exists a state $s'_{t-1} \in S_{t-1}$ that agrees with $s_{t-1}$ on $F_{t-1}$, and gives all vertices not in $F_{t-1}$ label $U$. This state is used by \Cref{alg:loctempunf} to produce a state $r_{t-1}$ which agrees with $s'_{t-1}$ on $F_t$. By \Cref{lem:labulater}, any vertices in $F_t \setminus F_{t-1}$ are given label $U$ by $s_{t-1}$, and $r_{t-1}$ also gives these vertices label $U$, and therefore $r_{t-1}$ agrees with $s_{t-1}$ on $F_t$, and gives label $U$ to every vertex not in $F_t$. \Cref{alg:loctempunf} will consider a state $s'_t$ that agrees with $s_t$ on $F_t$, such that every vertex not in $F_t$ is given label $U$, as it considers every possible state for $F_t$, extending these states by labelling the remaining vertices with $U$. Because $\textbf{Tr}(s_{t-1}, s_t, G_t,\beta)=\textbf{true}$ and $A_t$ is exactly the set of non-isolated vertices of $G_t$ we have that $s_{t-1}$ and $s_t$ give the same label to every vertex not in $A_t$ by \cref{def:loctempunf}. Now, consider any vertex $v \notin A_t$. If $v \in F_t$ then $s_t$ gives the same label to $v$ as $s'_t$, as $s'_t$ and $s_t$ agree on $F_t$. Then $s_{t-1}$ also gives the same label to $v$ as $v \notin A_t$ and $s_{t-1}$ and $s_t$ give the same label to any vertex not in $A_t$. Finally, $r_{t-1}$ gives the same vertex to $v$ as it agrees with $s_{t-1}$ on $F_t$. Otherwise, if $v \notin F_t$, both $s'_{t-1}$ and $r_{t-1}$ give label $U$ to $v$. Therefore $r_{t-1}$ and $s'_{t-1}$ give the same label to every vertex not in $A_t$. Furthermore, $s'_t$ and $s_t$ agree on $A_t$, as $A_t \subseteq F_t$, as do $r_{t-1}$ and $s_{t-1}$. Then by \Cref{def:loctempunf}, $\textbf{Tr}(r_{t-1}, s'_t, G_t,\beta)=\textbf{true}$, and line 10 of \Cref{alg:loctempunf} will place the state $s'_t$ in $S_t$.

    Conversely, consider any state $s'_t \in S_t$, and see that there must exist some state $s'_{t-1} \in S_{t-1}$ such that if $r_{t-1}$ is a state that agrees with $s'_{t-1}$ on $F_t$ and gives label $U$ to all vertices not in $F_t$, then $\textbf{Tr}(r_{t-1}, s'_t, G_t,\beta)=\textbf{true}$. By induction $s'_{t-1}$ agrees on $F_{t-1}$ with some state $s_{t-1}$ where there exists a sequence of states $s_0, ... s_{t-1}$ with $s_0 \in S_0$ and $\textbf{Tr}(s_{i-1}, s_i, G_i,\beta)=\textbf{true}$ for all timesteps $1 \leq i \leq t-1$. Furthermore, by definition, $s'_{t-1}$ gives label $U$ to every vertex not in $F_{t-1}$, and so $r_{t-1}$ gives label $U$ to every vertex in $F_t \setminus F_{t-1}$. By \Cref{lem:labulater} $s_{t-1}$ also gives label $U$ to every vertex in $F_t \setminus F_{t-1}$, and $r_{t-1}$ agrees with $s_{t-1}$ on $F_t \cap F_{t-1}$, so therefore $r_{t-1}$ agrees with $s_{t-1}$ on $F_t$. As $\textbf{Tr}(r_{t-1}, s'_t, G_t,\beta)=\textbf{true}$ we have that $r_{t-1}$ and $s'_t$ give the same label to every isolated vertex in $G_t$. Furthermore $r_{t-1}$ and $s_{t-1}$ agree on the non-isolated vertices of $G_t$, as these are $A_t \subseteq F_t$. Consider now the state $s_t$ that agrees with $s'_t$ on the non-isolated vertices of $G_t$, and gives the same label as $s_{t-1}$ to every isolated vertex in $G_t$. By \cref{def:loctempunf} we have that $\textbf{Tr}(s_{t-1}, s_t, G_t,\beta)=\textbf{true}$ because $\textbf{Tr}(r_{t-1}, s'_t, G_t,\beta)=\textbf{true}$. Finally see that for any vertex $v \in F_t \setminus A_t$, that is for any isolated vertex of $G_t$ in $F_t$, $s_t$ and $s_{t-1}$ give the same label to $v$. Now, as $r_{t-1}$ and $s_{t-1}$ agree on $F_t$, $r_{t-1}$ also gives the same label to vertex $v$. Finally $s'_t$ gives the same label to vertex $v$, as it gives the same label as $r_{t-1}$ to every isolated vertex of $\mathcal{G}$. Therefore $s'_t$ and $s_t$ agree on $F_t$ as required.
\end{proof}

We can use these tools to show that our algorithm solves any $(X,k,f_1,f_2)$-locally temporally uniform problem, and that the running time can be bounded in terms of the running times of the subroutines.
\end{toappendix}
\begin{theoremrep}\label{thm:loctempunf}
    Let $x=(\mathcal{G},\beta)$ be an instance of a $(X,k,f_1,f_2)$-locally temporally uniform problem $P$ where $\mathcal{G}$ has $n$ vertices, lifetime $\Lambda$ and VIM width $\omega$.  
    % Let $P$ have a transition routine $\textbf{Tr}$, acceptance routine $\textbf{Ac}$, and starting routine $\textbf{St}$. 
    We can determine if $x$ is a yes-instance of $P$ in time $O(\Lambda (\max_{t\in[\Lambda]}f_1(G_t,\beta))f_2(x)b^{2k}|X|^{2\omega})$, where 
    $b \in \mathbb{N}$ is an upper bound on the absolute value of any entry of a vector in a
    % $b$ is the maximum magnitude of any counter variable in a 
    $(X,k)$-state.
\end{theoremrep}
\begin{proof}
    We show that \Cref{alg:loctempunf} returns true if and only if it is given a yes-instance as input, and furthermore that \Cref{alg:loctempunf} runs in the required time.

    If \Cref{alg:loctempunf} returns true, then there exists a state $s_\Lambda \in S_\Lambda$ such that $\textbf{Ac}(s_\Lambda, x) = \textbf{true}$. Furthermore, as \Cref{alg:loctempunf} only places a state $s_t$ in $S_t$ if there exists a state $s_{t-1} \in S_{t-1}$ with $\textbf{Tr}(s_{t-1}, s_t, G_t,\beta)=\textbf{true}$, there exists a sequence $s_0, ..., s_\Lambda$ of states, such that $s_t \in S_t$ for every timestep $t$, and $\textbf{Tr}(s_{t-1}, s_t, G_t,\beta)=\textbf{true}$ for all timesteps $t \geq 1$. Therefore, by \Cref{def:loctempunf}, $x$ is a yes-instance.

    If $x$ is a yes-instance, then there exists a sequence $s_0, ..., s_\Lambda$ of $(X,k)$-states with $s_0 \in S_0$, and $\textbf{Tr}(s_{t-1}, s_t, G_t)=\textbf{true}$ for all timesteps $t \geq 1$, and $\textbf{Ac}(s_\Lambda, x)=\textbf{true}$. Then, by \Cref{lem:correctran} there exists a state $s'_\Lambda \in S_\Lambda$ that agrees with $s_\Lambda$ on $A_\Lambda$. Then, by \Cref{def:loctempunf} $\textbf{Ac}(s'_\Lambda, x)=\textbf{true}$, as $\textbf{Ac}(s_\Lambda, x)=\textbf{true}$.

    For every timestep $t$, there is at most one entry in $S_t$ for every possible state of $F_t$, of which there at most $b^k|X|^\omega$. Now for each timestep $t \in [\Lambda]$ \Cref{alg:loctempunf} runs the transition routine for every pair of a possible state in $S_t$, and a state in $S_{t-1}$. Since there are at most $\omega$ vertices in a bag at any given time, this can be achieved in time $O(f_1(G,\beta)b^{2k}|X|^{2\omega}\Lambda)$ over all timesteps. Finally, \Cref{alg:loctempunf} runs the acceptance routine for every state in $S_\Lambda$, which can be achieved in time $O(f_2(n, \Lambda,\beta)b^k|X|^\omega)$, giving an overall runtime of $O(\Lambda f_1(G,\beta)f_2(n,\Lambda,\beta)b^{2k}|X|^{2\omega}+f(n, \Lambda,\beta)b^k|X|^\omega)=O(\Lambda f_1(G,\beta)f_2(n, \Lambda,\beta)b^{2k}|X|^{2\omega})$.
\end{proof}

% Using this theorem, we show inclusion of \textsc{Temporal Hamiltonian Path} in FPT with respect to VIM width.

In fact, our meta-algorithm gives an exact characterisation of the problems in FPT with respect to VIM width.

% \begin{corollary}\label{cor:vim-meta}
%  If $P$ is a $(k,X,f_1,f_2)$-locally temporally uniform problem, where there exists a function $g$ such that, for any temporal graph $\mathcal{G}$ with VIM width at most $\omega$ and any snapshot $G$ of $\mathcal{G}$, $f_1(G,\beta) \leq g(\omega)(|G| + |\beta|)^{O(1)}$, $|X|$ is a function of $\omega$ alone and $b^kf_2(\mathcal{G},\beta) \leq g(\omega) (|\mathcal{G}| + |\beta|)^{O(1)}$, where $b$ is the maximum absolute value of any entry of a vector in a $(k,X)$-state of $\mathcal{G}$.
%  Then $P$ is in FPT with respect to the VIM width of the input temporal graph.
% \end{corollary}

\begin{theoremrep}\label{thm:VIM-iff}
    Let P be a problem that takes $x=(\mathcal{G},\beta)$ as input, where $\mathcal{G}$ is a temporal graph with VIM width $\omega$ and $\beta$ is a string. $P$ is in FPT with respect to $\omega$ if and only if $P$ is a $(X,k,f_1,f_2)$-locally temporally uniform problem such that 
      \begin{itemize}
         \item $k$ is a constant, 
          \item $|X|$ is upper bounded by a function of $\omega$ alone, and
          \item for a computable function $g$ and any snapshot $G$ of $\mathcal{G}$, $f_1(G,\beta)$, $f_2(\mathcal{G},\beta)$ and $b$ are all bounded above by $g(\omega)(|G| + |\beta|)^{O(1)}$,
     \end{itemize}
    % $k$ is a constant, $|X|$ is upper bounded by a function of $\omega$ alone, and for a computable function $g$ and any snapshot $G$ of $\mathcal{G}$, $f_1(G,\beta)$, $f_2(\mathcal{G},\beta)$ and $b$ are all bounded above by $g(\omega)(|G| + |\beta|)^{O(1)}$,
    where $b$ is the maximum absolute value of any entry of a vector in a $(X,k)$-state of $\mathcal{G}$.
\end{theoremrep}
\begin{proof}
    Using Theorem~\ref{thm:loctempunf}, we get that if $P$ is $(X,k,f_1,f_2)$-locally temporally uniform where $k$ is a constant, $|X|$ is upper bounded by a function of $\omega$, and for a computable function $g$ and any snapshot $G$ of $\mathcal{G}$, $f_1(G,\beta)$, $f_2(\mathcal{G},\beta)$ and $b$ are all bounded above by $g(\omega)(|G| + |\beta|)^{O(1)}$, then it is in FPT with respect to $\omega$.

    We now show the reverse direction. Suppose there exists an fpt-algorithm $A$ for $P$ with respect to $\omega$. Let the runtime of $A$ be $a(\omega)\text{poly}(n,\Lambda,|\beta|)$ for some computable function $a$. Then, we argue that $P$ must be $(X,1,n,a)$-locally temporally uniform where $X$ is a set consisting of a single label. To show this we construct the states required and prove that an instance $x=(\mathcal{G},\beta)$ is a yes-instance if and only if there exists a sequence of states $s_0,\ldots,s_{\Lambda}$ such that $s_0$ is an initial state, $\textbf{Tr}(s_{i-1},s_{i},G_{i},\beta)$ returns true for all $1\leq i \leq \Lambda$, and $\textbf{Ac}(s_{\Lambda}, x)$ returns true. Our set of labels consists of a single label, call it $U$. The counter vector consists of a vector with one entry, let that entry be $1$. Since this predetermines all states in the sequence, the only possible initial states are those such that all vertices are labelled $U$ and the counter is equal to $1$. Our transition routine returns true if and only if the two states are equal, which must always be the case given our description of the states. This leaves the acceptance routine. This is the algorithm $A$ with input $x$. It is clear that the acceptance routine returns true if and only if $x$ is a yes-instance. Note that, since all states are the same, output of the routines cannot depend on vertices not in their active interval.  Therefore, $P$ is $(X,1,n,a)$-locally temporally uniform and the statement holds.
\end{proof}
\subsection{Meta-algorithm parameterised by TIM width}\label{sec:TIM-alg}

We now move on to our second meta-algorithm. This builds on the techniques used in our first meta-algorithm to function on a more general decomposition. %Earlier, we drew a comparison between VIM width and TIM width, and pathwidth and treewidth. In this section we generalise a meta-algorithm for a decomposition graph which is a path to a meta-algorithm for a decomposition graph which is a tree. 
Previously, we needed only to be able to generate a set of acceptable starting states, to determine whether we could transition from one state to the next efficiently, and to determine whether we have an acceptable finishing state efficiently. In this meta-algorithm we require efficient subroutines to check starting states, finishing states, and the validity of states that are neither at the beginning or end of the temporal graph, all of which can be applied to connected components independently. We also require a transition routine which determines whether we can transition between consecutive labellings of vertices in a connected component of a snapshot. Finally, since all the other checks relate to connected components, we also introduce a vector bound on the sum of vectors associated with all components at all times, to allow us to enforce certain global constraints over the whole graph. Since the algorithm is parameterised by TIM width, it is most useful for problems where we do not need to consider all vertices in their active interval simultaneously, but can consider the vertices in different connected components at each time independently.

To be able to consider the connected components of each snapshot independently, we generalise the notion of a $(X,k)$-state. This generalisation allows us to have a different vector of counters for each connected component, which means we can count what happens in each connected component separately. That is, we label the vertices with elements of $X$ and, for each connected component, we may have a different vector with $k$ entries. 
\begin{defn}[$(X,k)$-Component State]
    A $(X,k)$-component state on a static graph $G=(V,E)$ with $c$ connected components is a tuple $(l, \textbf{\emph{v}}_1,\ldots,\textbf{\emph{v}}_{c}, \nu)$, where, for any subset $V'$ of $V(G)$, $l : V' \to X$ is a labelling of the vertices in $V'$ using the labels from set $X$, $\textbf{\emph{v}}_1,\ldots,\textbf{\emph{v}}_{c}$ are vectors of $k$ integers, each of which is of magnitude $|G|^{O(1)}$, and $\nu$ is a bijective map from connected components of $G$ to the vectors $\textbf{\emph{v}}_1,\ldots, \textbf{\emph{v}}_c$.
\end{defn}

Applying this definition to the Hamiltonian path example, we reuse the set of 
labels $X=\{\visited,\unvisited,\current\}$, and the vectors in the state now contain one integer $p$, which counts the number of current locations in each snapshot.

% Note that a $(k,X)$-component state as defined in Section~\ref{sec:Sam-chapter} is simply a $(1,k,X)$-state.

% We say that two states $s = (l, c_1,\ldots,c_{\kappa})$ and $s' = (l', c'_1,\ldots,c'_{\kappa})$ for vertex set $V$ \emph{agree} on a vertex set $U$ if and only if $V' \subseteq V$, $l$ gives the same label as $l'$ to every vertex in $V'$, and $c_i = c_i'$ for all $i\in [\kappa]$. 
Unlike the previous meta-algorithm, this algorithm functions by allowing the connected components of each snapshot to be considered separately. As a result, we allow a different vector in the state for each connected component in the snapshot. Define the \emph{restriction} of a $(X,k)$-component state $s=(l,\textbf{v}_1,\ldots,\textbf{v}_c, \nu)$ to a connected component $C$ to be the state $s|_C=(l|_C,\nu(C))$ where $l|_C$ is the restriction of the labelling $l$ to the vertices in $C$. We denote by $\mathcal{C}_t$ the set of connected components in the snapshot $G_t$ of the temporal graph $\mathcal{G}$ at time $t$.  Note that, when we restrict a $(X,k)$-component state to the subgraph induced by a connected component of a static graph, we have a $(X,k)$-state of that subgraph as defined earlier. Specifically, this holds when restricted to a component of a snapshot.

\begin{defn}[$(X,k,f)$-Component-Exchangeable Temporally Uniform Problem]\label{def:component-temp-unif}
    We say that a decision problem $P$ with input $x=(\mathcal{G},\beta)$ such that $\mathcal{G}$ has lifetime $\Lambda$ is $(X,k,f)$-component-exchangeable temporally uniform if and only if there exist:
    \begin{enumerate}
        \item a transition algorithm \textbf{Tr} that takes two labellings for the vertices of a connected, static graph $C$ with labels from the set $X$, the graph $C$, and the problem instance, runs in time at most $f(|C|,x)$, and returns \textbf{true} or \textbf{false},
        \item a starting algorithm \textbf{St}, a validity algorithm \textbf{Val}, and a finishing algorithm \textbf{Fin} that all take a $(X,k)$-state of a connected static graph $C$ and the problem instance, run in time at most $f(|C|,x)$, and return \textbf{true} or \textbf{false},
        % \item a starting algorithm \textbf{St} that takes the problem instance, a restriction of a $(k,X)$-component state to a connected component $C$ in the snapshot $G_1$ and $C$, runs in time at most $f(|C|,x)$, and returns \textbf{true} or \textbf{false},
        % \item a finishing algorithm \textbf{Fin} that takes the problem instance, a restriction of a $(k,X)$-component state to a connected component $C$ in the snapshot $G_{\Lambda}$ and $C$, runs in time at most $f(|C|,x)$ and returns \textbf{true} or \textbf{false}, and
        \item a vector $\textbf{\emph{v}}_{\text{upper}}$ of $k$ integers,
    \end{enumerate}
    \noindent such that $x=(\mathcal{G}, \beta)$ is a yes instance of $P$ if and only if there exists a sequence $s_0, ..., s_\Lambda$ of $(X,k)$-component states of each snapshot of $\mathcal{G}$ where
    \begin{enumerate}[i]
        \item for each connected component $C_1$ of $G_1$, $\textbf{St}(s_0|_{C_1},C_1,x)=\textbf{true}$;
        \item for each connected component $C_{\Lambda}$ of $G_{\Lambda}$, $\textbf{Fin}(s_{\Lambda}|_{C_{\Lambda}},C_{\Lambda},x)=\textbf{true}$;
        \item $\textbf{Tr}(l_{t-1}|_{C_t}, l_t|_{C_t}, C_t,x)=\textbf{true}$ where $l_i$ is the labelling of vertices of state $s_i$, for all times $1 \leq t \leq \Lambda$ and connected components $C_t$ of $G_t$;
        \item $\textbf{Val}(s_t|_{C_t}, C_t,x)=\textbf{true}$ for all times $1 \leq t < \Lambda$ and connected components $C_t$ of $G_t$; and
        \item the sum of vectors satisfies $\sum_{0\leq t\leq \Lambda} \sum_{C \in\mathcal{C}_t}\nu_{s_t} (C)\leq \textbf{\emph{v}}_{\text{upper}}$, where $\nu_{s_t}$ is the function $\nu$ in the $(X,k)$-component state $s_t$ and $\leq$ denotes element-wise vector inequality.
    \end{enumerate}    
\end{defn}
Note that we only require a vector to upper bound the counters as any problem where we require some or all of these to be lower bounded can be encoded using negative entries.

When applying this meta-algorithm to \textsc{Temporal Hamiltonian Path}, we make our upper bound on the sum of the values of the vectors, $\textbf{v}_{\text{upper}}=(\Lambda)$. With some effort, we can show that the combination of our subroutines and enforcing that the sum of counters is at most $\Lambda$ gives us that there is at most one ``current location'' at any time, if we construct the states in such a way that there is at least one current location in each snapshot.
The starting routine checks that there is at most one vertex labelled $\current$ in each connected component of the first snapshot, that the counter for the component matches the number of $\current$ vertices, and that every other vertex is labelled $\unvisited$. The finishing and validity routines similarly check that there is at most one vertex labelled $\current$ in the connected component, and the finishing routine additionally checks that every other vertex is labelled $\visited$. Algorithm~\ref{alg:temp-ham-timw} gives the transition algorithm.
% As in the VIM example, the transition routine checks that the labelling is either the same, or that there is an active edge between the vertices labelled $\current$ at each time, the vertex labelled $\current$ at the earlier time is labelled $\visited$ at the later time and the vertex labelled $\current$ at the later time is labelled $\unvisited$ at the earlier time.

\begin{algorithm}\caption{\textsc{Temporal Hamiltonian Path Component Exchangeable Transition Routine}}\label{alg:temp-ham-timw}
    \begin{algorithmic}[1]
    \Require A connected component $C$ and labellings $l_1$ and $l_2$ for $V(C)$.
    \Ensure Returns true when there exists a path in $\mathcal{G}$ visiting the vertices labelled $\visited$ ending at a vertex labelled $\current$ by $l_1$ if and only if there is a path which visits the vertices labelled $\visited$ by $l_2$ which ends at a vertex labelled $\current$ by $l_2$
    \State{Let $\text{Unvisited}_1$ and $\text{Unvisited}_2$ be the set of vertices labelled $\unvisited$ by $l_1$ and $l_2$ respectively, and equivalently for $\text{Visited}_1$, $\text{Visited}_2$, and $\text{Current}_1$, $\text{Current}_2$.}
    \If{$\text{Current}_1 \setminus \text{Current}_2$ contains a single vertex $c_1$, and $\text{Current}_2 \setminus \text{Current}_1$ contains a single vertex $c_2$}
        \If{$\{c_1, c_2\} \in E(C)$, $c_2 \in \text{Unvisited}_1$, and $\text{Visited}_1 \cup \{c_1\} = \text{Visited}_2$}
            \State{\Return{True}}
        \EndIf
        \ElsIf{$l_1= l_2$}
        \State{\Return{True}}
    \EndIf
    \State{\Return{False}}
    \end{algorithmic}
\end{algorithm}

For ease, we add an additional set of nodes to the TIM decomposition at time 0 by duplicating those at time 1 and adding an arc from the copy at time 0 to the copy at time 1. This has the same function as the additional bag at time 0 added to the VIM sequence in Section~\ref{sec:Sam-chapter} -- it allows us to run the transition routine on the first snapshot of the graph. This does not increase the width of the TIM decomposition, and only adds leaves to the decomposition. 

The meta-algorithm we describe is on a rooted TIM decomposition. The root can be arbitrarily chosen since its main purpose is to provide an orientation for the nodes of the decomposition. Note that this is unrelated to the direction of the edges in the decomposition tree. Here, the parent of a node is the unique vertex in its neighbourhood which is on the undirected path from the bag to the root in the underlying graph of the decomposition. The children of a node $s$ are all the nodes for which $s$ is their parent. Figure~\ref{fig:eg-tim-root} depicts a rooted TIM decomposition of the example graph in Figure~\ref{fig:eg-both} with the additional bags at time 0. We call a connected component of a snapshot of a temporal graph a \emph{timed} connected component.

\begin{figure}[ht]
    \centering
    \includegraphics[width=0.55\linewidth,page=9]{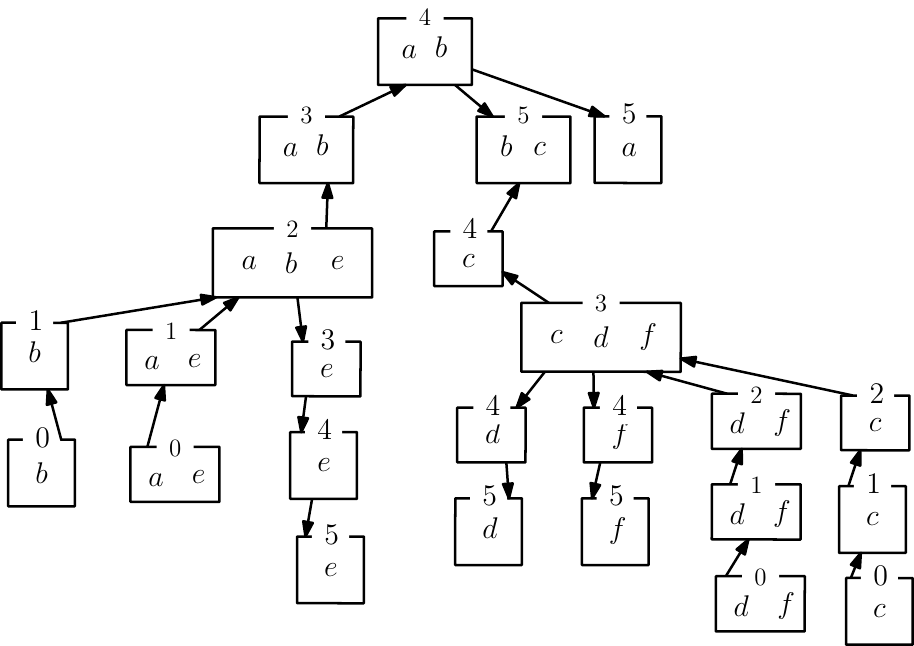}
    \caption{The TIM decomposition of the temporal graph in Figure~\ref{fig:eg-both} rooted at the bag containing vertices $a$ and $b$ at time 4 and with added bags at time 0. The time with which a node is labelled is depicted at the top of the box.}
    \label{fig:eg-tim-root}
\end{figure}

We now give some intuition for how our meta-algorithm functions. Given a $(X,k,f)$-component-exchangeable temporally uniform problem and a rooted TIM decomposition of the input temporal graph, our algorithm dynamically programs from leaves to root of an auxiliary decomposition. 
As part of the dynamic program, we keep a \textbf{total} vector of the vectors that appear in the subtree rooted at the node in question. This allows us to enforce that the sum of all vectors of states is bounded above by $\textbf{v}_{\text{upper}}$ when we reach the root.

Note that, in a TIM decomposition a node may have both a child and a parent such that the parent bag and child bag have the same label. For example, if $C$ is a connected component of $G_{t-1}$ and we need to know the label of all vertices from $C$ at time $t$ to determine whether a transition is possible, we may require multi-generational comparisons of states (since vertices of $C$ may belong to different components of $G_{t}$). To counteract this, we define an auxiliary decomposition graph which is a variation of the TIM decomposition. We find this decomposition by adding duplicate copies of each vertex in each bag to the parent bag. This allows us to only consider the connected components in a bag and its children to determine whether a state of the bag could correspond to a possible solution.

\begin{restatable}{theorem}{temphamtimthm}\label{thm:temphamtimthm}
         \textsc{Temporal Hamiltonian Path} is $(X,1,f)$-component-exchangeable temporally uniform, where $X=\{\visited,\unvisited,\current\}$, and $f(|C|,x)=\phi$ for every timed connected component $C$ of an input temporal graph with TIM width $\phi$.
\end{restatable}

\begin{toappendix}
\begin{defn}[2-Step TIM decomposition]
    Given a temporal graph $\mathcal{G}$ with a TIM decomposition $(T,B,\tau)$, the corresponding $\twostep$ \emph{TIM decomposition} $(T,B^2)$ of $\mathcal{G}$ is indexed by the same tree $T$ and the bag $B^2(s)$ of a node $s$ is the set of pairs
    $\{(v,\tau(s')\, :\, s'\in S, v\in B(s')\}$, where $S$ is the union of $s$ and its children in $T$.
    % \begin{enumerate}
    %     % \item creating a bag for each bag in the TIM decomposition,
    %     % \item adding arcs between nodes if they exist in the TIM decomposition,
    %     \item for all vertices $v$ in a bag $B(s)$ of the TIM decomposition, adding the vertex-time pair $(v,\tau(s))$ to the corresponding bag $B^2(s)$ of the $\twostep$ TIM decomposition,
    %     \item for each child $s_c$ of $s$, add the vertex-time pairs $(v,\tau(s_c))$ to $B^2(s)$ for all vertices $v\in B(s_c)$.
    % \end{enumerate}
    The width of this decomposition is the maximum cardinality of the bags in the decomposition. We refer to $(T,B,\tau)$ as a TIM decomposition \emph{associated to} $(T,B^2)$.
\end{defn}
An example of a $\twostep$ TIM decomposition can be found in Figure~\ref{fig:eg-tim-2step}.

\begin{figure}[ht]
    \centering
    \includegraphics[width=0.8\linewidth,page=10]{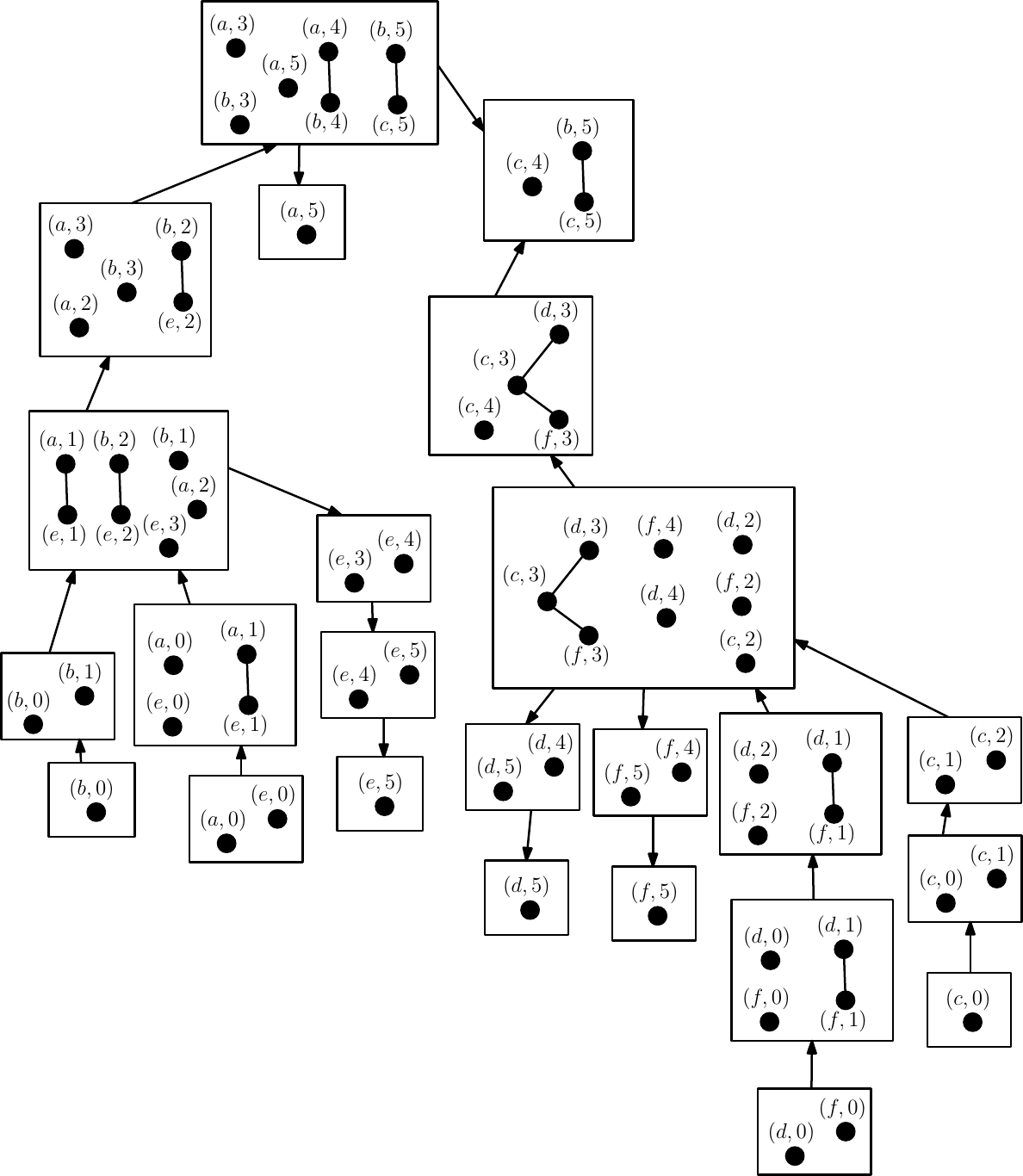}
    \caption{The $\twostep$ TIM decomposition corresponding to the rooted TIM decomposition in Figure~\ref{fig:eg-tim-root} of the temporal graph in Figure~\ref{fig:eg-both}.}
    \label{fig:eg-tim-2step}
\end{figure}

We refer to a bag $B(s)$ of a TIM decomposition as the bag \emph{corresponding} to a bag $B^2(s)$ in a $\twostep$ TIM decomposition (and vice versa) if $B^2(s)$ is the union of $B(s)$ and its children. Note that, for all bags in a $\twostep$ TIM decomposition, the corresponding bag of a TIM decomposition is indexed by the same node in the tree. 

\begin{observation}
    Given a TIM decomposition $(T,B,\tau)$ of width $\phi$ of a temporal graph $\mathcal{G}$ with $n$ vertices, we can construct the corresponding $\twostep$ TIM decomposition in $O(n\phi^2)$ time.
\end{observation}
We note, using Observation~\ref{obs:tim-neighbours}, that there are at most 2 copies of each vertex in a bag of a TIM decomposition in the bags of its children -- one copy in a bag of a node labelled with the time before and another at the time after. This implies that there are at most 3 times which appear in the vertex-time pairs in a bag of the $\twostep$ decomposition. These times are the time assigned to the corresponding bag of the TIM decomposition, and the times directly before and after that time. This combined with Observation~\ref{obs:tim-no-children} (that each bag of a TIM decomposition has at most $2\phi$ neighbours) leads us to the following observation.
\begin{observation}\label{obs:two-step-width}
    Given a temporal graph $\mathcal{G}$ with TIM decomposition $(T,B,\tau)$, then a $\twostep$ TIM decomposition of $\mathcal{G}$ has width at most $3\phi^2$. 
\end{observation}

By noting that, for each node $s$ with parent $s_p$ of a TIM decomposition $(T,B,\tau)$, each vertex in $B(s)$ appears in $B^2(s)$ and $B^2(s_p)$ of the corresponding $\twostep$ TIM decomposition, we get the following observation.

\begin{observation}
    Every vertex-time pair appears exactly twice in the $\twostep$ TIM decomposition, and the bags they appear in are adjacent. 
\end{observation}

\begin{observation}\label{obs:2-step-leaves}
    For every leaf node $l$ in a TIM decomposition $(T,B,\tau)$ with $\twostep$ TIM decomposition $(T,B^2)$, $B^2(l)=\{(u,\tau(l)) : u\in B(l)\}$.
\end{observation}
\begin{observation}\label{obs:two-step-component-in-children}
    Let $C$ be a connected component of a snapshot $G_t$ of a temporal graph $\mathcal{G}$. Then, let $S$ be the set of vertex-time pairs $\{(v,t) \,:\, v\in C\}$. All bags of any $\twostep$ TIM decomposition contain either all pairs in $S$ or none of them.
    Furthermore, there is at most one child of $B^2(s)$ in $(T,B^2)$ which contains elements of $S$.
\end{observation}

We can think of a bag in a $\twostep$ TIM decomposition as a collection of connected components of snapshots of temporal graphs. To that end, let $\mathcal{C}^{s}$ denote the set of timed connected components (connected component-time pairs) in $B^2(s)$. Let $\mathcal{C}^s_t$ be the set of connected components in $B^2(s)$ at time $t$. We say that two states \emph{agree} on a connected component $C_i$ of $G_t$ if the restriction of the states to $C_i$ is the same.  Furthermore, denote by $\mathcal{G}[B^2(s)]$ the temporal subgraph with vertex set consisting of all vertices $v$ such that there exists a time $t$ where $(v,t)\in B^2(s)$ and time-edges $\{(uv,t):(u,t),(v,t)\in B^2(s) \text{ and } (uv,t)\in\mathcal{E}(\mathcal{G})\}$. 

\begin{defn}[$\twostep$ $(X,k)$ profile]
    For a static graph $G$ and set of vertices $S$, we denote by $G[S]$ the subgraph of $G$ induced by $S$. We define a $\twostep$ $(X,k)$ \emph{profile} of a bag $B^2(s)$ of a $\twostep$ TIM decomposition as a tuple consisting of
\begin{itemize}
    \item for each time $t$ such that there exist pairs $(v,t)$ in $B^2(s)$ (recall there are at most 3 such times), a labelling $l^t$ of the pairs with time $t$ to elements of $X$,
    \item for each connected component $C$ of each snapshot $G_t[B^2(s)]$ of $\mathcal{G}[B^2(s)]$, a vector $\textbf{v}_C$ with at most $k$ entries,
    \item a vector with at most $k$ entries denoted \textbf{total}.
\end{itemize}
\end{defn}

As with the $(X,k)$-component states, we define a $\twostep$ $(X,k)$ profile's restriction to a connected component $C$ of a snapshot at time $t$ in the a bag to be a pair consisting of $l^t|_C$, the restriction of the labelling at time $t$ to the vertices in $C$ and the vector $\textbf{v}_C$ associated to $C$. Note that, as with the restriction of $(X,k)$-component states, the restriction of a $\twostep$ $(X,k)$ profile to a connected component $C$ of a snapshot $G_t$ of a temporal graph gives us a $(X,k)$-state of the static graph induced by $C$. 

Let $\mathcal{G}^s$ be the temporal graph consisting of the vertices and time-edges which appear in bags in the subtree rooted at a node $s$, and $\mathcal{C}(G_t^s)$ denote the set of connected components in the snapshot $G^s_t$ of $\mathcal{G}^s$. The set $\mathcal{C}(\mathcal{G}^s)$ is defined as the set of component-time pairs $\bigcup_{0\leq t\leq \Lambda(\mathcal{G}^s)}\{(C,t) \,: \,C \in \mathcal{C}(G_t^s)\}$. Let $\mathfrak{C}^s$ be the set of all timed connected components $(C,t)\in\mathcal{C}(\mathcal{G}^s)$ such that, for all vertices $v$ in $C$, the pair $(v,t-1)$ is in a bag in the subtree of $T$ rooted at $s$.

\begin{defn}\label{def:realisable}
    A $\twostep$ $(X,k)$ profile $\sigma=(l^{t-1}, l^t, l^{t+1}, \textbf{v}_{1},\ldots,\textbf{v}_{|\mathcal{C}^s|}, \textbf{total})$ of a bag $B^2(s)$ is \emph{realisable} if and only if there exists a configuration $\varsigma$ which maps every timed connected component $(C,t)\in \mathfrak{C}^s$ to a $(X,k)$-state $\varsigma(C,t)=(l_{\varsigma}(C,t),\textbf{v}_{\varsigma}(C,t))$
    % : \mathcal{C}^s\to \mathfrak{L}_X\times \mathbb{Z}^k$ of $\mathcal{G}^s$, where $\varsigma(C)=(\mathfrak{l}^t_C,\mathfrak{v}^t_C)$, $\mathfrak{L}_X$ is the set of all labellings $\mathfrak{l}^t_C$ which label the vertex-time pairs of the connected component $C$ at time $t$ with elements in $X$
    such that
    \begin{itemize}
        \item for every connected component-time pair $(C,t)$ in $B^2(s)$, $\sigma|_C=\varsigma(C,t)$,
        \item for all connected components $C$ of $G_1^s$, $\textbf{St}(\varsigma(C,0),C,x)=\textbf{true}$,
        \item for all connected components $C$ of $G^s_{\Lambda(\mathcal{G})}$, $\textbf{Fin}(\varsigma(C,\Lambda(\mathcal{G})),C,x)=\textbf{true}$,
        \item for all times $t\in(0,\Lambda(\mathcal{G}))$, every connected component $C$ of $G_{t}^s$, $\textbf{Val}(\varsigma(C,t),C,x)=\textbf{true}$,
        \item for all pairs $(C,t)$ in $\mathfrak{C}^s$, $\textbf{Tr}(l^{\varsigma}_{t-1}|_C,l_{\varsigma}(C,t),C,x)=\textbf{true}$ where $l^{\varsigma}_{t-1}$ is the labelling of all vertices in $G^s_{t-1}$ such that its restriction to any connected component $C\in G_{t-1}^s$ is $l_{\varsigma}(C,t-1)$, and
        \item $\textbf{total}= \sum_{0\leq t\leq \Lambda(\mathcal{G}^s)} \sum _{C \in \mathcal{C}(G^s_t)} \textbf{v}_{\varsigma}(C,t)$.
    \end{itemize}
    Here we say that $\varsigma$ \emph{realises} $\sigma$.
\end{defn}

\begin{restatable}{lem}{realisablestatesleaves}
    \label{lem:realisable-states-leaves}
    Given an instance $x=(\mathcal{G},\beta)$ of a temporal problem, a $\twostep$ $(X,k)$ profile $\sigma=(l, \textbf{v}_{1},\ldots, \textbf{v}_{|\mathcal{C}^l|}, \textbf{total})$ of a bag $B^2(l)$ of a leaf node $l$ in a $\twostep$ TIM decomposition $(T,B^2)$ is \emph{realisable} if and only if 
   \begin{itemize}
       \item all vertex-time pairs in $B^2(l)$ are at time $0$ and $\textbf{St}(\sigma|_{C}, C,x)$ returns true for all connected components $C$ in $B^2(l)$ and $\textbf{total}=\sum_{1\leq i\leq |\mathcal{C}^l|} \textbf{v}_{i}$, or 
       \item all vertex-time pairs in $B^2(l)$ are at time $\Lambda$ and $\textbf{Fin}(\sigma|_{C}, C,x)$ returns true for all connected components $C$ in $B^2(s)$ and $\textbf{total}=\sum_{1\leq i\leq |\mathcal{C}^l|} \textbf{v}_{i}$.
        \end{itemize}
\end{restatable}
\begin{proof}
    We begin by showing that, if the criteria of the lemma hold, $\sigma$ is realisable. We do this by construction of a configuration $\varsigma$ of the subgraph induced by $B^2(l)$ which realises $\sigma$. For a connected component $C_i$ at time $t$ in $B^2(l)$, let $\varsigma(C_i,t)$ be the pair consisting of the labelling $l_{\varsigma}(C_i,t)=l|_{C_i}$ and vector $\textbf{v}_{\varsigma}(C_i,t)=\textbf{v}_i$. Then, by construction, $\sigma|_C=\varsigma(C,t)$ for all sets $C$ of vertices $v$ such that $(v,t)$ in $B^2(l)$ and $C$ is a connected component in $G^s_t$. 

    Furthermore, if all vertex-time pairs in $B^2(l)$ are at time $0$, we assume that $\textbf{St}(\sigma|_{C},C,x)$ returns true for all connected components $C$ in $B^2(l)$ and $\textbf{total}=\sum_{1\leq i\leq |\mathcal{C}^l|} \textbf{v}_{i}$. Therefore, $\textbf{St}(\varsigma(C,0), C,x)=\textbf{true}$ for all connected components $C\in B^2(l)$ and $\textbf{total}=\sum_{C\in\mathcal{C}^l_0}\textbf{v}_{\varsigma}(C,0)$.

    Similarly, if all vertex-time pairs in $B^2(l)$ are at time $\Lambda$, we assume that $\textbf{Fin}(\sigma|_{C},C,x)$ returns true for all connected components $C$ in $B^2(l)$ and $\textbf{total}=\sum_{1\leq i\leq |\mathcal{C}^l|} \textbf{v}_{i}$. Therefore, $\textbf{Fin}(\varsigma(C,\Lambda), C,x)=\textbf{true}$ for all connected components $C\in B^2(l)$ and $\textbf{total}=\sum_{C\in\mathcal{C}^l_{\Lambda}}\textbf{v}_{\varsigma}(C,\Lambda)$. Therefore, $\sigma$ is realised by $\varsigma$.

    Now suppose that $\sigma$ is realised by a configuration $\varsigma$ of $\mathcal{G}^l$. Then, for every connected component $C$ in $B^2(l)$, $\sigma|_C=\varsigma(C,t)$ where $t$ is the time in the vertex-time pairs in $B^2(l)$. By definition of $\varsigma$ realising $\sigma$, we also have that $\sum_{C\in \mathcal{C}(\mathcal{G}^l)}\textbf{v}_{\varsigma}(C,t)= \sum_{1\leq i\leq |\mathcal{C}^l|}\textbf{v}_i=\textbf{total}$. What remains to show is that, for all connected components $C$ in $B^2(l)$, if all times in $B^2(l)$ are 0, $\textbf{St}(\sigma|_{C}, C,x)$ returns true; and, if all vertex-time pairs in $B^2(l)$ are at time $\Lambda$, $\textbf{Fin}(\sigma|_{C}, C,x)=\textbf{true}$. Both of these conditions hold since $\sigma|_C=\varsigma(C,t)$ and the starting (respectively finishing) routine is true for all connected components in the bag if the times in the bag are 0 (resp.~$\Lambda$). Hence, the conditions of the lemma hold if and only if $\sigma$ is realisable.
\end{proof}

Note that, in the following lemma, we extend a labelling of the vertex-time pairs in a bag of a $\twostep$ TIM decomposition with labellings given by profiles of the children of that node. We require that any vertices which appear in both the bag and a bag of its child are labelled the same way. With this in mind, recall that this extension must be well defined as no vertex-time pair can appear in a bag and more than one of the bags of its children. Furthermore, note that we only count vectors in the $\textbf{total}$ sum if the connected component to which the vector is associated does not appear in a bag of any children of the current node in question; this prevents double counting.

\begin{restatable}{lem}{realisablestates}\label{lem:realisable-states}
     Given an instance $x=(\mathcal{G},\beta)$ of a temporal problem, a $\twostep$ $(X,k)$ profile $\sigma=(l^{t-1}, l^t, l^{t+1}, \textbf{v}_{1},\ldots, \textbf{v}_{|\mathcal{C}^s|}, \textbf{total})$ of a non-leaf bag $B^2(s)$ of a $\twostep$ TIM decomposition $(T,B^2)$ is realisable if and only if 
   \begin{itemize}
       \item for all times $t$ which appear in the bag $B^2(s)$ and all connected components $C$ of $G_{t}[B^2(s)]$, $\textbf{Val}(\sigma|_{C},C,x)$ returns \textbf{true}, and 
       \item for each child $s_c$ of $s$, there is a realisable $\twostep$ $(X,k)$ profile $\sigma_{s_c}$ such that, for all times $t$ which appear in the bag $B^2(s)$ and all connected components $C$ of $G_{t}[B^2(s)]$:
        \begin{itemize}
            % \item all vertex-time pairs that appear in both $B^2(s)$ and a child of $B^2(s)$ are labelled the same way by $\sigma$ and the state of the child in $S$,
            \item if the vertices in $C$ are in a pair at time $t$ in a child bag $B^2(s_c)$, $\sigma|_C=\sigma_{s_c}|_C$,
            \item for all connected components $C$ in $\mathcal{C}^s_t \cap \mathfrak{C}^s$, $\textbf{Tr}(l^{t-1}_{\sigma}|_C,l^t|_C,C,x)$ returns \textbf{true}, where the labelling $l^{t-1}_{\sigma}$ is the labelling that extends the labelling in $\sigma$ at time $t-1$ with the labellings at time $t-1$ in each $\sigma_{s_c}$, and
            \item $\textbf{total}$ is the sum of $\textbf{total}_c$ in each $\sigma_{s_c}$ and, for each $(C,t)$ in $B^2(s)$ that is not in $B^2(s_c)$ for any child $s_c$ of $s$, the vector $\textbf{v}_C$ in $\sigma_C$.
        \end{itemize}
        \end{itemize}
\end{restatable}
\begin{proof}
    We begin by showing that, if the criteria of the lemma hold, the profile $\sigma$ as described is realised by a configuration $\varsigma$ of the temporal subgraph $\mathcal{G}^s$ induced by the bags in the subtree rooted at $s$. By our assumption, there is a realisable profile for each child of $s$. Denote by $\sigma_c$ the profile of the bag of child $c$ and by  $\varsigma_c$ the configuration which realises the profile. We construct $\varsigma$ by combining the configurations which realise each $\sigma_c$ to the connected components in $B^2(s)$. Note that, by Observation~\ref{obs:two-step-component-in-children}, there exists a well-defined extension of all of the $\varsigma_c$ to $\mathcal{G}^s$ because no connected component-time pair is in a bag of a node in more than one of the subtrees rooted at the children of $s$. Furthermore, we assume from the conditions of the lemma that any timed connected component $C$ which appears in both $B^2(s)$ and a child bag $B^2(c)$ of $B^2(s)$ must have the property that $\sigma|_C=\sigma_c|_C$. 
    % For any vertex-time pairs that appear in both $B^2(s)$ and a child bag $B^2(c)$ of $B^(s)$, we also have that the labellings at time $t$ in $\sigma$ and $\sigma_c$ label $v$ the same way. 
    Therefore, for all component-time pairs $(C,t)$ in $\mathcal{C}(\mathcal{G}^s)$ the configuration 
    \[\varsigma(C,t) =
    \begin{cases} 
      \varsigma_c(C,t) & \text{if there is a child }c\text{ where } (C,t)\in\mathcal{C}(\mathcal{G}^c) \\
      \sigma|_C & \text{otherwise}
   \end{cases}\]
    is well-defined. Let $\varsigma(C,t) = (l_{\varsigma}(C,t),\textbf{v}_{\varsigma}(C,t))$. It is clear that, for all times $t$ and connected components of $\mathcal{G}^s_t[B^2(s)]$, $\sigma|_C=\varsigma(C,t)$. Because $\textbf{total}$ is the sum of the vectors corresponding to timed connected components which do not appear in any child bags and $\textbf{total}_c$ for each $\sigma_c$, $\textbf{total}_c$ for each profile $\sigma_c$ is the sum of $\textbf{v}_{\varsigma}(C,t)$ for all timed connected components $(C,t)$ in $\mathcal{G}^c$, and $\sigma|_C=\varsigma(C,t)$ for all components $C$ in $B^2(s)$, then $\textbf{total}$ is the sum of $\textbf{v}_{\varsigma}(C,t)$ for all timed connected components in $\mathcal{C}(\mathcal{G}^s)$.
    Since $\varsigma$ is an extension of the configurations which realise the profiles of the children of $s$, we need only check that the remaining criteria of $\varsigma$ realising $\sigma$ hold for the connected components in $B^2(s)$. All that is left is to check that for all timed connected components $(C,t)$ in $B^2(s)$, $\textbf{Val}(\varsigma(C,t),C,x)$; and, for all timed connected components in $\mathfrak{C}^s\cap B^2(s)$, $\textbf{Tr}(l^{\varsigma}_{t-1}|_C,l_{\varsigma}(C,t),C,x)=\textbf{true}$, where $l_{\varsigma}(C,t)$ is the labelling in $\varsigma(C,t)$ and $l^{\varsigma}_{t-1}$ is the labelling which extends all labellings at time $t-1$ in $\varsigma$ to $\mathcal{G}^s$. Since $\varsigma(C,t)=\sigma|_C$, and extensions of the labellings in each $\varsigma_c$ to the vertex-time pairs in $B^2(s)$ are the same labellings as $l^{t-1},l^t,l^{t+1}$ in $\sigma$, both statements are true by our earlier assumptions. Therefore, if the criteria of the lemma hold, $\sigma$ is realised by $\varsigma$ as constructed.

    Now suppose there is a realisable $\twostep$ $(X,k)$ profile $\sigma=(l^{t-1}, l^t, l^{t+1}, \textbf{v}_{1},\ldots, \textbf{v}_{|\mathcal{C}^s|}, \textbf{total})$ of $B^2(s)$. Let $\varsigma$ be the configuration which realises $\sigma$. Then, by definition $\sigma|_C=\varsigma(C,t)$ for all connected components $C$ of $G^s_t[B^2(s)]$. Therefore, for all connected components $C$ of $G^s_t[B^2(s)]$, $\textbf{Val}(\sigma|_C,C,x)=\textbf{true}$. For each child $c$ of $s$, let $\varsigma|_c$ be the restriction of $\varsigma$ to the temporal subgraph $\mathcal{G}^c$ induced by the bags in the subtree rooted at $c$. Then, let $\sigma_c$ be the profile of $B^2(c)$ realised by $\varsigma_c$. 
    % any vertex-time pair that appears in both $B^2(s)$ and a child bag $B^2(c)$ must be labelled the same way by $\sigma$ and $\sigma_c$. 
     It is clear from construction that, if a connected component in $B^2(s)$ appears entirely in a bag of a child $c$ of $s$, $\sigma|_C$ must be equal to $\sigma_c|_C$. Furthermore, since $\varsigma$ realises $\sigma$, for all connected component pairs $(C,t)\in\mathfrak{C}^s$, $\textbf{Tr}(l^{\varsigma}_{t-1}|_C, l_{\varsigma}(C,t),C,x)=\textbf{true}$. If we let $l^{t-1}_{\sigma}$ be the extension of all labellings at time $l^{t-1}$ in $\sigma$ and all $\sigma_c$, then $l^{\varsigma}_{t-1}|_C=l^{t-1}_{\sigma}|_C$. In addition, $l^t|_C=l_{\varsigma}(C,t)$. Hence, $\textbf{Tr}(l^{t-1}_{\sigma}|_C, l^t|_C,C,x)=\textbf{true}$ for all connected component pairs $(C,t)\in\mathfrak{C}^s$. Finally, since \textbf{total} is the sum of each vector assigned to the connected components in $\mathcal{C}(\mathcal{G}^s)$ which do not appear in any child bag $s_c$, and the vector associated to each connected component in $B^2(s)$ under $\sigma$ is equal to the vector assigned by $\varsigma$, $\textbf{total}$ must be the sum of $\textbf{total}_c$ in $\sigma_c$ for each child $c$ of $s$ and the vectors in $\sigma$ associated to timed connected components which appear in $B^2(s)$ and no child bags. Note that this enforces that, for each timed connected component in $\mathcal{G}^s$, the associated vector is counted exactly once in $\textbf{total}$.
    Thus, a $\twostep$ $(X,k)$ profile $\sigma$ is realisable if and only if the conditions of the lemma hold.
\end{proof}

\begin{restatable}{lemma}{TIMrootrealisable}\label{lem:TIM-root-realisable}
    Let $P$ be a $(X,k,f)$-component-exchangeable temporally uniform problem. Then an instance $x=(\mathcal{G},\beta)$ is a yes-instance of $P$ if and only if there is a realisable $\twostep$ $(X,k)$ profile $\sigma=(l^{t-1},l^t,l^{t+1},\textbf{v}_1,\ldots\textbf{v}_{|\mathcal{C}^r|}, \textbf{total})$ of the root bag $B^2(r)$ of the $\twostep$ TIM decomposition of $\mathcal{G}$ such that $\textbf{total}\leq \textbf{v}_{\text{upper}}$.
\end{restatable}
\begin{proof}
    We begin by supposing that $x=(\mathcal{G},\beta)$ is a yes-instance of $P$. Then, by definition of being $(X,k,f)$-component-exchangeable temporally uniform, there exists a sequence of $(X,k)$-component states $s_0,\ldots,s_{\Lambda}$ of the form $s_t=(l_t,\textbf{w}^t_1,\ldots,\textbf{w}^t_c, \nu_t)$ of each snapshot of $\mathcal{G}$ such that 
    \begin{itemize}
        \item for each connected component $C_1$ of $G_1$, $\textbf{St}(s_0|_{C_1},C_1,x)=\textbf{true}$;
        \item for each connected component $C_{\Lambda}$ of $G_{\Lambda}$, $\textbf{Fin}(s_{\Lambda}|_{C_{\Lambda}},C_{\Lambda},x)=\textbf{true}$;
        \item $\textbf{Tr}(l_{t-1}|_{C_t}, l_t|_{C_t}, C_t,x)=\textbf{true}$ where $l_t$ is the labelling of vertices of state $s_t$, for all times $1 \leq t \leq \Lambda$ and connected components $C_t$ of $G_t$;
        \item $\textbf{Val}(s_t|_{C_t}, C_t,x)=\textbf{true}$ for all times $1 < t < \Lambda$ and connected components $C_t$ of $G_t$; and
        \item the sum $\sum_{0\leq t\leq \Lambda} \sum_{C \in\mathcal{C}_t}\nu_{t} (C)\leq \textbf{v}_{\text{upper}}$.
    \end{itemize}
    From this, using Definition~\ref{def:realisable}, we can construct a configuration $\varsigma$ of $\mathcal{G}$ by letting $\varsigma(C,t)=s_t|_C$ for all connected components $C$ of all snapshots $G_t$ of $\mathcal{G}$. Then $\varsigma$ realises the profile $\sigma$ of $B^2(r)$ consisting of the restrictions of $l_t$ to the set of vertices $\{v \,:\,(v,t)\in B^2(s)\}$ for each time in the bag, the vectors assigned to each connected component in $B^2(r)$ by $\varsigma$, and the vector \textbf{total} which is given by the sum $\sum_{t\in [0,\Lambda]}\sum_{C\in\mathcal{C}(G_t)}\textbf{v}_{\varsigma}(C,t)$, where $\textbf{v}_{\varsigma}(C,t)$ is the vector given by $\varsigma(C,t)$.

    Now, suppose that there exists a $\twostep$ $(X,k)$ profile $\sigma=(l^{t-1},l^t,l^{t+1},\textbf{v}_1,\ldots\textbf{v}_{|\mathcal{C}^r|}, \textbf{total})$ of $B^2(r)$ for an instance $x$ of $P$ such that $\textbf{total}\leq \textbf{v}_{\text{upper}}$ and $\sigma$ is realisable. Let $\varsigma$ be the configuration of $\mathcal{G}$ which realises $\sigma$. Then we construct a sequence $s_0,\ldots, s_{\Lambda}$ of $(X,k)$-component states from $\varsigma$ by letting $l_{t}$ be the labelling that extends $l_{\varsigma}(C,t)$ for all connected components $C$ in $G_t$, and $\nu_t(C)=\textbf{v}_{\varsigma}(C,t)$, where $\textbf{v}_{\varsigma}(C,t)$ is the vector given by $\varsigma(C,t)$. Since the vector $\textbf{total}$ in $\sigma$ is the sum of all vectors of all timed connected components in $\mathcal{G}$ under $\varsigma$, $\textbf{total}$ must also be $\sum_{t\in[0,\Lambda]}\sum_{C\in\mathcal{C}_t}\nu_t(C)\leq \textbf{v}_{\text{upper}}$. Note that, for the root bag $B^2(r)$, $\mathfrak{C}^r$ is the set of all connected components in the graph ($\mathfrak{C}^r=\mathcal{C}(\mathcal{G}))$. Therefore, by definition, we must have that 
    \begin{itemize}
        \item for each connected component $C_1$ of $G_1$, $\textbf{St}(s_0|_{C_1},C_1,x)=\textbf{true}$;
        \item for each connected component $C_{\Lambda}$ of $G_{\Lambda}$, $\textbf{Fin}(s_{\Lambda}|_{C_{\Lambda}},C_{\Lambda},x)=\textbf{true}$;
        \item $\textbf{Tr}(l_{t-1}|_{C_t}, l_t|_{C_t}, C_t,x)=\textbf{true}$ where $l_t$ is the labelling of vertices of state $s_t$, for all times $1 \leq t \leq \Lambda$ and connected components $C_t$ of $G_t$; and
        \item $\textbf{Val}(s_t|_{C_t}, C_t,x)=\textbf{true}$ for all times $1 < t < \Lambda$ and connected components $C_t$ of $G_t$.
    \end{itemize}
    Therefore, $x$ must be a yes-instance of $P$. In conclusion, an instance $x$ is a yes-instance of a $(X,k,f)$-component-exchangeable temporally uniform problem if and only if there exists a realisable $\twostep$ $(X,k)$ profile of the root bag $B^2(r)$ of the $\twostep$ TIM decomposition of $\mathcal{G}$ such that $\textbf{total}\leq \textbf{v}_{\text{upper}}$.
\end{proof}

We are now ready to prove correctness of our second meta-algorithm by showing that it solves any $(X,k,f)$-component-exchangeable temporally uniform problem, and does so in a time that can bounded in terms of $f$, and the TIM width.

\end{toappendix}
\begin{theoremrep}\label{thm:component-temp-unif}
    Let $x=(\mathcal{G},\beta)$ be an instance of a $(X,k,f)$-component-exchangeable temporally uniform problem $P$ where $\mathcal{G}$ has $n$ vertices and lifetime $\Lambda$.
    % , and $|X| \geq 2$. 
    Given a TIM decomposition of $\mathcal{G}$ with width $\phi$, we
    % We 
    can determine whether $x$ is a yes-instance of $P$ in time 
    % $f(\phi,x)(n \Lambda)^{O(k)} (b|X|)^{O(k\phi^3)}$,
    $O\left(n \Lambda |X|^{12 \phi^3} (3b)^{12k \phi^3} (3\Lambda n )^{4k}\phi^9 k^2 f(\phi,x)\right)$, 
    where $\phi$ is the TIM width of $\mathcal{G}$ and $b \in \mathbb{N}$ is an upper bound on the absolute value of any entry of a vector in a $(X,k)$-component state of $\mathcal{G}$.
\end{theoremrep}
\begin{proof}
    By Lemma~\ref{lem:TIM-root-realisable}, an instance $x$ of $P$ is a yes-instance if and only if there exists a realisable $\twostep$ $(X,k)$ profile $\sigma$ of the root $r$ of the $\twostep$ TIM decomposition such that the vector \textbf{total} is bounded above by $\textbf{v}_{\text{upper}}$. To determine whether such a profile exists, we work from leaves of the decomposition to the root, finding realisable profiles such that the validity and transition routines return true for all connected components of each parent bag until we reach the root. To determine the runtime of finding such a profile, we first bound the number of profiles we must consider. Recall that by Observation~\ref{obs:tim-node-bound}, there are at most $\Lambda n$ nodes in the TIM decomposition of a temporal graph $\mathcal{G}$ with lifetime $\Lambda$ and $n$ vertices. Since the $\twostep$ TIM decomposition is indexed by the same tree $T$, this must also bound the number of nodes in the $\twostep$ TIM decomposition. As noted in the construction of the $\twostep$ TIM decomposition, if $\phi$ is the TIM width of $\mathcal{G}$, the width of the corresponding $\twostep$ TIM decomposition is at most $3\phi^2$.  Note also that the number of vertices in any timed component remains bounded by $\phi$.

    Given a bag $B^2(s)$ of the $\twostep$ TIM decomposition, a $\twostep$ $(X,k)$ profile consists of a labelling of the vertex-time pairs, a vector for each connected component of each snapshot in $B^2(s)$, and a vector \textbf{total}. There are at most $|X|^{3\phi^2}$ possible labellings of the vertex-time pairs in the bag. The number of elements in the bag also upper bounds the number of timed connected components in the bag. Thus, there are at most $3\phi^2$ vectors in the profile associated with timed components; each such vector has at most $k$ entries, and each entry has absolute value at most $b$.  Thus there are at most $(2b+1)^{3k\phi^2}$ possible combinations of these vectors.  The single vector \textbf{total} has at most $k$ entries, and each entry has absolute value at most $\Lambda n b$ (since there is a contribution of magnitude at most $b$ from every timed component in the relevant subtree), so the number of possible values for the vector \textbf{total} is $(2 \Lambda n b + 1)^k$.  Thus there are at most $|X|^{3 \phi^2} (2b+1)^{3k\phi^2} (2 \Lambda n b + 1)^k\leq |X|^{3 \phi^2}(3b)^{3k \phi^2}(3 \Lambda n b)^k = |X|^{3 \phi^2} (3b)^{4 k \phi^2} (\Lambda n)^k$ possible profiles of any bag $B^2(s)$.

    % Given a bag $B^2(s)$ of the $\twostep$ TIM decomposition, a $\twostep$ $(k,X)$ profile consists of a labelling of the vertex-time pairs, a vector for each connected component of each snapshot in $B^2(s)$, and a vector $\textbf{total}$. There are at most $|X|^{3\phi^2}$ possible labellings of the vertex-time pairs in the bag. The number of elements in the bag also upper bounds the number of timed connected components in the bag. Thus, there are at most $3\phi^2+1$ vectors in the profile with at most $k$ entries each. Since $b$ upper bounds the value of any entry in a vector of a profile, there are at most $(2b+1)^{2k}$ possibilities for vectors in any profile. There are at most $3\phi^2$ connected components in a $\twostep$ decomposition, hence there are at most $|X|^{3\phi^2}{\left((2b+1)^{k}\right)}^{3\phi^2}=|X|^{3\phi^2}(2b+1)^{3\phi^2k}$ possible profiles of any bag $B^2(s)$.

    % To determine whether a profile $\sigma$ of $B^2(s)$ is realisable, we must perform (some of) the four routines on each connected component in $B^2(s)$. 

    In the first step towards determining whether a profile is realisable, we must check exactly one of the validity, starting and finishing routines for each timed connected component in $B^2(s)$.  As stated earlier, there are at most $3\phi^2$ connected components in each bag of the $\twostep$ TIM decomposition. %Note that, for each component we must only check one of validity, the starting routine, and the finishing routine. 
    Since each timed component contains at most $\phi$ vertices, checking any of these three properties for a single component can be achieved in $f(\phi,x)$ time. Thus, the total time spent performing the starting, validity and finishing routines over all timed components in the bag is at $O(\phi^2 f(\phi,x))$.
    
    For the second step, when the node is not a leaf, we must determine whether there exists a set $S$ of realisable profiles for all the child bags with the following properties:
    \begin{enumerate}
        \item the restriction of the profiles to any timed connected component that appears in the bag and a child bag must must be the same,
        \item the transition routine returns true for the labellings when restricted to each relevant timed connected component in the bag, and
        \item the \textbf{total} vector of the profile of the bag is the sum of the vectors in that profile associated to connected components which do not appear in any child bags and the $\textbf{total}$ vectors in the set $S$.
    \end{enumerate}   
    It is too costly to consider simultaneously every possible value of the \textbf{total} vector for all of the children (of which there might be $2\phi$), so instead we identify all sets $S'$ of partial profiles -- in which we omit the value of the \textbf{total} vector -- for each child which meet conditions 1 and 2, then for each such set $S'$ we determine (using a simple dynamic program) whether there is a set of realisable profiles for all the children which is consistent with $S'$ and satisfies condition 3.  Recall from Observation~\ref{obs:tim-no-children} that there are at most $2 \phi$ children.  Thus, the number of possible combinations of partial profiles for all children is $|X|^{6 \phi^3} (3b)^{6k\phi^3}$.

    For each candidate set $S'$ of partial profiles, we need to check conditions 1 and 2.  We now proceed to bound the time required to do this for a fixed set $S'$.  First, consider the time required to check that the restrictions of profiles to timed connected components are consistent between the parent and children.  For a timed component with $i$ vertices, this can be done in time $O(i\phi k)$, since we must consider each vertex in the component and each component of the associated vector, and there are $O(\phi)$ children against which we need to make the comparison (by Observation~\ref{obs:tim-no-children}).  Recall from Observation~\ref{obs:two-step-width} that the sum of cardinalities of all timed components in a bag is $O(\phi^2)$; thus, summing over all components, we see that this check can be performed for the whole bag in time $O(\phi^3k)$.  Next, we consider the time required to perform the transition routine on all the required timed connected components. To determine which of the connected components we must apply the transition routine to, we must check that all vertices in a timed connected component appear in a pair with the previous time in either the current bag or one of its children. This can be done in $O(i\phi^3)$ time for a single component with $i$ vertices since there are at most $\phi$ vertices in a single connected component, at most $2\phi$ children to check, and at most $3\phi^2$ elements in a bag with which to compare the vertices. Summing over all timed connected components in the bag, we identify those to which we must apply the transition routine in time $O(\phi^5)$.  Having identified the relevant timed connected components, the transition routine takes time $O(f(|C|,x))\in O(f(\phi,x))$ for each component, and so time $O(\phi^2 f(\phi,x))$ in total.  Overall, therefore, computation related to the transition routine takes time $O(\phi^5 + \phi^2 f(\phi,x))$.  Combining these bounds, we see that we can identify all sets $S$ of partial profiles for the children that satisfy conditions 1 and 2 in time $O(|X|^{6 \phi^3}(3b)^{6k\phi^3} \phi^3k(\phi^5 + \phi^2 f(\phi,x)) = O(|X|^{6 \phi^3} (3b)^{6k \phi^3} \phi^8 k f(\phi,x)$.

    Now we bound the time required to check, given a fixed set $S'$ of partial profiles, whether there exists a consistent set of realisable profiles for the child bags which satisfies condition 3.  Specifically, we want to determine whether there exists a set of realisable profiles for all the child bags which is consistent with $S'$ and moreover has the property that the sum of the \textbf{total} vectors for each profile gives the desired value (which is the guessed value of \textbf{total} for the parent, minus the sum of guessed component vectors in the parent bag). Note that, to avoid double counting, we only include the vectors associated to connected components which do not appear in any child bags. Finding this set of connected components requires $O(\phi^3)$ time, since we must check $O(\phi^2)$ components against $O(\phi)$ children. Computing the target value requires $O(\phi^2k)$ time as we have $O(\phi^2)$ components in a bag with an associated vector with $k$ entries. We check that the sum of the vectors gives the required result using a simple dynamic program.  We fix an arbitrary ordering of the $c$ child bags, and each entry of the table is indexed by some $i \in [c]$ and a $k$-element vector in which each entry has absolute value at most $n \Lambda b$; note that there are at most $2 \phi (2 \Lambda n b + 1)^k$ entries in the table.  We will set the entry indexed by $(i,\mathbf{r})$ to true if and only if there exist realisable profiles for the first $i$ children that are consistent with $S$ and such that the sum of their \textbf{total} vectors is equal to $\mathbf{r}$.  It is straightforward to see that we can initialise all entries of the form $(1,\mathbf{r})$ in time $O(|X|^{3 \phi^2} (3b)^{3k\phi^2}(3\Lambda n )^k k)$, as it suffices to examine the set of realisable profiles for the first child, and compare the $k$ entries of the \textbf{total} vector against $\mathbf{r}$.  Moreover, given all entries of the form $(i-1,\mathbf{r})$, we can compute any entry of the form $(i,\mathbf{r}')$ in time $O(|X|^{3 \phi^2} (3b)^{3k\phi^2}(3\Lambda n )^{2k} k)$: we iterate over all realisable profiles for the $i^{th}$ child that are consistent with $S'$, and if we find such a profile such that the \textbf{total} vector takes value $\mathbf{r}_i$ and there is a true table entry $(i-1,\mathbf{r})$ such that $\mathbf{r} + \mathbf{r}_i = \mathbf{r}'$, we set the entry $(i,\mathbf{r}')$ to true.  Thus, we can complete the entire table in time $O(\phi k |X|^{3 \phi^2} (3b)^{3k\phi^2}(3\Lambda n)^{3k})$.  Having completed the table, we conclude that the desired set of realisable child bag profiles exists if and only if the entry $(c,\mathbf{r})$ is true, where $\mathbf{r}$ is the desired value for the sum of target vectors.

    Therefore, to perform the checks required to determine there are profiles of the child bags with the desired properties, we iterate over all candidate sets $S'$, and for each we find the set of realisable states for each child that are consistent with the choice of $S'$, then we perform the dynamic program to check whether any combination of these gives the correct $\textbf{total}$ vector. It therefore follows that we can determine whether there exists a set of realisable profiles for all the child bags that satisfies the three conditions in time 
    \[O( |X|^{9 \phi^3} (3b)^{9k \phi^3} (3\Lambda n)^{3k}\phi^9 k^2 f(\phi,x).\]  
    Thus, we can determine whether a given profile is realisable in time 
    \[O(|X|^{9 \phi^3} (3b)^{9k \phi^3} (3\Lambda n )^{3k}\phi^9 k^2 f(\phi,x) + \phi^2 f(\phi,x)) \]
    \[= O(|X|^{9 \phi^3} (3b)^{9k \phi^3} (3\Lambda n  )^{3k}\phi^9 k^2 f(\phi,x),\] 
    and summing over all possible profiles for the bag we see that we can compute the set of realisable profiles for a bag in time
    \[ O\left(|X|^{12 \phi^3} (3b)^{12k \phi^3} (3\Lambda n )^{4k}\phi^9 k^2 f(\phi,x)\right). \]
    Finally, summing over all nodes, we see that we can compute the set of realisable profiles for the root and hence solve the problem in time
    \[
        O\left(n \Lambda |X|^{12 \phi^3} (3b)^{12k \phi^3} (3\Lambda n )^{4k}\phi^9 k^2 f(\phi,x)\right). 
         % = f(\phi,x)(n \Lambda)^{O(k)} (b |X|)^{O(k\phi^3)},
    \]
    % since $|X|\geq 2$.
  This gives the result.
\end{proof}
% Applying this theorem to \textsc{Temporal Hamiltonian Path} gives us that the problem is in FPT with respect to TIM width. 

We now give a characterisation for problems in FPT with respect to TIM width.

% \begin{corollary}\label{cor:tim-meta}
%     Let $P$ be a $(k,X,f)$-component-exchangeable temporally uniform problem such that $f$ is a computable function, $k$ is a constant, $|X|$ and the maximum value of any variable in a $(k,X)$-component state $b$ are a function of $\phi$ alone, where $\phi$ is the TIM width of the input temporal graph. 
%     Then $P$ is in FPT with respect to the TIM width of the input temporal graph.
% \end{corollary}

% \begin{restatable}{theorem}{componenttempunif}\label{thm:component-temp-unif}
%     Let $x$ be an instance, involving a temporal graph with $n$ vertices and lifetime $\Lambda$, of a $(k,X)$-component-exchangeable temporally uniform problem $P$ (where $|X| \geq 2$) with starting routine \textbf{St}, finishing routine \textbf{Fin}, transition routine $\textbf{Tr}$ and validity routine $\textbf{Val}$, along with a vector of upper bounds $\textbf{v}_{\text{upper}}$. We can determine if $x$ is a yes-instance of $P$ in time $(n \Lambda)^{O(k)} (b|X|)^{O(k\phi^3)} p(\phi)$,
%     %$O(n\Lambda \phi^2p(n)+|X|^{6\phi^2}b^{6\phi^2 k}(\phi^3k+\phi^2 p(n)))$, 
%     where $\phi$ is the TIM width, $b$ is the maximum absolute value of any entry of a vector in a $(k,X)$-component state of $\mathcal{G}$, and $p$ is a function such that \textbf{St}, \textbf{Fin}, $\textbf{Tr}$ and $\textbf{Val}$ all run in $O(p(q))$ time on components with $q$ vertices.
% \end{restatable}

\begin{theoremrep}\label{thm:TIM-iff}
    Let P be a problem that takes $x=(\mathcal{G},\beta)$ as input, where $\mathcal{G}$ is a temporal graph with a given TIM decomposition of width $\phi$ and $\beta$ is a string. $P$ is in FPT with respect to $\phi$ if and only if it is a $(X,k,f)$-component-exchangeable temporally uniform problem and there exists a computable function $g$ such that for each connected component $C$ of a snapshot of $\mathcal{G}$, $f(C,x)\leq g(\phi)|x|^{O(1)}$ where $k$ is a constant, and $|X|$ and the maximum value $b$ of any variable in a $(X,k)$-component state are bounded by a function of $\phi$ alone. 
\end{theoremrep}
\begin{proof}
    Using Theorem~\ref{thm:component-temp-unif}, we get that if $P$ is $(X,k,f)$-locally temporally uniform where 
    $k$ is a constant, $b$ and $|X|$ are a function of the TIM width alone, and, for every connected component $C$ of a snapshot of $\mathcal{G}$, the time required for each subroutine is bounded by $f(|C|,x)=g(\phi)|x|^{O(1)}$ for a computable function $g$, then it is in FPT with respect to $\phi$.

    We now show the reverse direction. Suppose there exists an fpt-algorithm $A$ for $P$ with respect to $\phi$. Let the runtime of $A$ be $a(\phi)\text{poly}(n,\Lambda,|\beta|)$ for some computable function $a$. Then, we argue that $P$ must be $(X,1,a)$-component-exchangeable temporally uniform where $X$ is a set consisting of a single label. To show this we construct the states required and prove that an instance $x=(\mathcal{G},\beta)$ is a yes-instance if and only if 
    there exists a sequence of states $s_0,\ldots,s_{\Lambda}$ such that 
    \begin{enumerate}
        \item for each connected component $C_1$ of $G_1$, $\textbf{St}(s_0|_{C_1},C_1,x)=\textbf{true}$;
        \item for each connected component $C_{\Lambda}$ of $G_{\Lambda}$, $\textbf{Fin}(s_{\Lambda}|_{C_{\Lambda}},C_{\Lambda},x)=\textbf{true}$;
        \item $\textbf{Tr}(l_{t-1}|_{C_t}, l_t|_{C_t}, C_t,x)=\textbf{true}$ where $l_t$ is the labelling of vertices of state $s_t$, for all times $1 \leq t \leq \Lambda$ and connected components $C_t$ of $G_t$;
        \item $\textbf{Val}(s_t|_{C_t}, C_t,x)=\textbf{true}$ for all times $1 \leq t < \Lambda$ and connected components $C_t$ of $G_t$; and
        \item the sum of vectors satisfies $\sum_{0\leq t\leq \Lambda} \sum_{C \in\mathcal{C}_t}\nu_{s_t} (C)\leq \textbf{\emph{v}}_{\text{upper}}$.
    \end{enumerate} 
    Our set of labels consists of a single label, call it $U$. The counter vector for every connected component consists of a vector with one entry, let that entry be $0$. Also, let $\textbf{\emph{v}}_{\text{upper}}$ be a zero vector. 
    The validity and starting routines are the same. They return true if and only if the vertices of the connected component have label $U$, and the vector associated to the connected component is a zero vector with one entry. 
    Our transition routine returns true if and only if the labelling of the vertices by the two states are the same, which must always be the case given our description of the states. This leaves the finishing routine. This is the algorithm $A$ with input $x$. It is clear that the finishing routine returns true if and only if $x$ is a yes-instance. It follows from the description of the states that $k$, $|X|$, and $b$ are all constants. Therefore, the statement holds.
\end{proof}

% \begin{corollary}
%     \textsc{Weak Nonstrict Equitable Temporally Connected Partition} is FPT-component-exchangeable temporally uniform.
% \end{corollary}

\section{Applications of VIM meta-algorithm}\label{sec:applications-vim}
% \label{sec:applications}

Using the machinery we have defined, it is possible to find tractable algorithms for a range of problems. In order to show that a given problem admits a tractable algorithm when parameterised by VIM width, we show that it is locally temporally uniform by providing efficient transition and accepting routines. We then use \Cref{thm:loctempunf} to obtain a tractable algorithm for our problem, thus negating the need to construct a dynamic programming algorithm from scratch. Similarly, to show that a problem is tractable with respect to TIM width, we need only prove it is component-exchangeable temporally uniform by providing starting, finishing, validity and transition routines, and a vector to bound the vectors in the states. This means that we can apply the meta-algorithm without needing to use the TIM decomposition directly, and thus we only need to think about the information that must be stored for every connected component. This gives a much more intuitive way of showing tractability with respect to TIM width than direct proof. 
% Whilst many problems can be shown to be in FPT by showing they are FPT-locally temporally uniform, there are natural examples for which this does not work. %it is worth noting that this is not always the case. %Any problem where the state of a vertex at a given time depends on vertices outside of their active interval at that time is not locally temporally uniform. 
% For example, consider a problem where we are given a temporal graph $\mathcal{G}$, a set of vertices $V$, and an integer $k$, and wish to determine if there exists a set of $V'$ of at least $k$ vertices such that each vertex in $V$ can be reached by every vertex in $V'$. This problem cannot be FPT-locally temporally uniform as it requires either one of the subroutines has a run time which is not fpt with respect to the bag size, or the set $X$ of states is large enough to encode precisely which subset of vertices in the whole graph reaches the labelled vertex at a given time.%, since we need to know if the set of vertices that reach a vertex $v \in V$ contains all of the vertices in a candidate $V'$.
% 
% Theorem~\ref{thm:loctempunf} gives analogous results, with better running time bounds, for the same three problems parameterised instead by VIM width.  We have further identified one problem for which is in FPT with respect to VIM width but is intractable when parameterised by TIM width (although it is in FPT when parameterised simultaneously by TIM width and lifetime).
% 
% We begin by giving the full details of the proof that we can apply our meta-algorithm to \textsc{Temporal Hamilitonian Path}.
We apply our VIM meta-algorithm to two problems: temporal analogues of \textsc{Hamiltonian Path} and \textsc{Dominating Set}. 

\begin{restatable}{thm}{vimham}\label{thm:vim-ham}
    \textsc{Temporal Hamiltonian Path} can be solved in time $
    O(\Lambda nnn^23^{2\omega})=
    O(\Lambda n^43^{2\omega})$, where $\Lambda$ is the lifetime of the input temporal graph, $n$ the number of vertices, and $\omega$ the VIM width.
\end{restatable}

\begin{restatable}{thm}{vimdom}\label{thm:vim-dom}
    \textsc{Temporal Dominating Set} can be solved in time $O(\Lambda n^2 2^{\omega}(n\Lambda)^42^{2\omega})=O(\Lambda^5 n^62^{3\omega})$, where $\Lambda$ is the lifetime of the input temporal graph, $n$ the number of vertices, and $\omega$ the VIM width.
\end{restatable}

\begin{toappendix}

\subsection{Temporal Hamiltonian Path}\label{sec:ham-path-vim}
% We produce an algorithm that will find any temporal Hamiltonian path beginning on a vertex in $F_1$, the first entry in the VIM sequence. If no such path exists, then we can re-run our algorithm on $\mathcal{G}$, except with all temporal edges active on timestep $1$ removed, and all other temporal edges active one timestep earlier. If any temporal Hamiltonian path exists, we would be able to find it in this manner using at most $O(\Lambda)$ executions of the algorithm.
We begin by applying our VIM meta-algorithm to \textsc{Temporal Hamiltonian Path}. Recall the formal definition of the problem.

\begin{nolinenumbers}
\decisionproblem{Temporal Hamiltonian Path}{A temporal graph $\mathcal{G}$.}{Does there exist a strict temporal path containing every vertex in $\mathcal{G}$?}
\end{nolinenumbers}
Note that the only input for this problem is the temporal graph itself. In our definitions of the meta-algorithms, we use a string $\beta$ to encode the input of a given problem which is not the temporal graph. In this case $\beta$ would be the empty string. For ease, we omit $\beta$ in this section.

Throughout this section, we use the notion of an \emph{empty path}. This is a path consisting of no vertices or edges. We use this for the case where the temporal Hamiltonian path starts on a vertex whose active interval begins at a time later than~1. For ease, we assume without loss of generality that any non-empty path consists of at least one edge, and that there are at least two vertices in the input temporal graph. Let the \emph{arrival time} of a temporal path be the time of the final time-edge in the path. We use the convention that empty paths have arrival time~0.

%Not ideal to use V as this is normally the whole vertex set
We use $(X,1)$-states, where the label set $X=\{\visited,\unvisited,\current\}$, and the counter vector contains a single integer $h$, which counts the total number of visited vertices. We will define our transition routine and initial states such that each state produced by repeated applications of the transition routine corresponds to the existence of a temporal path that traverses $h$ vertices, such that if there is a vertex given label $\current$, the path is currently at that vertex, and any vertices labelled $\visited$ are traversed by the path. The accepting routine then returns true if $h=|V(\mathcal{G})|$. We now show that \textsc{Temporal Hamiltonian Path} is locally temporally uniform, by giving the transition and acceptance routines, and the set of initial states.

%C_1 \setminus C_2 etc

% \resetlinenumber[1135]
Given a state $(l, (h))$ and an instance of \textsc{Temporal Hamiltonian Path}, that is a temporal graph $\mathcal{G}$, the acceptance routine runs in constant time, returning true if and only if $h = |V(\mathcal{G})|$.
% We use $|F_0|+1$ initial states. For each vertex $v \in F_0$, construct a state for $F_0$ where $v$ is labelled $\current$, all other vertices in $F_0$ are labelled $\unvisited$, and $h=1$. The other state labels every vertex with $\unvisited$ and sets $h=0$. The algorithm which returns this set of states runs in time linear in the size of $|F_0|$ which is bounded above by VIM width.
We use one initial state where each vertex is labelled with $\unvisited$ and $h=0$. The algorithm which returns the set of initial states therefore runs in linear time in $n$.

We now show that there exists a correspondence between the temporal paths on $\mathcal{G}$ and sequences of states beginning with initial states and related by the transition routine. We say that a sequence $s_0, ..., s_t$ of states \textit{corresponds} to a (potentially empty) temporal path if and only if:
\begin{enumerate}
    \item $s_0$ is an initial state,
    \item $\textbf{Tr}(s_{i-1}, s_i, G_i)=\textbf{true}$ for every $1 \leq i \leq t$,
    \item $s_t$ gives at most one vertex label $\current$, and this is the final vertex on the path, and the path has arrival time $t' \leq t$, and
    \item the vertices traversed by the path (not including the final vertex) are exactly the vertices given label $\visited$ by $s_t$,
    \item the value of $h$ given by $s_t$ is equal to the number of vertices traversed by the path. 
\end{enumerate}
\begin{restatable}{lemma}{vimhamiltoniff}
    For any timestep $t$, there exists a (potentially empty) temporal path arriving on a timestep $t' \leq t$ if and only if there exists a corresponding sequence of states $s_0, \ldots, s_t$.
\end{restatable}
\begin{proof}
    We proceed by induction on the timestep $t$, with base case $t=0$. Since there are no edges active at time~0, any temporal path corresponding to an arriving by time~0 must be empty. Furthermore, the only initial state $s_0$ labels all vertices $\unvisited$ and sets $h=0$. Therefore, there is a path corresponding to $s_0$ if and only if $s_0$ is an initial state. Thus the base case holds.

    For induction, assume that for $t''<t$ there exists a temporal path arriving at time $t''$ at the latest if and only if there exists a corresponding sequence of states $s_0,\ldots,s_{t''}$. Now suppose that there is a (potentially empty) temporal path $P$ arriving on a timestep $t' \leq t$ at a vertex $v$, and traversing $\ell$ vertices. We now show that there exists a corresponding sequence of states. If $t' < t$, then, by induction, there exists a sequence of states $s_0, ... s_{t-1}$ corresponding to $P$. Now take $s_t = s_{t-1}$, and see that $\textbf{Tr}(s_{t-1}, s_t, G_t[F_t])=\textbf{true}$ as line 10 of \Cref{alg:temp-ham} will return true.

    Now let $t'=t$. There are two cases to consider: $P$ consists of a single time-edge at time $t$, and $P$ consists of more than one time-edge. We first consider the former. By induction, there is a sequence of states $s_0,\ldots, s_{t-1}$ corresponding to the empty path found by removing the only time-edge from $P$. Let $s_t$ be the state where the final vertex on the path is labelled $\current$, the other vertex traversed by $P$ is labelled $\visited$, all other vertices are labelled $\unvisited$ and $h=2$. Since $s_0,\ldots, s_{t-1}$ is a sequence corresponding to an empty path, all vertices must be labelled $\unvisited$ by $s_{t-1}$ and their counters must all be~0. Therefore, Algorithm~\ref{alg:temp-ham} must return \textbf{true} in line 5 with inputs $G_t$, $s_{t-1}$, and $s_t$.
    
    If $P$ has length at least~2 and $t'=t$, there must exist some non-empty temporal path arriving on a timestep $t' \leq t-1$ at a vertex $u$ adjacent to $v$ on timestep $t$, and hence by induction there exists a sequence of states $s_0, \ldots, s_{t-1}$ corresponding to this path. 
    Let $s_t$ be the state that has a value of $h$ one greater than $s_{t-1}$, gives label $\current$ to $v$, label $\visited$ to $u$, and labels all other vertices as they are by $s_{t-1}$. See that $\{u,v\} \in E(G_t[F_t])$, as $P$ traverses $\{u,v\}$ on timestep $t$, and therefore $t \in \lambda(\{u,v\})$. Therefore $\textbf{Tr}(s_{t-1}, s_t, G_t[F_t])=\textbf{true}$, as line 8 of \Cref{alg:temp-ham} will return true.
    
    Now assume there is a sequence of states $s_0, \ldots, s_t$ such that $s_0$ is an initial state, and $\textbf{Tr}(s_{i-1}, s_i, G_t[F_t])=\textbf{true}$ for every $1 \leq i \leq t$, and if such a vertex exists let $v$ be the vertex given label $\current$ by $s_t$ (if one exists). Now consider the sequence $s_0, ..., s_{t-1}$. By induction, there must exist a temporal path $P$ corresponding to this sequence. If $P$ is non-empty, let $u$ be the final vertex on this path. See that $\textbf{Tr}(s_{t-1}, s_t, G_t[F_t])=\textbf{true}$. If this is the case because line 5 of \Cref{alg:temp-ham} returns true, then $P$ must be empty and $h=2$. Let $u$ be the vertex labelled $\visited$ by $s_t$. Then, the edge $\{u,v\}$ must be in $G_t$, and the temporal path consisting of the time-edge $(u,v,t)$ must traverse $2$ vertices, visit $u$ and end at $v$. Therefore, there is a temporal path corresponding to the sequence of states.
    If the transition routine returns true because  of line 8 of \Cref{alg:temp-ham} returns true then $\{u,v\} \in G_t$, and therefore $t \in \lambda(\{u,v\})$. There is then a temporal path that traverses the same edges as $P$, before traversing $\{u,v\}$ on timestep $t$. This path will traverse all the vertices traversed by $P$, along with the additional vertex $u$, and therefore corresponds to $s_t$.
    Otherwise, if line 10 of \Cref{alg:temp-ham} returns true, then $s_t = s_{t-1}$, and therefore $P$ corresponds to $s_0, ..., s_t$.
\end{proof}
\vimhamiltonthm*

\begin{proof}
    The only initial state gives label $\unvisited$ to all vertices, thus all initial states give label $\unvisited$ to any vertex not in $F_0$, as required. 

    Line 5 of \Cref{alg:temp-ham} is one of two lines of the algorithm that returns true if the two input states are labelled differently. This line returns true if only the two vertices $v_2$ and $c_2$ have different labels, with all other vertices labelled identically. Furthermore if line 5 returns true then $\{v_2, c_2\} \in E(G)$, and neither $v_2$ and $c_2$ are isolated in $G$, and the algorithm returns true only if all isolated vertices in $G$ are given the same label as required.
    
    The other line to return true if the two input states are labelled differently is line 8. This line returns true if only the two vertices $c_1$ and $c_2$ have different labels, with all other vertices labelled identically. Furthermore if line 8 returns true then $\{c_1, c_2\} \in E(G)$, and neither $c_1$ nor $c_2$ are isolated in $G$, and the algorithm returns true only if all isolated vertices in $G$ are given the same label as required.

    Consider any graph $G$ and pair of states $s$ and $s'$ such that $\textbf{Tr}(s, s', G)=\textbf{true}$. Let $C_s$, $V_s$, $U_s$ be the sets of vertices labelled $\current$, $\visited$ and $\unvisited$ by $s$ respectively, and equivalently for $C_{s'}$, $V_{s'}$, $U_{s'}$ and $s'$. Now consider any vertex $v$ isolated in $G$. If $v \in C_s$, then $v \in C_{s'}$, as $C_s \setminus C_{s'}$ only contains one vertex, and this vertex is not isolated in $G$. If $v \in C_{s'}$, then similarly $v \in C_s$ as $C_{s'} \setminus C_s$ only contains one vertex, and this vertex is not isolated in $G$. If $v \in V_s$, then $v \in V_{s'}$ as $V_s \subseteq V_{s'}$. If $v \in V_{s'}$, then $v \in V_s$, as $V_{s'} \setminus V_s$ only contains one vertex, and this vertex is not isolated in $G$. Finally, as $C_s, V_s, U_s$ partition the vertices of $G$, and so does $C_{s'}, V_{s'}, U_{s'}$, we must have that if $v \in U_s$ if and only if $v \in U_{s'}$. Therefore $s'$ gives the same label as $s$ to every isolated vertex of $G$.

    Consider any graph $G$ and a quadruple of states $r$, $r'$, $s$, and $s'$ such that $r$ and $s$ agree on the non-isolated vertices in $G$, $r'$ and $s'$ agree on the non-isolated vertices in $G$, and the pairs $s,s'$ and $r,r'$ both give the same label to every isolated vertex not in $G$. Assume without loss of generality that $\textbf{Tr}(r, r', G)=\textbf{true}$.
    
    If this is because line 10 of \Cref{alg:temp-ham} returns true, see that $r=r'$. Then as $s$ agrees with $r$ on the non-isolated vertices of $G$, $s$ also agrees with $r'$ on the non-isolated vertices of $G$, and $r'$ agrees with $s'$ on the non-isolated vertices of $G$. Therefore $s$ agrees with $s'$ on the non-isolated vertices of $G$, and by definition $s$ and $s'$ give the same label to every isolated vertex of $G$, and so $s=s'$ and line 10 of \Cref{alg:temp-ham} will return true when given $s$, $s'$ and $G$ as input.
    
    Otherwise, if $\textbf{Tr}(r, r', G)=\textbf{true}$ because either line 5 or 8 of \Cref{alg:temp-ham} returns true, let $I$ be the non-isolated vertices in $G$, and $C_s$, $C_{s'}$, $C_r$, and $C_{r'}$ be the vertices labelled $\current$ by $s$, $s'$, $r$, and $r'$ respectively. See that $C_s \setminus I = C_{s'} \setminus I$, and therefore $C_s \setminus C_{s'} = (C_s \cap I) \setminus (C_{s'} \cap I)$, and $C_{s'} \setminus C_s = (C_s \cap I) \setminus (C_{s'} \cap I)$. Furthermore, $C_s \cap I = C_r \cap I$ and $C_{s'} \cap I = C_{r'} \cap I$, and therefore $C_s \setminus C_{s'} = (C_r \cap I) \setminus (C_{r'} \cap I)$, and $C_{s'} \setminus C_{s} = (C_{r'} \cap I) \setminus (C_{r} \cap I)$. Then as $r$ and $r'$ give the same label to every vertex not in $I$, $C_s \setminus C_{s'} = C_r \setminus C_{r'}$, and $C_{s'} \setminus C_s = C_{r'} \setminus C_r$.

    Thus, if Algorithm~\ref{alg:temp-ham} returns true in line 5, $C_{s'} \setminus C_s$ and $C_{r'} \setminus C_r$ contain the same vertex $c_2$ and $C_{s} \setminus C_{s'}$ and $C_{r} \setminus C_{r'}$ are both empty. As $\textbf{Tr}(r, r', G)=\textbf{true}$, we have that $\{u, c_2\} \in E(G)$. Furthermore, the vertex $u$ is given label $\visited$ in $s'$ since it has this label in $r'$ and they agree on non-isolated vertices. Any vertex $v \in I$ and not equal to $u$ or $c_2$ is given the same label by $s$ and $s'$, as $r$ and $r'$ give the same label to $v$, and $s$ agrees with $r$ on $I$, and $s'$ agrees with $r'$ on $I$. Any vertex $v \notin I$ is given the same label by $s$ and $s'$ by definition, and hence $V_s \cup \{u\} = V_{s'}$. Since $s$ and $r$, and $s'$ and $r'$ agree on the vertices in $I$, the counter for both $s$ and $r$ must be 0 and the counter for $s'$ and $r'$ must be 2. Therefore, $\textbf{Tr}(s, s', G)=\textbf{true}$
    
    If Algorithm~\ref{alg:temp-ham} returns true in line 8, $C_s \setminus C_{s'}$ and $C_{s'} \setminus C_s$ both contain the same single vertices $c_1$ and $c_2$ as $C_r \setminus C_{r'}$ and $C_{r'} \setminus C_r$ respectively. We have that $\{c_1, c_2\} \in E(G)$ as $\textbf{Tr}(r, r', G)=\textbf{true}$, and $c_2$ is given label $\unvisited$ by $s$, as it is given label $\unvisited$ by $r$, $c_2 \in I$, and $s$ and $r$ agree on $I$. Also, $s'$ has the same value of $h$ as $r'$, and $s$ has the same value of $h$ as $r$, so $h_{s'} = h_s + 1$. Finally, $c_1$ is given label $\visited$ by $s'$, as it is given label $\visited$ by $r'$, and $s'$ and $r'$ agree on $I$. Any vertex $v \in I$ and not equal to $c_1$ or $c_2$ is given the same label by $s$ and $s'$, as $r$ and $r'$ give the same label to $v$, and $s$ agrees with $r$ on $I$, and $s'$ agrees with $r'$ on $I$. Any vertex $v \notin I$ is given the same label by $s$ and $s'$ by definition, and hence $V_s \cup c_1 = V_{s'}$, $\textbf{Tr}(s, s', G)=\textbf{true}$ because line 8 of \Cref{alg:temp-ham} returns true when given $s$, $s'$ and $G$ as input. Note that the transition routine runs in $O(n)$ time.

    The starting routine needs only output a single state where every vertex is labelled $\unvisited$ and the counter is~0. Therefore, it runs in $O(n)$ time. Finally, see that the acceptance routine checks only the value of a counter variable. Therefore it runs in constant time and, if it returns true when given a state $s$, then it will return true when given any state agreeing with $s$ on any vertex set.
\end{proof}
Then, as \Cref{alg:temp-ham} runs in time $O(n)$ by checking the label on each vertex in turn, and we use $1$ counter variable of size at most $n$ and 3 labels, we finally obtain our result.

\vimham*

% \begin{corollary}\label{cor:timham}
%     \textsc{Temporal Hamiltonian Path} can be solved in time $O(n\Lambda \phi^2\phi+3^{6\phi^2}1^{6\phi^2}(\phi^3+\phi^2 n))=O(n\Lambda \phi^3+3^{6\phi^2}(\phi^3+\phi^2 n))$, where $\Lambda$ is the lifetime of the input temporal graph, $n$ the number of vertices, and $\phi$ the TIM width.
% \end{corollary}

 % \subsection{Temporal Dominating Set}\label{sec:dominating}

\subsection{Temporal Dominating Set}\label{sec:vimw-dominating}
We now consider a temporal analogue of \textsc{Dominating Set} given by Casteigts and Flocchini~\cite{casteigts_deterministic_2013}. This problem asks if it is possible to find a set $D$ of size $h$ or less consisting of vertex-time pairs (vertex appearances) such that every vertex $v$ is \textit{covered}, that is it either appears in a pair in $D$, or there exists a $(u, t) \in D$ such that $v$ is adjacent to $u$ on timestep $t$. We assume that the underlying graph $\mathcal{G}{\downarrow}$ contains no isolated vertices, and that for any vertex appearance $(u, t) \in D$, $u$ has an incident edge active on timestep $t$. If the graph contains an isolated vertex, then it must be included in any dominating set; we can therefore find an equivalent instance of \textsc{Temporal Dominating Set} by removing the isolated vertices and decrementing the number of appearances allowed in the dominating set. Adding a vertex appearance when it has no incident edges, means that that appearance can only dominate the vertex itself. Therefore, we can assume that all vertex appearances in a solution to \textsc{Temporal Dominating Set} have an active incident edge, as swapping a vertex appearance without an active incident edge for a time where there is an incident edge gives a solution that is no worse than the original.

\begin{nolinenumbers}
\decisionproblem{Temporal Dominating Set}{A temporal graph $\mathcal{G}$ and an integer $h$.}{Does there exist a set $D$ of vertex appearances such that $|D|\leq h$ and the appearances in $D$ cover every vertex in $\mathcal{G}$?} 
\end{nolinenumbers}

This problem has recently been studied under the lens of parameterised algorithms with the parameters lifetime, underlying treewidth, solution size, and maximum degree in any snapshot~\cite{herrmann2025temporal}. Furthermore, the tractability of a different temporal version of \textsc{Dominating Set} has been studied with respect to VIM width. This version requires domination of vertices in the temporal graph using $k$ intervals of length $l$; Hermann et al.~\cite{herrmann2025timeline} show their problem to be in FPT with respect to $k$, $l$ and VIM width combined.

We can denote an instance $(\mathcal{G},h)$ of \textsc{Temporal Dominating Set} by $x=(\mathcal{G},\beta)$ where $\beta$ is a string encoding $h$. For clarity, we instead denote an instance as $x=(\mathcal{G},h)$.

We use $(X,2)$-states, where the label set $X=\{U, C\}$ contains a label for uncovered and covered vertices respectively, and the counter vector contains an integer $c$ which counts the number of covered vertices, and an integer $d$ which counts the size of the dominating set $D$. We will define our transition routine and initial states such that each state $s_t$ for which there exists a sequence of states $s_0,\ldots,s_t$ where $s_0$ is in the set generates by $\textbf{St}$ and \textbf{Tr} returns \textbf{true} for each pair of consecutive states corresponds to the existence of a set $D$ of size $d$ of vertex appearances which covers $c$ vertices and these are all given label $C$, and all uncovered vertices are given label $U$. Note that both $c$ and $d$ can be at most the number of vertices in the input graph.

We use one initial state, in which every vertex is labelled $U$, and $c$ and $d$ are both equal to $0$. The algorithm which generates this state runs in linear time in the number of vertices. We now give the transition algorithm for this process, and thus prove it is in FPT with respect to VIM width.

 \begin{algorithm}\caption{\textsc{Temporal Dominating Set Transition}}\label{alg:temp-dom}
     \begin{algorithmic}[1]
     \Require A static graph $G$, states $(l_1, (c_1, d_1))$ and $(l_2, (c_2, d_2))$ for $V(\mathcal{G})$, and the integer $h$.
     \Ensure Returns true when $(l_2, (c_2, d_2))$ corresponds to adding $d_2-d_1$ vertex appearances to any set of vertex appearances corresponding to $(l_1, (c_1, d_1))$ and false otherwise.
     \State{Let $U_1$ and $U_2$ be the set of vertices labelled $U$ by $l_1$ and $l_2$ respectively, and equivalently for $C_1$ and $C_2$.}
     \State{Let $I$ be the non-isolated vertices of $G$}
     \If{$\exists D \subseteq I$ such that $d_2 = d_1 + |D|$ and $C_2 = C_1 \cup N_G[D]$ and $c_2 = c_1 + |C_2 \setminus C_1|$}
         \State{\Return{True}}
     \Else
         \State{\Return{False}}
     \EndIf
     \end{algorithmic}
 \end{algorithm}

Given a state $(l, (c, d))$ and an instance of \textsc{Temporal Dominating Set}, that is a temporal graph $\mathcal{G}$ along with an integer $k$, the acceptance routine returns true if and only if $d \leq h$ and $c = |V(\mathcal{G})|$.

We now show that there exists a correspondence between sets of vertex appearances and sequences of states beginning with initial states and related by the transition routine. We say that a sequence $s_0, ..., s_t = (l_t,(c_t,d_t))$ of states \textit{corresponds} to a set $D$ of vertex appearances up to timestep $t$ from $\mathcal{G}$ if and only if:
 \begin{enumerate}
     \item $s_0$ is an initial state,
     \item $\textbf{Tr}(s_{i-1}, s_i, G_i,h)=\textbf{true}$ for every $1 \leq i \leq t$,
     \item the vertices given label $C$ by $s_t$ are exactly those covered by $D$,
     \item $D$ contains $d_t$ vertices,
     \item $D$ covers $c_t$ vertices.
 \end{enumerate}

 \begin{restatable}{lemma}{vimdslem}
     For any timestep $t$, there exists a set $D$ of vertex appearances up to timestep $t$ that covers a set $C'$ if and only if there exists a sequence of states $s_0,\ldots,s_t$ corresponding to $D$ such that all vertices in $C'$ are labelled $C$ by $s_t$.
 \end{restatable}
 \begin{proof}
     We proceed by induction on the timestep $t$. When $t=0$ there are $0$ vertices that appear on or before timestep $0$, and so $0$ vertices are covered. Our initial state gives label $U$ to every vertex in the graph, and sets $d$ and $c$ to $0$, as required.

     Now assume that there exists some set $D$ consisting of vertex appearances on or before timestep $t$. By induction there exists a sequence of states $s_0, ..., s_{t-1}$ corresponding to $D_{t-1} = \{(v, i) \in D : i \leq t - 1\}$. Now let $s_t$ be the state giving all vertices covered by $D$ label $C$, all other vertices label $U$, and with $d = |D|$, and $c$ the number of vertices covered by $D$. Let $D_t = \{v : (v, t) \in D\}$ be the set of vertices appearing on timestep $t$ in $D$, and see that $D_t$ is a subset of the non-isolated vertices of $G_t$. Recall, that we assume that all vertex appearances in $D$ have an incident active edge. Then $d_2 = |D| = |D_{t-1}|+|D_t| = d_1+|D_t|$. The vertices labelled $C$ by $s_t$ are those covered by $D$, so those covered by $D_{t-1}$, which are the vertices labelled $C$ by $s_{t-1}$, and any vertex in the closed temporal neighbourhood of $D_t$ on timestep $t$. Finally, $c_2$ is the number of vertices covered by $D$, that is the number of vertices covered by $D$ up to timestep $t-1$, so $c_1$, plus the number of vertices covered by $D$ on timestep $t$ and not before, so $|C_2 \setminus C_1|$. Therefore $\textbf{Tr}(s_{t-1}, s_t, G_t,h)=\textbf{true}$, and $s_0, ..., s_t$ corresponds to $D$ as required.

     Conversely assume that there exists some sequence of states $s_0, ..., s_t$ such that $\textbf{Tr}(s_{i-1}, s_i, G_i,h)=\textbf{true}$ for every $1 \leq i \leq t$, and $s_0$ is the initial state. By induction there exists a set $D_{t-1}$ of vertex appearances corresponding to the sequence $s_0, ..., s_{t-1}$. Let $A$ be a set of non-isolated vertices in $G_t$ such that $d_2 = d_1 + |A|$, $C_2 = C_1 \cup N_{G_t}[A]$, and $c_2 = c_1 + |C_2 \setminus C_1|$, seeing that such a set exists as \cref{alg:temp-dom} returns true when given $s_{t-1}$, $s_t$, and $G_t$ as input. Now let $D_t$ be the set of vertex appearances $\{(v, t) : v \in A\}$.

     See that $|D_{t-1} \cup D_t| = d_1 + |D_t| = d_2$. Also, the vertices covered by $D_{t-1} \cup D_t$ are those covered by $D_{t-1}$ along with those in the closed temporal neighbourhood of $A$ on timestep $t$, so $C_1 \cup N_{G_t}[A]=C_2$. The number of vertices covered by $D_{t-1} \cup D_t$ is then the number of vertices covered by $D_{t-1}$, so $c_1$, plus the number of vertices covered by $D_t$ but not $D_{t-1}$, so $|C_2 \setminus C_1|$. We have that $c_1 + |C_2 \setminus C_1|$, and thus $D_{t-1} \cup D_t$ corresponds to $s_0, ..., s_t$.
 \end{proof}
 \begin{restatable}{theorem}{vimdsthm}
     \textsc{Temporal Dominating Set} is $(X,2,f_1,f_2)$-locally temporally uniform, where $X=\{U,C\}$, $f_1(G,\beta)=n2^{\omega}$ and $f_2(x)=n$ for any snapshot $G$ of the input temporal graph with $n$ vertices and VIM width $\omega$.
     \end{restatable}
 \begin{proof}
     The initial state gives label $U$ to all vertices as required.

     Consider any graph $G$ and pair of states $s$ and $s'$ such that $\textbf{Tr}(s, s', G,h)=\textbf{true}$. Let $U_s$ and $C_s$ be the vertices labelled $U$ and $C$ by $s$, and equivalently for $U_{s'}$ and $C_{s'}$ and $s'$. Consider any vertex $v$ isolated in $G$. If $v \in C_s \subseteq C_s \cup N_G[D]$ for any set $D$ then $v \in C_{s'}$. If $v \in C_{s'}$ then $v \in C_s$, as $C_{s'} = C_s \cup N_G[D]$, for some set $D$ of non-isolated vertices of $G$. Now as $U_s$ and $C_s$ partition the vertices of $G$, as do $U_{s'}$ and $C_{s'}$, we have that $v \in U_{s'}$ if and only if $v \in U_s$. Therefore $s$ and $s'$ give the same label to any isolated vertex in $G$.

     Consider any graph $G$ and a quadruple of states $r$, $r'$, $s$, and $s'$ for $V(G)$, such that $r$ and $s$ agree on the non-isolated vertices in $G$, $r'$ and $s'$ agree on the non-isolated vertices in $G$, and the pairs $s,s'$ and $r,r'$ both give the same label to every isolated vertex not in $G$. Assume without loss of generality that $\textbf{Tr}(r, r', G,h)=\textbf{true}$.

     Let $C_{s'}$, $C_s$, $C_{r'}$, and $C_r$ be the vertices given label $C$ by $s'$, $s$, $r'$, and $r$ respectively. Also let $I$ be the non-isolated vertices of $G$, and $A$ a set of vertices such that $d_{r'}=d_r + |A|$, $C_{r'} = C_r \cup N_G[A]$, and $c_{r'} = c_r + |C_{r'} \setminus C_r|$, noting that such a set must exist as $\mathbf{Tr}(r, r', G,h)=\textbf{true}$.

     Consider any vertex $v \in C_{s'}$. If $v \in I$ then $v \in C_{r'} = C_r \cup N_G[A]$ as $s'$ and $r'$ agree on $I$. Then $v \in C_s \cup N_G[A]$ as $s$ and $r$ agree on $I$. Otherwise if $v \notin I$ then $v \in C_{s}$ as $s$ and $s'$ give the same label to every vertex not in $I$. Therefore $C_{s'} \subseteq C_{s} \cup N_G[A]$. Consider now any vertex $v \in C_{s} \cup N_G[A]$, if $v \in I$ then $v \in C_{r} \cup N_G[A] = C_{r'}$ and then $v \in C_{s'}$. Otherwise if $v \notin I$ then $v \in C_{s}$, as $N_G[A] \subseteq I$. Then $v \in C_{s'}$ as $s$ and $s'$ give the same label to any vertex not in $I$. Therefore $C_{s} \cup N_G[A] \subseteq C_{s'}$ and $C_{s'}=C_{s} \cup N_G[A]$. Also, $d_{s'} = d_{r'} = d_r + |A| = d_s + |A|$. 
    
     Next, see that $C_{s'}\setminus C_s = (C_{s'} \cap I)\setminus (C_s \cap I)$, as $s'$ and $s$ give the same label to every vertex not in $I$. Equivalently see that $C_{r'} \setminus C_r = (C_{r'} \cap I)\setminus (C_r \cap I)$, and therefore $C_{s'} \setminus C_s = C_{r'} \setminus C_r$ as $s'$ and $r'$ agree on $I$, as do $s$ and $r$. Therefore $c_{s'} = c_{r'} = c_{r} + |C_{r'}\setminus C_r|=c_{s} + |C_{s'} \setminus C_s|$, and we have that $\mathbf{Tr}(s, s', G,h)=\textbf{true}$ as required. The transition routine functions by looking for existence of a set $D$ of non-isolated vertices that dominate the vertices newly labelled $C$ and satisfy the counters. Finding $D$ requires $O(2^{\omega})$ time since $\omega$ bounds the number of non-isolated vertices at any time. Checking that $D$ has the required properties takes $O(n)$ time, as we must check the labels of all of the vertices. Therefore, the transition routine requires $O(2^{\omega}n)$ time.

     Finally, see that the acceptance routine checks only the value of a counter variable, and therefore if it returns true when given a state $s$, then it will return true when given any state agreeing with $s$ on any vertex set. This requires a constant-time check of the counters, which is upper bounded by the time required to generate the starting states ($O(n)$).
 \end{proof}
 Then, as we use $2$ counter variable of size at most $n\Lambda$ and 2 labels, we finally obtain the following theorem from \Cref{thm:loctempunf}.

\vimdom*

\end{toappendix}
\section{Applications of TIM meta-algorithm}\label{sec:applications-tim}

We illustrate application of our meta-algorithms by using Theorem~\ref{thm:component-temp-unif} to obtain the following results. For ease of reading, we replicate the discussions of the problems to which we have already applied our VIM meta-algorithm in the relevant sections.
We begin with two problems to which we applied the VIM meta-algorithm.
\begin{restatable}{theorem}{TIMhampath}\label{thm:TIM-hampath}
    Given a TIM decomposition with width $\phi$,
    \textsc{Temporal Hamiltonian Path} can be solved in time $O(n^5\Lambda^5\phi^{10}3^{24\phi^3})$, where the input temporal graph has $n$ vertices and lifetime $\Lambda$.
\end{restatable}
\begin{restatable}{theorem}{TIMdomset}\label{thm:TIM-domset}
    Given a TIM decomposition with width $\phi$,
    \textsc{Temporal Dominating Set} can be solved in time $O(n^5\Lambda^53^{12\phi^3}\phi^{12\phi^3+10})$, where $\Lambda$ is the lifetime of the input temporal graph and $n$ the number of vertices.
\end{restatable}
The third problem we consider is a temporal analogue of matching~\cite{mertzios_computing_2023}. In this problem we look for a set $M$ of at least $h$ time-edges such that each pair in $M$ either occur at times at least $\Delta$ apart, or do not share an endpoint. It is simple to show that this problem is in FPT with respect to the maximum number of edges active at any given time, which bounds VIM width. The final problem asks if there is a set of time-edges $\mathcal{E}'$ with cardinality $h$ such that a given source has temporal reachability at most $r$ following the deletion of $\mathcal{E}'$ from $\mathcal{G}$~\cite{enright_deleting_2021}. A generalisation of this is already known to be in FPT with respect to VIM width~\cite{bumpus_edge_2023}.
\begin{restatable}{theorem}{TIMmatching}\label{thm:TIM-matching}
    Given a TIM decomposition with width $\phi$,
    \textsc{$\Delta$-Temporal Matching} can be solved in time $O(n^5\Lambda^5\phi^{12\phi^3+11.5}(4\Delta^2)^{12\phi^3})$, where the input temporal graph has $n$ vertices and lifetime $\Lambda$.
\end{restatable}

\begin{restatable}{theorem}{TIMdeletion}\label{thm:TIM-deletion}
    Given a TIM decomposition with width $\phi$,
    \textsc{SingReachDelete} can be solved in time $O(n^9\Lambda^93^{60\phi^3}\phi^{48\phi^3+10})$, where the input temporal graph has $n$ vertices and lifetime $\Lambda$.
\end{restatable}

% Although we show \textsc{Temporal Firefighter Reserve} to be hard on graphs with bounded TIM width, we use our meta-algorithm to show inclusion in FPT with respect to the combined parameter of TIM width and lifetime.
% \begin{restatable}{theorem}{TIMfire}\label{thm:TIM-fire}
%     \textsc{Temporal Firefighter Reserve} can be solved in time $n^{O(1)} (\phi+\Lambda)^{O(\phi^2\Lambda)}$, where the input temporal graph has $n$ vertices, lifetime $\Lambda$ and TIM width $\phi$.
% \end{restatable}
\begin{toappendix}
\subsection{Temporal Hamiltonian Path}\label{sec:ham-path-tim}
Previously, we showed this problem to be in FPT with respect to VIM width by showing it to be $(X,1,f_1,f_2)$-locally temporally uniform where $f_1$ and $f_2$ are both linear in $n$. We extend this result by showing that \textsc{Temporal Hamiltonian Path} is in FPT with respect to TIM width by applying our TIM meta-algorithm.

Recall the formal definition of the problem.

\begin{nolinenumbers}
\decisionproblem{Temporal Hamiltonian Path}{A temporal graph $\mathcal{G}$.}{Does there exist a strict temporal path containing every vertex in $\mathcal{G}$?}
\end{nolinenumbers}
Note that the only input for this problem is the temporal graph itself. In our definitions of the meta-algorithms, we use a string $\beta$ to encode the input of a given problem which is not the temporal graph. In this case $\beta$ would be the empty string. For ease, we omit $\beta$ in this section.

%Not ideal to use V as this is normally the whole vertex set
We use $(X,1)$-component states, where the label set $X=\{\visited,\unvisited,\current\}$ contains a label for visited, unvisited, and current vertices respectively, and the vectors in the state contain one integer $p$, which counts the number of current locations in each snapshot. We will define our starting, finishing and transition routines and vector upper bound so that each sequence of $(X,1)$-component states where the routines return true for all relevant connected components corresponds to the existence of a temporal path such that, if there is a vertex given label $\current$, the path is at that vertex at that time, and any vertices labelled $\visited$ are traversed by the path by that time. Our upper bound on the sum of the values of the vectors, $\textbf{v}_{\text{upper}}=(\Lambda)$. Enforcing that $\textbf{v}_{\text{upper}}\leq \Lambda$ gives us that there is at most one ``current location'' at any time, if we construct the states such that there is at least one current location in each snapshot. We now show that \textsc{Temporal Hamiltonian Path} is component-exchangeable temporally local, by giving the starting, finishing, validity and transition routines.

Given a connected component $C_1$ of $G_1$, let the starting routine be an algorithm which takes in $C_1$, a labelling $l$ of $V(C_1)$, and a vector $\textbf{v}=(p)$ and returns true if and only if $p=|l^{-1}(\current)|\in \{0,1\}$, and $l^{-1}(\unvisited)=V(C_1)\setminus l^{-1}(\current)$.

We define the finishing routine in a similar way. Let the finishing routine be an algorithm which takes in a connected component $C_{\Lambda}$ of $G_{\Lambda}$, a labelling $l$ of $V(C_{\Lambda})$, and a vector $\textbf{v}=(p)$ and returns true if and only if $p=|l^{-1}(\current)|\in \{0,1\}$, and $|l^{-1}(\unvisited)|=0$.

For a connected component $C_t$ of a snapshot $G_t$ of $\mathcal{G}$, we define our validity routine as follows. The validity routine is an algorithm which takes in $C_{t}$, a labelling $l$ of $V(C_t)$, and a vector $\textbf{v}=(p)$ and returns true if and only if $p=|l^{-1}(\current)|\in \{0,1\}$.

The transition routine is given in Algorithm~\ref{alg:temp-ham-timw}.
%C_1 \setminus C_2 etc
% \resetlinenumber[1269]
We now show that there exists a correspondence between the temporal paths on $\mathcal{G}$ and sequences of states such that the starting routine returns true for all connected components of the first snapshot, and the validity and transition routines return true for all connected components of all snapshots. We say that a sequence $s_0, ..., s_t$ of $(X,1)$-component states of the form $s_t=(l_t,\textbf{w}^t_1,\ldots,\textbf{w}^t_c, \nu_t)$ such that $t\leq\Lambda$ \textit{corresponds} to a temporal path $P$ if and only if:
\begin{enumerate}
    \item for all connected components $C_1$ in $G_1$, $\textbf{St}(s_0|_{C_1}, C_1,\mathcal{G})=\textbf{true}$,
    \item for all times $1\leq i\leq t$ and all connected components $C_{i}$ in $G_{i}$, $\textbf{Val}(s_{i}|_{C_{i}}, C_{i},\mathcal{G})=\textbf{true}$,
    \item for all times $1\leq i\leq t$ and all connected components $C_{i}$ in $G_{i}$, $\textbf{Tr}(l_{i-1}|_{C_i}, l_i|_{C_i}, C,\mathcal{G})=\textbf{true}$,
    \item the labelling $l_t$ gives a single vertex the label $\current$, this is the final vertex on $P$, and $P$ has arrival time $t' \leq t$, and
    \item the vertices traversed by $P$ are exactly the vertices given label $\visited$ by $s_t$. 
\end{enumerate}
Recall that $\nu_i$ is the function which maps all connected components of $G_i$ to a vector in the $(X,1)$-component state $s_i$.

\begin{restatable}{lemma}{temphamtimlem}\label{lem:temp-ham-tim}
    For any timestep $t$, there exists a temporal path $P$ arriving at time $t' \leq t$ if and only if there exists a sequence of $(X,1)$-component states $s_0, ... ,s_t$ corresponding to $P$ such that for each $0\leq i\leq t$, $\sum_{C\text{ connected component of }G_i}\nu_i(C)= 1$ and the vertex that $P$ is on at time $i$ is labelled $\current$ by $s_i$ for all $1\leq i\leq t$.
    \end{restatable}
    \begin{proof}
    We proceed by induction on the timestep $t$. Any temporal path that has arrival time $0$ can only contain a single vertex. Therefore, for a path $P$ consisting of the vertex $v$, let the state be $s_0$ such that $v$ is labelled $\current$ and all remaining vertices are labelled $\unvisited$, and $\nu_0(C)=(0)$ for all connected components apart from the connected component $C'$ containing $v$ for which $\nu_0(C')=(1)$. It is clear that the starting routine returns true for all restrictions of $s_0$ to connected components of $G_1$. Now suppose there is a state $s_0$ such that $\textbf{St}$ returns true for every restriction to a connected component in $G_1$ and the sum of vectors in $s_0$ is 1. Then, there must be one vertex $v$ labelled current by $s_0$. This vertex must be the path which corresponds to $s_0$. 

    For induction, we now assume that, for any timestep $t^*$, there exists a temporal path arriving at time $t' \leq t^*$ if and only if there exists a corresponding sequence of states $s_0, ... s_{t^*}$ such that the sum of vectors at each time is 1.
    
    Now assume that there is a temporal path $P$ arriving on a timestep $t' \leq t$ at a vertex $v$. If $t' < t$, then $P$ must arrive at $v$ by $t-1$. By induction, there exists a sequence of states $s_0, ..., s_{t-1}$ corresponding to $P$ such that the sum of vectors in each state is 1. Now take $l_t = l_{t-1}$. Note that $\textbf{Tr}(l_{t-1}|_C, l_t|_C, C,\mathcal{G})=\textbf{true}$ for all connected components $C$ of $G_t$, since line 6 of \Cref{alg:temp-ham-timw} will return true. For each connected component $C$ of $G_{t+1}$, let $\nu_{t+1}(C)=|l^{-1}|_C(\current)|$. By construction, the validity routine must return true for all connected components in $G_{t+1}$. Also, since the number of vertices labelled $\current$ by each state is the same, the sum of vectors in $s_t$ must be the same as the sum of vectors in $s_{t-1}$, which is 1. 
    
    If $t' = t$, there must exist some temporal path $P'$ arriving at a time $t' \leq t-1$ at a vertex $u$ adjacent to $v$ on timestep $t$. Hence, by induction, there exists a sequence of states $s_0, ..., s_{t-1}$ corresponding to $P'$. Therefore, the first condition (that the starting routine returns true for all connected components of the first snapshot) holds. Let $s_t$ be a state with labelling $l_t$ that gives label $\current$ to $v$, label $\visited$ to $u$, and labels all other vertices as they are in $s_{t-1}$. Let the vector $\nu_t(C)$ be 1 if $v\in C$ and 0 otherwise. Then, by construction, the validity routine will return true for all connected components in $G_t$. It is clear that the sum of the vectors in $s_t$ is at exactly 1, and that the vertices labelled $\visited$ by $l_t$ are traversed by $P$. We now show that condition 3 holds by the definition of correspondence. Note that $u$ and $v$ must be in the same connected component of $G_t$. Call this component $C'$. We note that, for all other connected components $C$ of $G_t$, the labellings $l_t|_C$ and $l_{t-1}|_C$ are the same. Therefore, $\textbf{Tr}(l_{t-1}|_C, l_t|_C, C,\mathcal{G})=\textbf{true}$ and the third condition holds in this case. See that $\{u,v\} \in E(G_t[C'])$, as $P$ traverses $\{u,v\}$ on timestep $t$. Therefore, $\textbf{Tr}(l_{t-1}|_{C'}, l_t|_{C'}, C',\mathcal{G})=\textbf{true}$, as line 4 of \Cref{alg:temp-ham-timw} will return true.
    
    Now assume there is a sequence of states $s_0, \ldots, s_t$ such that $\textbf{St}(s_0|_C,C,\mathcal{G})=\textbf{true}$ for all connected components $C$ of $G_1$, $\textbf{Val}(s_t|_{C_i},C_i,\mathcal{G})$ and $\textbf{Tr}(l_{i-1}|_{C_i}, l_i|_{C_i}, C_i,\mathcal{G})=\textbf{true}$ for every $1 \leq i \leq t$ and connected component $C_i$ of $G_i$, and the sum of the vectors in each state is 1. If such a vertex exists, let $v$ be the vertex in $C$ given label $\current$ by $s_t$. We know there must be at most one such vertex since the validity routine returns true if and only if the number of vertices assigned $\current$ in each connected component $C$ by $s_t$ is equal to $\nu_t(C)$ and their sum over all connected components is 1. Consider the sequence $s_0, ..., s_{t-1}$. By induction, there must exist a temporal path $P$ corresponding to this sequence. Let $u$ be the final vertex on this path. Note that $u$ has label $\current$ under $l_{t-1}$. Recall that $\textbf{Tr}(l_{t-1}|_C, l_t|_C, C,\mathcal{G})=\textbf{true}$ for all connected components $C$ in $G_t$. We note that this implies that the set of vertices labelled $\current$ by $l_t$ contains exactly at most one vertex not labelled $\current$ by $l_{t-1}$ and this must be $u$, if it exists. If this is the case, the routine must return true because line 4 of \Cref{alg:temp-ham-timw} returns true then $\{u,v\} \in C$, and therefore $t \in \lambda(\{u,v\})$. Since the validity routine must also return true for all connected components of $G_t$, $\nu_t(C)$ must give the number of vertices labelled $\current$ by $l_t$. By our assumption that the sum of all vectors is 1, there can only be one vertex labelled $\current$ by any of the states in the sequence. We claim that there is a temporal path $P'$ that traverses each vertex labelled current. By the inductive hypothesis, $P$ traverses all vertices labelled current in snapshots with times at most $t-1$. The path $P'$ must traverse the same edges as $P$, before traversing $\{u,v\}$ on timestep $t$. Then $P'$ will traverse all the vertices traversed by $P$, along with the additional vertex $u$, and therefore corresponds to $s_t$. Otherwise, if line 6 of \Cref{alg:temp-ham-timw} returns true, then $s_t = s_{t-1}$, and therefore $P$ corresponds to $s_0, ..., s_t$.
\end{proof}

\temphamtimthm*

\begin{proof}
    We begin by showing an instance $\mathcal{G}$ of \textsc{Temporal Hamiltonian Path} is a yes-instance if and only if the criteria of Definition~\ref{def:component-temp-unif} hold. That is, for each connected component $C_1$ of $G_1$, $\textbf{St}(s_0|_{C_1},C_1,\mathcal{G})=\textbf{true}$; for each connected component $C_{\Lambda}$ of $G_{\Lambda}$, $\textbf{Fin}(s_{\Lambda}|_{C_{\Lambda}},C_{\Lambda},\mathcal{G})=\textbf{true}$; $\textbf{Tr}(l_{t-1}|_{C_t}, l_t|_{C_t}, C_t,\mathcal{G})=\textbf{true}$ where $l_t$ is the labelling of vertices of state $s_t$, for all times $1 \leq t \leq \Lambda$ and connected components $C_t$ of $G_t$; $\textbf{Val}(s_t|_{C_t}, C_t,\mathcal{G})=\textbf{true}$ for all times $1 < t < \Lambda$ and connected components $C_t$ of $G_t$; and the sum of vectors satisfies $\sum_{0\leq t\leq \Lambda} \sum_{C \in\mathcal{C}_t}\nu_{s_t} (C)\leq \textbf{v}_{\text{upper}}$.

    By Lemma~\ref{lem:temp-ham-tim}, there is a path $P$ arriving at time $t\leq \Lambda$ if and only if there is a sequence of $(X,1)$-component states $s_0,\ldots,s_{\Lambda}$ corresponding to $P$ such that the sum of vectors in each state is 1, the final vertex in the path is labelled $\current$ by $l_t$, and all other vertices traversed by $P$ are labelled $\visited$. 
    
    We now show that, if the criteria on the sequence of $(X,1)$-component states hold, there must be exactly one vertex labelled $\current$ by each $(X,1)$-component state in the sequence. Recall that the finishing routine requires all vertices to either be labelled $\current$ or $\visited$. If the transition routine returns true for all connected components of all snapshots $G_i$, $1<i\leq \Lambda$, vertices cannot receive these labels $\visited$ or $\current$ unless there is at least one vertex labelled $\current$ in the each state in the sequence $s_0,\ldots,s_{\Lambda-1}$. Therefore, there must be at least one vertex labelled current by $s_i$ for each time $1\leq i\leq \Lambda$. To see this note that, if a vertex is labelled $\visited$, by construction of the transition routine there must exist a vertex labelled $\current$ at some time $i$ such that the two are adjacent at time $i$. Further note that if there exists a vertex labelled $\current$ in any of the states, there must be at least one vertex labelled $\current$ by each of the states in the sequence. This is a consequence of the fact that the transition routine returns true if the labellings are the same, or the set of vertices labelled $\current$ by the earlier of the labellings and not the other is of cardinality one and vice versa. Pairing this with the requirement that the sum of all vectors $\textbf{v}_{upper}$ must be at most $\Lambda$ implies that there is exactly one vector labelled $\current$ in each state $s_i$ in the sequence. Hence, if the criteria hold, there exists a path $P$ characterised by the vertices labelled $\current$. Furthermore, if the finishing routine returns true for all connected component, the path $P$ must visit all vertices. Therefore, $P$ is a temporal Hamiltonian path.  

    If there exists a temporal Hamiltonian path $P$ of $\mathcal{G}$, then by Lemma~\ref{lem:temp-ham-tim}, there is a sequence of $(X,1)$-component states corresponding to $P$ such that the sum of vectors at each snapshot is $1$ and the vertices traversed by $P$ are labelled $\visited$. Since $P$ is a temporal Hamiltonian path, all vertices must be labelled $\visited$ by $s_{\Lambda}$ except one which is labelled $\current$. Therefore, there exists a sequence of $(X,1)$-component states $s_0,\ldots,s_{\Lambda}$ such that, for each connected component $C_1$ of $G_1$, $\textbf{St}(s_0|_{C_1},C_1,\mathcal{G})=\textbf{true}$; for each connected component $C_{\Lambda}$ of $G_{\Lambda}$, $\textbf{Fin}(s_{\Lambda}|_{C_{\Lambda}},C_{\Lambda},\mathcal{G})=\textbf{true}$; $\textbf{Tr}(l_{t-1}|_{C_t}, l_t|_{C_t}, C_t,\mathcal{G})=\textbf{true}$ where $l_t$ is the labelling of vertices of state $s_t$, for all times $1 \leq t \leq \Lambda$ and connected components $C_t$ of $G_t$; $\textbf{Val}(s_t|_{C_t}, C_t,\mathcal{G})=\textbf{true}$ for all times $1 < t < \Lambda$ and connected components $C_t$ of $G_t$; and the sum of vectors satisfies $\sum_{0\leq t\leq \Lambda} \sum_{C \in\mathcal{C}_t}\nu_{s_t} (C)\leq \textbf{v}_{\text{upper}}=\Lambda$. 
    
    Recall that all of the subroutines run in time at most linear in the size of the connected component, which is bounded by $\phi$. Therefore, \textsc{Temporal Hamiltonian Path} is $(X,1,f)$-component-exchangeable temporally uniform, where $f$ is linear in $\phi$. 
\end{proof}
Then, as \Cref{alg:temp-ham-timw} runs in time $O(|C|)\leq O(\phi)$ by checking the label on each vertex in the connected component $C$ in turn, we use vectors with one entry of magnitude at most $1$ and 3 labels, we finally obtain our main result by applying \Cref{thm:component-temp-unif}.

\TIMhampath*

% \section{Proofs of results in Section~\ref{sec:applications}}
% For ease of reading, we split this section of the appendix into two. We begin with applying the VIM algorithm to \textsc{Temporal Hamiltonian Path}.

\subsection{Temporal Dominating Set}
We now consider a temporal analogue of \textsc{Dominating Set} given by Casteigts and Flocchini~\cite{casteigts_deterministic_2013}. This problem asks if it is possible to find a set $D$ of size $h$ or less consisting of vertex-time pairs (vertex appearances) such that every vertex $v$ is \textit{covered}, that is it either appears in a pair in $D$, or there exists a $(u, t) \in D$ such that $v$ is adjacent to $u$ on timestep $t$. We assume that the underlying graph $\mathcal{G}{\downarrow}$ contains no isolated vertices, and that for any vertex appearance $(u, t) \in D$, $u$ has an incident edge active on timestep $t$. If the graph contains an isolated vertex, then it must be included in any dominating set; we can therefore find an equivalent instance of \textsc{Temporal Dominating Set} by removing the isolated vertices and decrementing the number of appearances allowed in the dominating set. Adding a vertex appearance when it has no incident edges, means that that appearance can only dominate the vertex itself. Therefore, we can assume that all vertex appearances in a solution to \textsc{Temporal Dominating Set} have an active incident edge, as swapping a vertex appearance without an active incident edge for a time where there is an incident edge gives a solution that is no worse than the original.

\begin{nolinenumbers}
\decisionproblem{Temporal Dominating Set}{A temporal graph $\mathcal{G}$ and an integer $h$.}{Does there exist a set $D$ of vertex appearances such that $|D|\leq h$ and the appearances in $D$ cover every vertex in $\mathcal{G}$?} 
\end{nolinenumbers}

We can denote an instance $(\mathcal{G},h)$ of \textsc{Temporal Dominating Set} by $x=(\mathcal{G},\beta)$ where $\beta$ is a string encoding $h$. For clarity, we instead denote an instance as $x=(\mathcal{G},h)$.

For this application of the TIM width meta-algorithm, we use $(X,1)$-component states, where the label set $X=\{\covered, \uncovered,\dominating\}$, and the vectors of the states contain an integer $d$ which counts the number of vertices in the dominating set $D$ within the relevant timed component. 
 We will define our transition routine, validity routine, and starting routine such that each state for which their restrictions to the connected components of each snapshot return true for all relevant routines corresponds to the existence of a set $D$ of size $d$ of vertex appearances labelled $\dominating$ which covers vertices which are given label $\covered$, and all uncovered vertices are given label $\uncovered$. Our upper bound on the sum of the vectors of the states is $\textbf{v}_{\text{upper}}=h$. Since this is the sum of all vertex appearances labelled with $\dominating$, this ensures that there are at most $h$ elements in a potential dominating set.

 Our starting routine returns true if and only if every vertex is labelled $\uncovered$ and each vector is a zero vector in the state. The finishing routine returns true if and only if, for every connected component $C$ of $G_{\Lambda}$, every vertex in $C$ is either labelled $\covered$ or $\dominating$, and the vector $(d)$ is the number of vertices in $C$ labelled $\dominating$ by the state. Finally, the validity routine is defined as an algorithm which returns true for a connected component $C$ if and only if the number of vertices in $C$ labelled with $\dominating$ is equal to the vector $(d)$. Algorithm~\ref{alg:temp-dom-tim} gives our transition routine.

 \begin{algorithm}\caption{\textsc{Temporal Dominating Set Component Exchangeable Transition Routine}}\label{alg:temp-dom-tim}
     \begin{algorithmic}[1]
     \Require A connected component $C$ of a snapshot $G_t$, labellings $l_1$ and $l_2$ for $V(C)$, and input instance $x$.
     \Ensure Returns true if there is a set $D$ of vertex appearances in $\mathcal{G}$ that dominates the vertices labelled $\covered$ by $l_2$ if and only if the set $D$ without the pairs in $C$ at time $t$ dominates the vertices labelled $\covered$ by $l_1$ and false otherwise.
     \State{Let $\text{Uncovered}_1$ and $\text{Uncovered}_2$ be the set of vertices labelled $\uncovered$ by $l_1$ and $l_2$ respectively, and equivalently for $\text{Covered}_1$ and $\text{Covered}_2$, and $\text{Dominating}_1$ and $\text{Dominating}_2$.}
     \If{$\text{Covered}_2 = \text{Covered}_1 \cup N_C[\text{Dominating}_2]\cup (\text{Dominating}_1\setminus\text{Dominating}_2)$}
         \State{\Return{True}}
     \Else
         \State{\Return{False}}
     \EndIf
     \end{algorithmic}
 \end{algorithm}

 We now show that there exists a correspondence between a \emph{partial temporal dominating set} and sequences of states such that the starting routine returns true for all connected components of the first snapshot, and the validity and transition routines return true for all connected components of all snapshots. We say that a sequence $s_0, ..., s_t$ of $(X,1)$-component states of the form $s_t=(l_t,\textbf{w}^t_1,\ldots,\textbf{w}^t_c, \nu_t)$ \textit{corresponds} to a partial temporal dominating $D$ of vertex appearances up to timestep $t$ from $\mathcal{G}$ if and only if:
 \begin{enumerate}
     \item for all connected components $C_1$ in $G_1$, $\textbf{St}(s_0|_{C_1}, C_1,x)=\textbf{true}$,
     \item for all times $1\leq i\leq t$ and all connected components $C_{i}$ in $G_{i}$, $\textbf{Val}(s_{i}|_{C_{i}}, C_{i},x)=\textbf{true}$,
     \item for all times $1\leq i\leq t$ and all connected components $C_{i}$ in $G_{i}$, $\textbf{Tr}(l_{i-1}|_C, l_i|_C, C,x)=\textbf{true}$,
     \item for each vertex-time pair $(v,t)$ in $D$, $v$ is labelled $\dominating$ by $l_t$,
     \item the vertices given label $\covered$ or $\dominating$ by $s_t$ are exactly those covered by $D$,
     \item $D$ contains $|\nu_i(C)|$ vertices in each connected component $C$ of a snapshot $G_i$ of $\mathcal{G}$.
 \end{enumerate}
 \begin{restatable}{lemma}{timdslem}\label{lem:temp-dom-tim}
  For any timestep $t$, a set of vertex appearances $D$ up to timestep $t$ from $\mathcal{G}$ is a partial dominating set covering a set $C$ of vertices, if and only if there exists a corresponding sequence of $(X,1)$-component states $s_0, \ldots, s_t$ such that all vertices in $C$ which do not appear in a pair in $D$ are labelled $\covered$.
 \end{restatable}
 \begin{proof}
     We proceed by induction on the timestep $t$. For time $t=0$, we can assume without loss of generality that any partial temporal dominating set is empty since the vertex appearances cannot dominate any other vertices. Begin by considering the partial dominating set $D=\emptyset$. Let the state $s_0$ be such that all vertices are labelled $\uncovered$, and $\nu_0(C)=(0)$ for all connected components $C$ of $G_1$. It is clear that the starting routine returns true for all restrictions of $s_0$ to connected components of $G_1$. Now suppose there is a state $s_0$ such that $\textbf{St}$ returns true for every restriction to a connected component in $G_1$. Then, all vertices must be labelled $\uncovered$. The empty set temporally dominates itself. Therefore, $D=\emptyset$ corresponds to $s_0$, and the base case holds. 

     We now assume for all times $t'\leq t$, there is a partial temporal dominating set of vertex appearances $D$ up to timestep $t'$ if and only if there is a corresponding sequence of $(X,1)$-component states $s_0, ... s_{t'}$.

     Consider time $t+1$. We first assume that there exists partial temporal dominating set $D$ of vertex appearances on or before $t+1$. Let $D_{t}$ be the set of vertex appearances in $D$ with times up to $t$. That is, $D_{t} = \{(v, i) \in D : i \leq t \}$. By our inductive hypothesis, there exists a corresponding sequence of states $s_0,\ldots,s_t$ for $D_t$.
    
     Now let $s_{t+1}$ be the state labelling all vertices in pairs in $D$ with the label $\dominating$, all remaining vertices covered by $D$ with the label $\covered$, and all other vertices with the label $\uncovered$. For each connected component $C$ in $G_{t+1}$, let $\nu_{t+1}(C)=(|l_{t+1}^{-1}(\dominating)\cap C|)$. It is clear by construction that $(|D\cap C|)=\nu_{t+1}(C)$.
%     % Let $D_{t+1} = \{v : (v, t+1) \in D\}$ be the set of vertices appearing on timestep $t+1$ in $D$.
     Then it is clear that $\textbf{Val}(s_{t+1}|_C,C,x)=\textbf{true}$ for all connected components $C$ of $G_{t+1}$. For a vertex to have label $\covered$, it must either be covered by a vertex appearance before $t+1$, or be adjacent to a vertex $v$ such that $(v,t+1)\in D$. Therefore, for all connected components $C$ in $G_{t+1}$, the vertices labelled with $\covered$ are the union of those labelled covered or dominating by $s_t$ and those in the neighbourhood of those labelled $\dominating$ by $s_{t+1}$. Therefore, Algorithm~\ref{alg:temp-dom-tim} returns true in line 3 for all connected components of $G_{t+1}$. Hence $s_0,\ldots,s_{t+1}$ corresponds to $D$ as required.
    
     Assume that there exists some sequence of states $s_0, ..., s_{t+1}$ such that, for all connected components $C_1$ in $G_1$, $\textbf{St}(s_0|_{C_1}, C_1)=\textbf{true}$; for all times $1\leq i\leq t$ and all connected components $C_{i}$ in $G_{i}$, $\textbf{Val}(s_{i}|_{C_{i}}, C_{i})=\textbf{true}$; and for all times $1\leq i\leq t$ and all connected components $C_{i}$ in $G_{i}$, $\textbf{Tr}(l_{i-1}|_C, l_i|_C, C,x)=\textbf{true}$. By induction, there exists a partial temporal dominating set $D_{t}$ corresponding to the sequence $s_0, ..., s_{t}$. Let $D_{t+1}$ be the set of vertices labelled $\dominating$ by $s_{t+1}$, and $D=D_{t+1}\cup D_t$. Note that, since the transition routine returns true for all connected components in $G_{t+1}$, the vertices covered by $D$ are precisely those which are labelled $\covered$ or $\dominating$ by $s_{t+1}$. Since the validity routine returns true for all connected components $C$ of $G_{t+1}$, the number of vertices in $C$ labelled $\dominating$ by $s_{t+1}$ is exactly the cardinality of the set $\{(v,t+1)\,:\, v\in C\}$ for all connected components $C$ on $G_{t+1}$. In other words, $D\cap C=\nu_{t+1}(C)$. Thus, $D$ corresponds to $s_0, ..., s_{t+1}$, and, for any timestep $t$, there exists a set of vertex appearances $D$ up to timestep $t$ from $\mathcal{G}$, if and only if there exists a corresponding sequence of $(X,1)$-component states $s_0, \ldots, s_t$.
 \end{proof}
 \begin{restatable}{theorem}{timdsthm}
     \textsc{Temporal Dominating Set} is $(X,1,f)$-component-exchangeable temporally uniform, where $X=\{\dominating,\covered,\uncovered\}$, and $f(|C|,x)=\phi$ for every timed connected component $C$ of an input temporal graph with TIM width $\phi$.
 \end{restatable}
 \begin{proof}
     We begin by showing an instance $(\mathcal{G},h)$ of \textsc{Temporal Dominating Set} is a yes-instance if and only if the criteria of Definition~\ref{def:component-temp-unif} hold. That is, for each connected component $C_1$ of $G_1$, $\textbf{St}(s_0|_{C_1},C_1,x)=\textbf{true}$; for each connected component $C_{\Lambda}$ of $G_{\Lambda}$, $\textbf{Fin}(s_{\Lambda}|_{C_{\Lambda}},C_{\Lambda},x)=\textbf{true}$; $\textbf{Tr}(l_{t-1}|_{C_t}, l_t|_{C_t}, C_t,x)=\textbf{true}$ where $l_t$ is the labelling of vertices of state $s_t$, for all times $1 \leq t \leq \Lambda$ and connected components $C_t$ of $G_t$; $\textbf{Val}(s_t|_{C_t}, C_t,x)=\textbf{true}$ for all times $1 < t < \Lambda$ and connected components $C_t$ of $G_t$; and the sum of vectors satisfies $\sum_{0\leq t\leq \Lambda} \sum_{C \in\mathcal{C}_t}\nu_{s_t} (C)\leq \textbf{v}_{\text{upper}}$.

     By Lemma~\ref{lem:temp-dom-tim}, there exists a partial dominating set $D$ up to timestep $\Lambda$ from $\mathcal{G}$, if and only if there exists a corresponding sequence of $(X,1)$-component states $s_0, ... s_{\Lambda}$. Therefore, $(\mathcal{G},h)$ is a yes-instance if and only if there is a sequence of $(X,1)$-component states $s_0,\ldots,s_{\Lambda}$ such that the sum of vectors in each state is at most $h=\textbf{v}_{\text{upper}}$, and all vertices apart from $l_{\Lambda}^{-1}(\dominating)$ are labelled $\covered$ by $l_{\Lambda}$. This holds if and only if $\textbf{Fin}(s_{\Lambda}|_C,C,x)=\textbf{true}$ for all connected components $C$ in $G_{\Lambda}$. Note that all subroutines run in time at most $O(\phi)$, where $\phi$ is the TIM width of $\mathcal{G}$. Therefore, \textsc{Temporal Dominating Set} is $(X,1,f)$-component-exchangeable temporally uniform, where $f$ is linear in $\phi$. 
 \end{proof}
 Then, as we use a vector with one entry with magnitude at most $\phi$ and 3 labels, we finally obtain the following corollary from \Cref{thm:component-temp-unif}. 

\TIMdomset*

\subsection{\texorpdfstring{$\Delta$}{Delta}-Temporal Matching}\label{sec:matching}

In this section, we consider a temporal analogue to the maximum matching problem and show it to be in FPT with respect to TIM width by leveraging our meta-algorithm. In this framework, we require that, for all pairs of time-edges in our matching, the endpoints of the pairs of time-edges are either disjoint, or have non-empty intersection and are sufficiently far apart in time. This problem is introduced by Mertzios et al.~\cite{mertzios_computing_2023}. Their definition of a $\Delta$-temporal matching is as follows. 

\begin{defn}[Definition 2,~\cite{mertzios_computing_2023}]
    A \emph{$\Delta$-temporal matching} of a temporal graph $\mathcal{G}$ is a set $M$ of time-edges of $\mathcal{G}$ such that, for every pair of distinct time-edges $(e,t),(e',t')$ in $M$, we have that $e\cap e'=\emptyset$, or $|t-t'|\geq \Delta$.
\end{defn}

The decision problem we are considering is as follows.

\begin{nolinenumbers}
\decisionproblem{$\Delta$-Temporal Matching}{A temporal graph $\mathcal{G}$ and an integer $h$.}{Does there exist a $\Delta$-temporal matching of cardinality at least $h$?}
\end{nolinenumbers}

Note that $\Delta$ is a fixed constant and not part of the input. We show this to be in FPT with respect to TIM width of the input graph by our algorithm given in Theorem~\ref{thm:component-temp-unif}. Note that we can denote an instance $x$ of $\Delta$\textsc{-Temporal Matching} by $x=(\mathcal{G},\beta)$ where $\beta$ is a string encoding $h$. For clarity, we instead denote an instance as $x=(\mathcal{G},h)$.

We use $(X,1)$-component states, where the label set $X$ consists of the labels $(1,\Delta,\Delta)\cup \{(0,a,b)\,:\, a,b\in [\Delta]\}$, and the vectors of the states contain an integer $m$ which counts the number of time-edges in the connected component $C$ in a partial $\Delta$-temporal matching $M$. Since we require the size of the $\Delta$-temporal matching to be \emph{at least} $h$, $m$ is a negative integer such that $|M\cap C|=|m|$. 
Each integer in the labelling of a vertex $v$ at time $t$ signifies
    \begin{itemize}
        \item whether $v$ is an endpoint of an edge $e$ such that $t\in\lambda(e)$ and the time-edge $(e,t)$ is in the matching $M$;
        \item the difference between $t$ and latest time $t_1<t$ such that $v$ can be an endpoint of an edge $e$ such that $t_1\in\lambda(e)$ and the time-edge $(e,t_1)$ is in the matching $M$;
        \item the difference between $t$ and earliest time $t_2>t$ such that $v$ can be an endpoint of an edge $e$ such that $t_2\in\lambda(e)$ and the time-edge $(e,t_2)$ is in the matching $M$,
    \end{itemize}
    respectively. Informally, the latter two integers describe how many timesteps we need to go back, or forward in time until $v$ can be an endpoint in the matching $M$, respectively.
For example, in a $(X,1)$-component state $s_i$, $l_i(v_1) = (1,\Delta,\Delta)$, and $l_i(v_2)=(0,2,3)$ tell us that there exists a vertex $u_1$ in the bag such that $(v_1u_1,t)$ is in $M$, and for any other vertex $u_2$ in the graph, $(v_1u_2,t')$ cannot be included in $M$ for any times $t'\in (t-\Delta,t+\Delta)$; and for any vertices $u_3$ in the bag $(v_2u_3,t)$ is not in $M$, and $v_2$ can the endpoint of a time-edge included in the $M$ in a bag labelled $t-2$ or earlier, or in a bag labelled with time $t+3$ or later (and at no times in between). We note that, if the first value of $l_i(v)$ is 1 for any vertex $v$, then the other entries must be $\Delta$. We say that a $\Delta$-temporal matching $M$ \emph{respects} the labels given to a vertex set at time $t$ if the vertices labelled $(1,\Delta,\Delta)$ at time $t$ are the endpoints of time-edges in $M$ at time $t$, and no vertex with label $(0,a,b)$ at time $t$ in the set is an endpoint of a time-edge in $M$ with time in the interval $(t-a,t+b)$. 

We will define our transition routine, validity routine, and starting routine such that each state for which its restrictions to the connected components of each snapshot return true for all relevant routines corresponds to the existence of a set $M$ of size $|m|$ of time-edges $(uv,i)$ whose endpoints $u$ and $v$ are both labelled $(1,\Delta,\Delta)$ by the state $s_i$ such that there is no time-edges $(wv,i')$ or $(wu,i')$ in $M$ where $i'-i\leq \Delta$. Our upper bound on the sum of the vectors of the states is $\textbf{v}_{\text{upper}}=(-h)$. Since its absolute value is the sum of all vertices labelled with $(1,\Delta,\Delta)$, this ensures that there are at least $h$ elements in a potential $\Delta$-temporal matching.

Our starting routine returns true if and only if every vertex is labelled with a label in $\{(0,1,b)\,:\, b\in[\Delta]\}$, and each vector is a zero vector in the $(X,1)$-component state. In this case, the finishing routine and validity routine are the same. From this point onwards, we will refer only to the validity routine. The validity routine is defined as an algorithm which returns true for a connected component $C$ if and only if there exists a perfect matching of the vertices labelled $(1,\Delta,\Delta)$ in $C$ and the matching has size $-m$. 
Our transition routine returns true given a connected component $C$ and labellings $l_1$ and $l_2$ for $V(C)$ if and only if all $v$ labelled $(1,\Delta,\Delta)$ by $l_1$ are labelled $(0,1,\Delta-1)$ or $(0,1,\Delta)$ by $l_2$, all $u$ labelled $(1,\Delta,\Delta)$ by $l_2$ are labelled $(0,\Delta-1,1)$ or $(0,\Delta,1)$ by $l_1$, and, for all remaining vertices, $w$ which are labelled $(0,a,b)$ by $l_1$, $l_2(w)=(0,a',b')$ where $a'=\min\{\Delta,a+1\}$ and $b'=\max\{1,b-1\}$. Observe that, if there is a sequence of states $s_0,\ldots,s_t$ such that the transition routine returns true for all connected components in all snapshots of a temporal graph $\mathcal{G}$ and $u$ is labelled $(1,\Delta,\Delta)$ at time $t$, then the labelling of $u$ in all states $s_i\in \{s_{t-\Delta},\ldots,s_{t-1}\}$ must be $(0,a_i,b_i)$, where $a_i$ is at most the difference $i-t'$ of times where $t'$ is the latest time before $i$ that $u$ is an endpoint of a time-edge in $M_t$, and $b_i$ is the value $\Delta-(t-i)$.

We now show that there exists a correspondence between $\Delta$-temporal matchings and sequences of states such that the starting routine returns true for all connected components of the first snapshot, and the validity and transition routines return true for all connected components of all snapshots. We say that a sequence $s_0, ..., s_t$ of $(X,1)$-component states of the form $s_t=(l_t,\textbf{w}^t_1,\ldots,\textbf{w}^t_c, \nu_t)$ \textit{corresponds} to a $\Delta$-temporal matching $M$ of time-edges up to timestep $t$ from $\mathcal{G}$ if and only if:
\begin{enumerate}
    \item for all connected components $C_1$ in $G_1$, $\textbf{St}(s_0|_{C_1}, C_1,x)=\textbf{true}$,
    \item for all times $1\leq i\leq t$ and all connected components $C_{i}$ in $G_{i}$, $\textbf{Val}(s_{i}|_{C_{i}}, C_{i},x)=\textbf{true}$,
    \item for all times $1\leq i\leq t$ and all connected components $C_{i}$ in $G_{i}$, $\textbf{Tr}(l_{i-1}|_C, l_i|_C, C,x)=\textbf{true}$, and 
    \item $M$ respects each labelling $l_i$ in a state $s_i$ for all $1\leq i\leq t$.
\end{enumerate}

\begin{restatable}{lem}{timmatchlem}\label{lem:temp-match-tim}
    For any timestep $t$, a set $M$ consisting of time-edges up to timestep $t$ from $\mathcal{G}$ is a $\Delta$-temporal matching, if and only if there exists a corresponding sequence of $(X,1)$-component states $s_0, \ldots, s_t$.
\end{restatable}
\begin{proof}
    We proceed by induction on the timestep $t$. For time $t=0$, there are no edges active and so any $\Delta$-temporal matching consisting of time-edges with times up to $0$ must be empty. Consider the $\Delta$-temporal matching $M=\emptyset$. Let the state $s_0$ be such that all vertices are given a label in $\{(0,1,b)\,:\, b\in[\Lambda]\}$, and $\nu_0(C)=(0)$ for all connected components $C$ of $G_1$. It is clear from construction that the starting routine returns true for all restrictions of $s_0$ to connected components of $G_1$. Thus $s_0$ corresponds to $M=\emptyset$. Now suppose there is a state $s_0$ such that $\textbf{St}$ returns true for every restriction to a connected component in $G_1$. Then, all vertices must have a label in $\{(0,1,b)\,:\, b\in[\Lambda]\}$ and all vectors must be the zero vector. The empty set then corresponds to $s_0$, since under $l_0$ no vertices are labelled $(1,\Delta,\Delta)$, and is a trivial $\Delta$-temporal matching. Therefore, our base case holds.

    We now assume for all times $t'\leq t$, a set $M$ consisting of time-edges up to timestep $t$ from $\mathcal{G}$ is a $\Delta$-temporal matching if and only if there exists a corresponding sequence of $(X,1)$-component states $s_0, \ldots, s_{t'}$.

    Consider time $t+1$. We first assume $M$ is a $\Delta$-temporal matching of time-edges at or before $t+1$. Let $M_{t}$ be the set of time-edges in $M$ with times up to $t$. By our inductive hypothesis, there exists a corresponding sequence of states $s_0,\ldots,s_t$ for $M_t$.
    
    Now let $s_{t+1}$ be the state labelling all vertices which are an endpoint of an edge $e$ such that $(e,t+1)$ is in $M$ with the label $(1,\Delta,\Delta)$, all vertices which are an endpoint of an edge $e$ such that $(e,t)$ is in $M$ with the label $(0,1,\Delta-1)$, and all remaining vertices $u$ with the label $(0,a',b')$ where $l_t(u)=(0,a,b)$ and $a'=\min\{\Delta,a+1\}$ and $b'=\max\{1,b-1\}$. For each connected component $C$ in $G_{t+1}$, let $\nu_{t+1}(C)=-\frac{1}{2}(|l_{t+1}^{-1}((1,\Delta,\Delta))\cap C|)$. 
    Since there must be two endpoints of each time-edge in a connected component $C$, halving the number of vertices labelled $(1,\Delta,\Delta)$ in $C$ must give the number of time-edges in $M\cap C$. Therefore, $\textbf{Val}(s_{t+1}|_C,C,x)=\textbf{true}$ for all connected components $C$ of $G_{t+1}$. By construction of the labelling $l_{t+1}$, the transition routine must return true for all connected components of $G_{t+1}$. Hence $s_0,\ldots,s_{t+1}$ corresponds to $M$ as required.

    Assume that there exists some sequence of $(X,1)$-component states $s_0, ..., s_{t+1}$ such that, for all connected components $C_1$ in $G_1$, $\textbf{St}(s_0|_{C_1}, C_1,x)=\textbf{true}$; for all times $1\leq i\leq t$ and all connected components $C_{i}$ in $G_{i}$, $\textbf{Val}(s_{i}|_{C_{i}}, C_{i},x)=\textbf{true}$; and for all times $1\leq i\leq t$ and all connected components $C_{i}$ in $G_{i}$, $\textbf{Tr}(l_{i-1}|_C, l_i|_C, C,x)=\textbf{true}$. By induction, there exists a $\Delta$-temporal matching $M_{t}$ corresponding to the sequence $s_0, ..., s_{t}$. Let $M_{t+1}$ be a perfect matching of the set of vertices labelled $(1,\Delta,\Delta)$ by $s_{t+1}$. We know that such a matching exists by the validity routine returning true for all connected components of all snapshots up to time $t+1$. Furthermore, since the transition routine returns true for all connected components in $G_{t+1}$, the set of vertices $S$ labelled $(1,\Delta,\Delta)$ by $s_{t+1}$ must be labelled $(0,\Delta-1,1)$ by $l_t$. Since we have assumed that $M_t$ corresponds to $s_0,\ldots,s_t$, $M_t$ must respect the labelling of these vertices in $s_t$. Therefore, by definition any time-edge in $M_t$ with an endpoint in $S$ cannot be at a time in the interval $[t-\Delta-1,t]$. As a result, the difference between $t+1$ and the time at which any vertex in $S$ is an endpoint in $M_t$ must be at least $\Delta$. Hence, the set $M=M_t\cup M_{t+1}$ must be a $\Delta$-temporal matching. Furthermore, by our earlier observation, since the transition routine returns true for all connected components in all snapshots $G_1,\ldots,G_{t+1}$, the labelling of all vertices $u$ in $M_{t+1}$ in states $s_i\in \{s_{t+1-\Delta},\ldots, s_t\}$ are $(0,a_i,b_i)$, where $a_i$ is at most the difference $i-t'$ of times where $t'$ is the latest time before $i$ that $u$ is an endpoint of a time-edge in $M_t$, and $b_i$ is the value $\Delta-(t+1-i)$. By this reasoning and the inductive hypothesis, there is no vertex $v$ whose labelling under $l_t$ is $(\mu,a,b)$ for which there exists a time-edge $(e,t')$ in $M$ such that $v$ is an endpoint of $e$, $t< t'$, and $t'-t<b$.
   What remains to show is that $M$ respects the labelling of vertices not in $S$.

   Vertices $v$ which are not an endpoint of an edge in $M_{t+1}$ are given a label $(0,a',b')$ such that $a'=\min\{\Delta,a+1\}$ and $b=\max\{1,b-1\}$ where $l_t(v)=(0,a,b)$. Since $M_t$ respects the labelling $l_t$ of all such vertices, and every time-edge in $M$ has time at most $t+1$, we only need to check that, for any vertex $v$, any time-edges in $M$ which are incident to $v$ are at time $t+1-a'$ at the latest. Specifically, we do not need to check time-edges with times later than $t+1$ against $b'$ because there are none in $M$. We note that, since $M_t$ respects the labelling in $s_t$, any edges containing the endpoint $v$ must be at time $t-a$ at the latest. By construction of the labelling, the latest time an edge containing $v$ can be and be in $M$ is $t-a=t-a+1-1=t+1-(a+1)\leq t+1-a'$. Therefore, $M$ respects the labelling of all vertices in the graph, and the sequence $s_0,\ldots,s_{t+1}$ corresponds to $M$.
\end{proof}
\begin{restatable}{theorem}{timmatchthm}
     \textsc{$\Delta$-Temporal Matching} is $(X,1,f)$-component-exchangeable temporally uniform, where $X=(1,\Delta,\Delta)\cup \{(0,a,b)\,:\, a,b\in [\Delta]\}$
     and $f(|C|,x)=\phi^{2.5}$ for every timed connected component $C$ of an input temporal graph with TIM width $\phi$.
\end{restatable}
\begin{proof}
    We begin by showing an instance $(\mathcal{G},\Delta,h)$ of \textsc{$\Delta$-Temporal Matching} is a yes-instance if and only if the criteria of Definition~\ref{def:component-temp-unif} hold.  Recall the criteria of Definition~\ref{def:component-temp-unif}: for each connected component $C_1$ of $G_1$, $\textbf{St}(s_0|_{C_1},C_1,x)=\textbf{true}$; for each connected component $C_{\Lambda}$ of $G_{\Lambda}$, $\textbf{Fin}(s_{\Lambda}|_{C_{\Lambda}},C_{\Lambda},x)=\textbf{true}$; $\textbf{Tr}(l_{t-1}|_{C_t}, l_t|_{C_t}, C_t,x)=\textbf{true}$ where $l_t$ is the labelling of vertices of state $s_t$, for all times $1 \leq t \leq \Lambda$ and connected components $C_t$ of $G_t$; $\textbf{Val}(s_t|_{C_t}, C_t,x)=\textbf{true}$ for all times $1 < t < \Lambda$ and connected components $C_t$ of $G_t$; and the sum of vectors satisfies $\sum_{0\leq t\leq \Lambda} \sum_{C \in\mathcal{C}_t}\nu_{s_t} (C)\leq \textbf{v}_{\text{upper}}$. Recall our earlier discussion where we note that the finishing routine is the same as the validity routine. Therefore, we require that $\textbf{Val}(s_t|_{C_t}, C_t,x)=\textbf{true}$ for all times $1 < t \leq \Lambda$ and connected components $C_t$ of $G_t$.

    By Lemma~\ref{lem:temp-match-tim}, there exists a $\Delta$-temporal matching of time-edges up to timestep $\Lambda$ from $\mathcal{G}$, if and only if there exists a corresponding sequence of $(X,1)$-component states $s_0, ... s_{\Lambda}$. Furthermore, by construction of the validity routine, for each connected component $C$ of a snapshot $G_t$, $\nu_t(C)$ in $s_t$ gives $-\frac{1}{2}|(l_t^{-1}((1,\Delta,\Delta))|$. The absolute value of this number is precisely the number of time-edges in a perfect matching of the vertices in $C\cap l^{-1}_t((1,\Delta,\Delta))$. Therefore, $(\mathcal{G},\Delta,h)$ is a yes-instance if and only if there is such a sequence of $(X,1)$-component states $s_0,\ldots,s_{\Lambda}$ such that the sum of vectors in each state is at most $-h=\textbf{v}_{\text{upper}}$. Our transition routine runs in time $O(|C|)\leq O(\phi)$ by checking the label on each vertex in the connected component $C$ in turn. This is bounded above by the time needed to perform the validity check. For this, we must determine whether there exists a perfect matching of the vertices labelled $(1,\Delta,\Delta)$. This takes $O(|C|^{2.5})\leq O(\phi^{2.5})$ time~\cite{micali_ovv_1980}. Therefore, \textsc{$\Delta$-Temporal Matching} is $(X,1,f)$-component-exchangeable temporally uniform, where $X=(1,\Delta,\Delta)\cup \{(0,a,b)\,:\, a,b\in [\Delta]\}$
     and $f(|C|,x)=\phi^{2.5}$ for every timed connected component $C$.
\end{proof}

In this application of our meta-algorithm, we use vectors with one entry of magnitude at most $\frac{1}{2}|C|\leq \phi$ and $\Delta^2+1$ labels, so we finally obtain our main result by applying \Cref{thm:component-temp-unif}.

\TIMmatching*

% \begin{corollary}
%     \textsc{$\Delta$-Temporal Matching} can be solved in time $O(n\Lambda \phi^2\phi^{2.5}+(\Lambda^2+1)^{6\phi^2}\phi^{6\phi^2 }(\phi^3+\phi^2 \phi^{2.5}))=O(\phi^{4.5}(n\Lambda+\Lambda\phi^{6\phi^2}))$, where $\Lambda$ is the lifetime of the input temporal graph, $n$ the number of vertices, and $\phi$ the TIM width.
% \end{corollary}

\subsection{Singleton Temporal Reachability Edge Deletion}\label{sec:edge-del}

In this section, we use our TIM width meta-algorithm on a problem which asks, given a temporal graph with a source vertex $v_s$ and integers $h$ and $r$, if there is a deletion of a set of time-edges of cardinality at most $h$ such that at most $r$ vertices are temporally reachable from $v_s$ in the resulting temporal graph. We show this problem to be in FPT with respect to TIM width. An optimisation version of this problem has been studied by Enright et al.~\cite{enright_structural_2024}; it is a single-source version of the \textsc{Temporal Reachability Edge Deletion} problem studied by Enright et al.~\cite{enright_deleting_2021} and Molter et al.~\cite{molter_temporal_2021} (which they call \textsc{MinReachDelete}). In that version, they bound the maximum number of vertices reachable from any vertex in the graph rather than a chosen source. 

In this section, we only consider strict temporal paths. We say that a vertex $v$ is \emph{(temporally) reachable} from a vertex $u$ if and only if there is a temporal path from $u$ to $v$. We refer to the temporal reachability of a vertex as the number of vertices reachable from it. We use the convention that a vertex is temporally reachable from itself. 

\begin{nolinenumbers}
\decisionproblem{Singleton Temporal Reachability Edge Deletion (SingReachDelete)}{A temporal graph $\mathcal{G}$, a vertex $v_s \in V(G)$ and positive integers $r$ and $h$.}{Is there a set of time-edges $\mathcal{E}$ of cardinality at most $h$ such that the vertex $v_s$ has temporal reachability at most $r$ after their deletion from $\mathcal{G}$?}
\end{nolinenumbers}

We can denote an instance $(\mathcal{G},r,h)$ of \textsc{SingReachDelete} by $x=(\mathcal{G},\beta)$ where $\beta$ is a string encoding $r$ and $h$. In this problem, we assume without loss of generality that a time-edge $(e,t)$ is only deleted if at least one endpoint is reached from the source $v_s$ before time $t$. Otherwise, the deletion of $(e,t)$ has no impact on the set of vertices reachable from $v_s$. We will use the phrase ``arrive at/before time $t$'' to describe a temporal path whose final time-edge occurs at/before time $t$.

We use $(X,2)$-component states, where the label set $X$ consists of the labels $\reached$, $\current$ and $\unreached$, and the vectors of the states contain two integers $d, r'$ which count the number of time-edges deleted and number of vertices reached from the source in each timed connected component, respectively. To avoid double-counting of vertices reached, the entry $r'$ will count only the vertices labelled $\current$. These can be thought of as the vertices reached from the source by a path that has arrived exactly at the time in question.

% To avoid double-counting of vertices reached, the entry $r'$ will be $0$ for all states with time earlier than $\Lambda$, and count the vertices marked as $\reached$ in the final timestep.

We will define our transition routine, validity routine, and starting routine such that, if there exists a sequence of states whose restrictions to each connected component of the relevant snapshot returns true for all relevant routines, there exists a set $\mathcal{E}'$ of size $d$ of time-edges whose deletion results in only the vertices which are labelled $\reached$ or $\current$ being temporally reachable from the source vertex. Our upper bound on the sum of the vectors of the states is $\textbf{v}_{\text{upper}}=(h, r)$. Since the entry $d$ counts the number of time-edges deleted and $r'$ counts the sum of the number of vertices temporally reachable at time $t$ for all $t$, this ensures that there are at most $h$ deletions of time-edges and the temporal reachability of $v_s$ in the resulting graph is at most $r$.

We say a labelling of vertices at time $t$ is \emph{respected} by a deletion $\mathcal{E}'$ of time-edges if, under the deletion $\mathcal{E}'$, the source temporally reaches only the vertices labelled $\reached$ at time $t'<t$ and the vertices labelled $\current$ by a path that arrives at time $t$ in the resulting temporal graph.

If the connected component in question contains the source, our starting routine returns true if and only if the source is labelled $\current$ and every other vertex is labelled with $\unreached$, and the vector is $(0,1)$. Otherwise, the starting routine returns true if and only if all vertices are labelled $\unreached$, and the vector is $(0,0)$. To determine the validity of a restriction of a $(X,2)$-component state to a component, we perform a validity check on the labelling. The purpose of this is to determine the minimum number of edges that must be deleted to ensure that any vertices labelled $\unreached$ are not temporally reachable from the source. This subroutine is defined in Algorithm~\ref{alg:temp-edge-tim-labels}.

\begin{algorithm}\caption{\textsc{Temporal Reachability Edge Deletion Label Validity Routine}}\label{alg:temp-edge-tim-labels}
    \begin{algorithmic}[1]
    \Require A connected component $C$, labellings $l$ of $V(C)$, and input instance $x$.
    \Ensure Returns the cardinality $d$ of a minimum deletion of edges in $E(C)$ under which no vertex labelled $\unreached$ is adjacent to a vertex labelled $\reached$ under $l$.
    \State Initialise $d=0$.
    \State{Let $U$ be the set of all vertices labelled $\unreached$ by $l$, and $R$ the set of vertices labelled reached.}
    \For{all vertices $v$ in $R$}
        \State $d=d + |N(v)\cap U|$.
    \EndFor
    \State \textbf{return} $d$.
    \end{algorithmic}
\end{algorithm}
% \resetlinenumber[2224]
In this case, the finishing routine and validity routine are the same. From this point onwards, we will refer only to the validity routine. The validity routine is defined as an algorithm which returns true for a restriction of a state $(l|_C,(d,r'))$ to connected component $C$ if and only if the vector in the state for $C$ is $(d,r')$ where $d$ is the output of Algorithm~\ref{alg:temp-edge-tim-labels} when run with inputs $C$ and $l|_C$, and $r'$ is the number of vertices labelled $\current$ in $C$. Algorithm~\ref{alg:temp-edge-tim} gives our transition routine.

\begin{algorithm}\caption{\textsc{Temporal Reachability Edge Deletion Transition Routine}}\label{alg:temp-edge-tim}
    \begin{algorithmic}[1]
    \Require A connected component $C$, labellings $l_1$ and $l_2$ for $V(C)$, and input instance $x$.
    \Ensure Returns true when there exists a time-edge deletion which respects $l_1$ if and only if there exists a time-edge deletion which respects $l_2$ and false otherwise.
    \State{Let $R_1$ be the set of all vertices labelled $\reached$ by $l_1$, and $R_2$ be the set of all vertices labelled $\reached$ by $l_2$. Similarly, let $U_1$ be the set of all vertices labelled $\unreached$ by $l_1$, and $U_2$ be the set of all vertices labelled $\unreached$ by $l_2$, and let $N_1$ be the set of all vertices labelled $\current$ by $l_1$, and $N_2$ be the set of all vertices labelled $\current$ by $l_2$.}
    \If{$R_2=R_1\cup N_1$}
        \If{$N_2=U_1\cap N_C(R_2)$}
        \State{\Return{True}}
        \Else
        \State{\Return{False}}
        \EndIf
    \Else
        \State{\Return{False}}
    \EndIf
    \end{algorithmic}
\end{algorithm}
% \resetlinenumber[2245] 

We now show that there exists a correspondence between time-edge deletions and sequences of states such that the starting routine returns true for all connected components of the first snapshot, and the validity and transition routines return true for all connected components of all snapshots. We say that a sequence $s_0, ..., s_t$ of $(X,2)$-component states of the form $s_t=(l_t,\textbf{w}^t_1,\ldots,\textbf{w}^t_c, \nu_t)$ \textit{corresponds} to a time-edge deletion $\mathcal{E}'$ of time-edges up to timestep $t$ from $\mathcal{G}$ if and only if:
\begin{enumerate}
    \item for all connected components $C_1$ in $G_1$, $\textbf{St}(s_0|_{C_1}, C_1,x)=\textbf{true}$,
    \item for all times $1\leq i\leq t$ and all connected components $C_{i}$ in $G_{i}$, $\textbf{Val}(s_{i}|_{C_{i}}, C_{i},x)=\textbf{true}$,
    \item for all times $1\leq i\leq t$ and all connected components $C_{i}$ in $G_{i}$, $\textbf{Tr}(l_{i-1}|_C, l_i|_C, C,x)=\textbf{true}$, 
    \item $\mathcal{E}'$ respects each labelling $l_i$ in a state $s_i$ for all $1\leq i\leq t$, and
    \item for each connected component $C$ of a snapshot $G_i$ of $\mathcal{G}$, the number of time-edges $(e,i)$ in $\mathcal{E}'$ such that $e\in E(C)$ is $d$ where $\nu_i(C)=(d,r')$ for some $r'$.
\end{enumerate}
Recall that $x$ is the input instance of the problem.

\begin{restatable}{lemma}{timedgelem}\label{lem:temp-edge-tim}
    For any timestep $t$, there exists a deletion of time-edges $\mathcal{E}'$ consisting of time-edges up to timestep $t$ from $\mathcal{G}$ such that only a set of vertices $R$ is reached from the source by time $t$, if and only if there exists a corresponding sequence of $(X,2)$-component states $s_0, \ldots, s_t$ for $t<\Lambda$.
\end{restatable}
\begin{proof}
    We proceed by induction on the timestep $t$. For time $t=0$, there are no edges active and so only the source is reachable from the source, and any deletion $\mathcal{E}'$ must be empty. Consider the deletion $\mathcal{E}'=\emptyset$, following which only $v_s$ is temporally reachable from $v_s$ by time 0. Let the state $s_0$ be the $(X,2)$-component state which gives label $\current$ to $v_s$ and $\unreached$ to all other vertices, and sets $\nu_0(C)$ to $(0,1)$ for the connected component $C$ containing $v_s$, and $\nu_0(C')=(0,0)$ for all other connected components $C'$ of $G_1$. It is clear from construction that the starting routine returns true for all restrictions of $s_0$ to connected components of $G_1$, and $s_0$ corresponds to $\mathcal{E}'$. Now suppose there is a state $s_0$ such that $\textbf{St}$ returns true for every restriction to a connected component in $G_1$. Then, the source must be labelled $\current$ and all remaining vertices must be labelled $\unreached$ under $s_0$ and the vector of the connected component containing $v_s$ must be $(0,1)$ and all other vectors in $s_0$ must be $(0,0)$. The set $D=\emptyset$ must correspond to $s_0$. Thus, the base case holds.

    We now assume for all times $t'\leq t<\Lambda-1$, there exists a deletion of time-edges $\mathcal{E}'$ consisting of time-edges up to timestep $t$ from $\mathcal{G}$ such that only a set of vertices $R$ is reached from the source by time $t$, if and only if there exists a corresponding sequence of $(X,2)$-component states $s_0, \ldots, s_{t'}$.

    Consider time $t+1$. We first assume that there exists a deletion $\mathcal{E}'$ of time-edges at or before $t+1$ such that only the vertices in $R$ are temporally reachable from $v_s$ by time $t+1$. Let $\mathcal{E}'_{t}$ be the set of time-edges in $\mathcal{E}'$ with times up to $t$ and $\mathcal{E}'_{t+1}$ be the set of time-edges in $\mathcal{E}'$ with time $t+1$. By our inductive hypothesis, there exists a corresponding sequence of states $s_0,\ldots,s_t$ for $\mathcal{E}'_t$.
    
    Now let $s_{t+1}$ be the state containing the labelling $l_{t+1}$ where all vertices which are temporally reachable from $v_s$ by paths which arrive strictly before $t+1$ are labelled $\reached$, any vertices which temporally reachable from $v_s$ by paths that arrive at $t+1$ are labelled $\current$, and all other vertices are labelled $\unreached$. These must be exactly the vertices in $R$.
    For each connected component $C$ in $G_{t+1}$, let $\nu_{t+1}(C)=(d,0)$ where $d$ is the number of edges in $\mathcal{E}'_{t+1}\cap E(C)$. 
    
    We claim that Algorithm~\ref{alg:temp-edge-tim-labels} returns the same value $d$ given a labelled connected component $C$. The algorithm increments $d$ by the number of pairs of vertices where one is labelled $\reached$ and the other is a neighbour of the first labelled with $\unreached$. We show that the number of such pairs must be the number of edges in $\mathcal{E}'_{t+1}\cap E(C)$.
    Suppose there is an edge $e\in E(C)$ such that one endpoint $v$ is labelled $\reached$, and the other $u$ is labelled $\unreached$. Then, if $e$ is not deleted, then $u$ must be reachable from $v_s$ at time $t+1$ by appending $(e,t+1)$ to the path by which $v$ is reachable from $v_s$. Therefore, $u$ would be labelled $\current$ by construction of our state; a contradiction. This gives us that $d\leq \mathcal{E}'_{t+1}\cap E(C)$. Equality follows from our assumption that all time-edges that are deleted have at least one endpoint which is reached before time $t+1$. This then gives us that the validity routine returns true for all connected components of $G_{t+1}$.
    
    We now consider the transition routine given by Algorithm~\ref{alg:temp-edge-tim}. By construction of our labelling of the vertices of each connected component $C$ in $G_{t+1}$, the vertices labelled $\reached$ in $s_{t+1}$ must be the union of those labelled $\reached$ and those labelled $\current$ in $s_t$. In addition, all vertices are labelled $\current$ if and only if they are labelled $\unreached$ by $s_t$ and adjacent to a vertex labelled $\reached$ by $s_{t+1}$ at time $t+1$. Therefore, Algorithm~\ref{alg:temp-edge-tim} returns true in line 4, and the transition routine returns true for all connected components $C$ of $G_{t+1}$.
    Therefore, $s_0,\ldots,s_{t+1}$ corresponds to $\mathcal{E}'$ as required.

    Assume that there exists some sequence of $(X,2)$-component states $s_0, ..., s_{t+1}$ such that, for all connected components $C_1$ in $G_1$, $\textbf{St}(s_0|_{C_1}, C_1,x)=\textbf{true}$; for all times $1\leq i\leq t$ and all connected components $C_{i}$ in $G_{i}$, $\textbf{Val}(s_{i}|_{C_{i}}, C_{i},x)=\textbf{true}$; and for all times $1\leq i\leq t$ and all connected components $C_{i}$ in $G_{i}$, $\textbf{Tr}(l_{i-1}|_C, l_i|_C, C,x)=\textbf{true}$. By induction, there exists a time-edge deletion $\mathcal{E}'_{t}$ corresponding to the sequence $s_0, ..., s_{t}$. Let $\mathcal{E}'_{t+1}$ be the set of time-edges $(e,t+1)$ such that $e\in E(G_{t+1})$, one endpoint of $e$ is labelled $\reached$ by $s_t$, and the other is labelled $\unreached$. Let $\mathcal{E}'=\mathcal{E}'_t\cup \mathcal{E}'_{t+1}$. We now show that $\mathcal{E}'$ respects $s_{t+1}$. To show this, we need to show that $v_s$ temporally reaches only the vertices labelled $\reached$ at time $t'<t+1$, and only the vertices labelled $\current$ are reached from the source at time $t+1$ following the deletion of $\mathcal{E}'$. Since we have assumed that the transition routine returns true for all connected components $C$ of $G_{t+1}$, the vertices labelled $\reached$ by $s_{t+1}$ must be those labelled either reached or $\current$ by $s_t$. By the inductive hypothesis, $\mathcal{E}'_t$ respects the labelling in $s_t$. Also note that the deletion of any time-edges in $\mathcal{E}'\setminus\mathcal{E}'_{t}$ cannot affect the times at which these vertices are first reached. Therefore, the vertices labelled $\reached$ must be precisely those reached from the source before time $t+1$ under the deletion $\mathcal{E}'$. The vertices labelled $\current$ are precisely those which are labelled $\unreached$ by $s_t$ and adjacent to a vertex labelled $\reached$ in $G_{t+1}\setminus \mathcal{E}'_{t+1}$. Therefore, these vertices must be temporally reached from the source at time $t+1$.

    What remains to check is that, for all connected components $C$ of $G_{t+1}$, the number of time-edges in $E(C)\cap \mathcal{E}'_{t+1}$ is $d$ where $(d,r')=\nu_{t+1}(C)$. Since the validity routine returns true for all such connected components, $d$ must be the output of Algorithm~\ref{alg:temp-edge-tim-labels} when run with input $C$ and labelling $l_{t+1}|_C$ where $l_{t+1}$ is the labelling in $s_{t+1}$. As discussed earlier, the output of the algorithm is exactly the set of edges in $E(C)$ such that one endpoint is labelled $\reached$ and the other is labelled $\unreached$. By construction, this is the set of edges $\mathcal{E}'_{t+1}$. Therefore, the sequence $s_0,\ldots,s_{t+1}$ corresponds to $\mathcal{E}'$, and $\mathcal{E}'$ is a deletion such that only the vertices labelled $\reached$ or $\current$ are reached from the source.
\end{proof}
\begin{restatable}{theorem}{timedgethm}
    \textsc{SingReachDelete} is $(X,2,f)$-component-exchangeable temporally uniform, where $X=\{\reached,\unreached,\current\}$, and $f(|C|,x)=\phi$ for every timed connected component $C$ of an input temporal graph with TIM width $\phi$.
\end{restatable}
\begin{proof}
    We begin by showing an instance $(\mathcal{G},r,h)$ of \textsc{SingReachDelete} is a yes-instance if and only if the criteria of Definition~\ref{def:component-temp-unif} hold.  That is, for each connected component $C_1$ of $G_1$, $\textbf{St}(s_0|_{C_1},C_1,x)=\textbf{true}$; for each connected component $C_{\Lambda}$ of $G_{\Lambda}$, $\textbf{Fin}(s_{\Lambda}|_{C_{\Lambda}},C_{\Lambda},x)=\textbf{true}$; $\textbf{Tr}(l_{t-1}|_{C_t}, l_t|_{C_t}, C_t,x)=\textbf{true}$ where $l_t$ is the labelling of vertices of state $s_t$, for all times $1 \leq t \leq \Lambda$ and connected components $C_t$ of $G_t$; $\textbf{Val}(s_t|_{C_t}, C_t,x)=\textbf{true}$ for all times $1 < t < \Lambda$ and connected components $C_t$ of $G_t$; and the sum of vectors satisfies $\sum_{0\leq t\leq \Lambda} \sum_{C \in\mathcal{C}_t}\nu_{s_t} (C)\leq \textbf{v}_{\text{upper}}$. Recall our earlier discussion where we note that the finishing routine is the same as the validity routine. Therefore, we require that $\textbf{Val}(s_t|_{C_t}, C_t,x)=\textbf{true}$ for all times $1 < t \leq \Lambda$ and connected components $C_t$ of $G_t$.

    By Lemma~\ref{lem:temp-edge-tim}, there exists a deletion of time-edges up to timestep $\Lambda$ from $\mathcal{G}$, if and only if there exists a corresponding sequence of $(X,2)$-component states $s_0, ... s_{\Lambda}$. Furthermore, for each connected component $C$ of a snapshot $G_t$, $\nu_t(C)$ in $s_t$ is $(d,r')$ where $d$ is the number of time-edges that must be removed in $C$ and $r'$ is the number of vertices in $C$ newly reachable from $v_s$. Therefore, $(\mathcal{G},r,h)$ is a yes-instance if and only if there is such a sequence of $(X,2)$-component states $s_0,\ldots,s_{\Lambda}$ such that the sum of vectors in each state is at most $(h,r)$. Note that the subroutines run in time at most linear in $\phi$. Therefore, \textsc{SingReachDelete} is $(X,2,f)$-component-exchangeable temporally uniform, where $f(|C|,x)=\phi$ for every timed connected component $C$ of $\mathcal{G}$.
\end{proof}
 We use vectors with two entries of magnitude at most $\phi^2$ and $3$ labels, we finally obtain our main result by applying \Cref{thm:component-temp-unif}.

\TIMdeletion*

\end{toappendix}
\section{Future directions}\label{sec:conclusion}

% In this work, we define the new parameter TIM width, a generalisation of VIM width, and give meta-algorithms with respect to both parameters. Furthermore, we provide an exact characterisation of problems in FPT with respect to these parameters. We explore the relationships between these parameters and other known temporal parameters. 
% In this work we have applied our meta-algorithms to several problems. 
Resolving the complexity of building a TIM decomposition with smallest width (or within a constant factor) is our first future goal.  
Beyond that, the clearest extension of this work is further application of our meta-algorithms; we expect that many other temporal problems will exhibit the particular properties required. Another potential avenue would explore the relationship between VIM and TIM width and other parameters not discussed here.
% , such as temporal analogues of feedback vertex number, vertex cover, and treedepth. 
% We also leave open whether there is a problem which behaves differently when parameterised by TIM width and bidirectional connected-VIM width.
Whether there is a problem which behaves differently when parameterised by TIM width and bidirectional connected-VIM width is also open.

The existing interval-membership parameters mentioned earlier all have analogues for the edges (rather than vertices) in a temporal graph. We believe that this toolkit will carry over to edge-interval-membership width (an edge analogue of VIM width). If we label the edges, rather than the vertices, of a temporal graph, and allow the label on an edge to change only within its active interval, then we would expect a similar meta-algorithm to be obtainable using analogous concepts to those introduced here. While it is not clear what the appropriate edge analogue of TIM width should be, obtaining such an analogue is likely to have interesting algorithmic consequences.  %We suggest as future work to define edge-locally temporally uniform problems, and prove an equivalent theorem to \Cref{thm:loctempunf}. 
%We believe that \textsc{Temporal Eulerian Walk} and \textsc{Temporal Vertex Cover} would be examples of such problems, as these are analogous to \textsc{Temporal Hamiltonian Path} and \textsc{Temporal Dominating Set} respectively, except that a solution consists of a set of edges, rather than vertices. Note that, for example in \textsc{Temporal Eulerian Walk}, we may want to distinguish the endpoints of edges. Specifically, for path-based problems we need to know which endpoint is the end of the path. Therefore, it seems that the edge version of our meta-algorithm for VIM width would require a little more care when applied.

%Another potential formulation for TIM width is to allow two bags labelled with the same time to have an arc between them. To ensure the resulting decomposition graph is a tree, there cannot be edges between bags with the same time label if there exist vertices in those bags which are in the same bag elsewhere in the decomposition. This means that this formulation would only really have an advantage of TIM width as defined for graphs which have large connected components at one time. Given such a temporal graph, it may be interesting to treat this snapshot like a static graph problem and combine a solution with the rest of the graph. 

% \section{Acknowledgements}\label{sec:ack}
% For the purpose of open access, the author(s) has applied a Creative Commons Attribution (CC BY) license to any Author Accepted Manuscript version arising from this submission.

\bibliography{used-refs}
 % \appendix

\end{document}